\documentclass[hidelinks,onefignum,onetabnum]{siamart251216}

\usepackage{booktabs}

\usepackage{lipsum}
\usepackage{amsfonts}
\usepackage{graphicx}
\usepackage{epstopdf}
\usepackage{algorithmic}
\ifpdf
  \DeclareGraphicsExtensions{.eps,.pdf,.png,.jpg}
\else
  \DeclareGraphicsExtensions{.eps}
\fi

\newsiamremark{remark}{Remark}
\crefname{remark}{Remark}{Remarks}
\newsiamremark{hypothesis}{Hypothesis}
\crefname{hypothesis}{Hypothesis}{Hypotheses}
\newsiamthm{claim}{Claim}
\newsiamremark{fact}{Fact}
\crefname{fact}{Fact}{Facts}
\newsiamremark{assumption}{Assumption}
\crefname{assumption}{Assumption}{Assumptions}

\headers{Policy Gradient for Log-Growth Control}{Q.~Pan, Y.~Shen, L.~Zhang, C.~Chen and X.~Guan}

\title{Sample Complexity of Policy Gradient for Log-Growth Control}

\author{Qiuhua Pan\thanks{%
State Key Laboratory of Submarine Geoscience, School of Automation and
Intelligent Sensing, Shanghai Jiao Tong University, Shanghai 200240, China;
Key Laboratory of System Control and Information Processing, Ministry of
Education of China, Shanghai 200240, China; and
Shanghai Key Laboratory of Perception and Control in Industrial Network
Systems, Shanghai 200240, China
(\email{panqh@sjtu.edu.cn},
 \email{shen010704@sjtu.edu.cn},
 \email{cailianchen@sjtu.edu.cn},
 \email{xpguan@sjtu.edu.cn}).
Corresponding author: Xinping Guan.}
\and Yukai Shen\footnotemark[2]
\and Liwei Zhang\thanks{Paris Elite Institute of Technology, Shanghai Jiao
Tong University, Shanghai 200240, China
(\email{zhangliwei01@sjtu.edu.cn}).}
\and Cailian Chen\footnotemark[2]
\and Xinping Guan\footnotemark[2]}

\usepackage{amsopn}

\ifpdf
\hypersetup{
  pdftitle={Sample Complexity of Policy Gradient for Log-Growth Control},
  pdfauthor={Qiuhua Pan, Yukai Shen, Liwei Zhang, Cailian Chen, Xinping Guan}
}
\fi

\begin{document}

\maketitle

\begin{abstract}
We study the sample complexity of policy gradient for log-growth control---the problem of learning, from observed state transitions, a feedback gain that optimally stabilizes a scalar linear system driven through a multiplicative-noise actuation channel. The objective $J(K) = \mathbb{E}[\log|1+BK|]$ is the top Lyapunov exponent of the closed loop. This problem carries a structural difficulty we call the \emph{cusp obstruction}: the optimal gain $K^*$ always places the noise singularity $b_{\rm sing}(K) = -1/K$ in the interior of the support. At this singular optimum the policy gradient exists only as a Cauchy principal value, not as a Lebesgue integral, and the natural single-sample gradient estimator has infinite variance. Standard first-order stochastic-optimization analysis is thus inapplicable at the optimum, and merely smoothing the objective does not resolve the difficulty. The obstruction, however, has an exploitable symmetry: the Cauchy kernel is an odd function of the displacement from the moving pole, so pairing each observation with its reflection through the pole cancels the divergent part. This one cancellation simultaneously controls the population curvature, the gradient-estimator variance, and the bias incurred when the noise density is estimated. Combining these bounds with a closed-form single-transition gradient oracle, we prove that projected mini-batch policy gradient, initialized in any compact subset of the stabilizing region, attains total sample complexity $\tilde O(1/\eta)$ when the noise density is known and $\tilde O(\eta^{-(2s+1)/(2s)})$ when it must be estimated, for $C^s$ noise densities with $s \ge 2$.
\end{abstract}



\section{Introduction}
\subsection{The problem and the core difficulty}\label{11-problem-and-motivation}
Consider a scalar linear plant whose control enters through a noisy
actuation channel,
\begin{equation}\label{eq:plant}
    X_{t+1} = a\,(X_t + B_t U_t), \qquad t \ge 0,
\end{equation}
where \(|a| \ge 1\) is the open-loop growth factor, \(U_t\) is the control
applied at step \(t\), and \((B_t)_{t \ge 0}\) is an independent and
identically distributed (i.i.d.)\ sequence from a density \(\rho\) supported
on \([b_{\min}, b_{\max}] \subset (0, \infty)\). The control is not delivered
directly: it passes through an actuation channel whose gain \(B_t\)
fluctuates stochastically from step to step --- the multiplicative-noise
model of Ranade and Sahai's control-capacity
framework~\cite[sect.~II.A]{ranade2018control}. We restrict to static linear
state feedback \(U_t = K X_t\) with gain \(K \in \mathbb{R}\), which is
without loss of optimality for the cost considered
below~\cite[Thm.~3]{ranade2018control}; the resulting closed loop is
\begin{equation}\label{eq:sys}
    X_{t+1} = a\,(1 + B_t K)\,X_t, \qquad t \ge 0.
\end{equation}

\paragraph{The cost is the closed-loop contraction rate}
Iterating \eqref{eq:sys} gives $X_t = a^t X_0 \prod_{\tau=0}^{t-1}(1 + B_\tau K)$,
so the strong law of large numbers yields the almost-sure exponential
growth rate
\[
    \lim_{t\to\infty} \tfrac{1}{t}\log|X_t|
    \;=\; \log|a| + \mathbb{E}[\log|1+BK|]
    \qquad \text{a.s.}
\]
Normalizing to \(|a| = 1\) by the logarithmic shift \(J \mapsto J + \log|a|\),
the infinite-horizon objective is the \emph{log-growth cost}
\begin{equation}\label{eq:cost}
    J(K) := \mathbb{E}[\log|1+BK|],
\end{equation}
the top Lyapunov exponent of the random product \(\prod_{t\ge 0}(1+B_t K)\).
Thus \(J(K)\) is exactly the almost-sure geometric rate at which the closed-loop
state contracts (\(J(K)<0\)) or diverges (\(J(K)>0\)): the stabilizing region is
\(\mathcal{K}_{\rm stab} := \{K : J(K) < 0\}\), and minimizing \(J\) over
\(\mathcal{K}_{\rm stab}\) selects the gain \(K^* := \arg\min_K J(K)\) achieving
the steepest contraction. Optimizing \(J\) is therefore the natural
infinite-horizon stabilization problem for \eqref{eq:sys}, the multiplicative-noise
analogue of minimizing the linear-quadratic regulator (LQR) cost for an additive-noise system.

\paragraph{What the learner observes}
The system parameter \(a\) is known and the gain \(K\) is chosen by the learner,
so every state transition \((X_t, X_{t+1})\) with \(X_t \ne 0\) reveals the
realized channel gain exactly, via
\(B_t = \bigl(X_{t+1}/(aX_t) - 1\bigr)/K\). Observing the closed-loop trajectory
thus yields an i.i.d.\ sample stream \((B_t)_{t\ge 0}\) from \(\rho\), one sample
per transition; this is the operational meaning of ``observing a transition''
throughout the paper. We study two access models: \emph{density-known}, in which
\(\rho\) is additionally available in closed form (e.g.\ the actuation channel
has been characterized offline), and \emph{density-unknown}, in which \(\rho\) is
unknown and must be learned from the observed stream \((B_t)\).

We ask whether, and at what sample cost, the optimal stabilizing gain can be learned from observed transitions \((X_t, X_{t+1})\) alone: whether a policy-gradient algorithm returns a gain \(\hat K\) with \(\mathbb{E}[J(\hat K) - J^*] \le \eta\), and how many transitions this requires. What sets the question apart is where its difficulty sits. The policy gradient is
\begin{equation}\label{eq:grad}
    \frac{dJ}{dK}(K) = \int_{b_{\min}}^{b_{\max}} \frac{b\,\rho(b)}{1 + bK}\,db,
\end{equation}
and a stationary gain satisfies \(dJ/dK(K) = 0\) in the principal-value sense. Since \(b\rho(b) > 0\) throughout the support, the kernel \(1/(1+bK)\) can integrate to zero only by changing sign, which drives the zero of \(1+bK\) --- the singularity \(b_{\rm sing}(K) = -1/K\) of the integrand --- into the interior of the support. Optimality therefore \emph{forces} the singularity inside the noise support (\cref{lem:strict-cusp-optimum}): the optimum \(K^*\) does not merely lie near a singular point; it is one.

At \(K^*\) the integrand of \eqref{eq:grad} has a non-integrable simple pole, so \(dJ/dK\) is no longer an ordinary integral --- it exists only as a Cauchy principal value (\cref{prop:non-lebesgue-gradient}). This is the \emph{cusp obstruction}, and it places the problem beyond the reach of standard policy-gradient analysis: the descent lemma, linear convergence under a Polyak--{\L}ojasiewicz (PL) inequality, and the variance control of mini-batch stochastic gradient descent (SGD) all presuppose a locally Lipschitz gradient --- a hypothesis that fails at the cusp, hence at the very point the iterates must reach. The obstruction has a statistical face as well: the same pole renders the natural single-transition gradient estimator non-square-integrable, so its variance at \(K^*\) is infinite.

Smoothing is the natural response, and by itself it does not suffice. Replacing \(J\) by a regularized family \(J_\varepsilon\) removes the pole for every \(\varepsilon > 0\), but reaching accuracy \(\eta\) forces \(\varepsilon \downarrow 0\); along that limit the variance of the smoothed single-transition estimator grows as \(\Theta(1/\varepsilon)\) (\cref{thm:variance-scaling-naive-estimator}), while the curvature that drives linear convergence may collapse. Naive regularization thus exchanges a non-integrable gradient for an estimator of diverging variance --- it relocates the obstruction rather than removing it.

\subsection{State of the art and its limits}\label{12-state-of-the-art-and-its-limits}

Policy-gradient (PG) sample complexity for linear systems originates with Fazel et al.\,\cite{fazel2018global} (LQR, gradient domination), with follow-up work on the stabilizing-gain geometry \cite{bu2019lqr}, the model-based/model-free gap \cite{tu2019gap}, model-free LQR constants \cite{mohammadi2021convergence}, derivative-free rates \cite{malik2020derivative}, and noisy-LQR PG over a finite horizon \cite{hambly2021policy}. For multiplicative-noise LQR, Gravell, Mohajerin Esfahani, and Summers \cite{gravell2020learning} establish $\tilde O(1/\eta^2)$ via an algebraic Lyapunov-equation almost-smoothness bound. A parallel continuous-time line of work places entropy regularization at the level of the control formulation (Wang--Zariphopoulou--Zhou \cite{wang2020reinforcement}; Giegrich--Reisinger--Zhang \cite{giegrich2024convergence}), obtaining geometry-aware PG iterations with linear convergence for relaxed exploratory LQ; entropy there is a \emph{modeling} device for exploration on top of an a priori Lebesgue-smooth gradient, whereas our Cauchy regularization plays an \emph{analytical} role---repairing a cost whose formal gradient is non-Lebesgue (\cref{prop:non-lebesgue-gradient}). All of these analyses require global gradient-Lipschitz continuity together with a Lyapunov-equation identity for $dJ/dK$; the cusp obstruction of \cref{11-problem-and-motivation} rules out both. The same obstruction excludes subgradient/proximal frameworks for non-smooth optimization: weakly convex objectives (Davis--Drusvyatskiy \cite{davis2019stochastic}) and tame functions (Bolte--Pauwels \cite{bolte2021conservative}) deliver convergence to critical points but not explicit $\eta$-rates, and $J$ is not weakly convex on $\mathcal{K}_{\rm stab}$ (its subgradient correspondence is unbounded near $K^*$).

\paragraph{Position in the risk-sensitive family}
A second view, which we exploit analytically, is that $J$ sits at the boundary of the risk-sensitive family $\beta^\theta(K) := \theta^{-1}\log\mathbb{E}[|1+BK|^\theta]$---the exponential-of-cost criterion of Jacobson~\cite{jacobson1973optimal} and Whittle~\cite{whittle1990risk}, in the policy-optimization form studied by Borkar--Meyn~\cite{borkar2002risk} and Anantharam--Borkar~\cite{anantharam2017variational}. For every $\theta > 0$, $\beta^\theta$ is Lebesgue-smooth in $K$ with gradient $\mathbb{E}[|1+BK|^\theta\,b/(1+BK)]/\mathbb{E}[|1+BK|^\theta]$, since the $|1+BK|^\theta$ factor cancels the $1/(1+BK)$ singularity. The limit $\beta^\theta(K) \to J(K)$ as $\theta \downarrow 0$ is precisely the degeneration of this Lebesgue smoothness; the cusp obstruction is the structural boundary phenomenon of the $\theta$-family at $\theta = 0$. The present analysis thus occupies the missing $\theta = 0$ corner of risk-sensitive policy-gradient theory, complementary to the $\theta > 0$ analyses of Borkar \cite{borkar2001sensitivity}.

A natural alternative under density-unknown access is certainty equivalence~\cite{mania2019certainty}: estimate $\hat\rho$ and solve $\mathrm{PV}\!\int b\hat\rho(b)/(1+b\hat K)\,db = 0$. \Cref{prop:plug-and-solve-rate} establishes rate parity with our policy-gradient \cref{alg:density-unknown-pg}; the contribution of this paper is therefore a sample-complexity analysis of \emph{policy gradient} --- the canonical first-order stochastic procedure in adaptive control --- in a regime where its standard tools fail, together with the parity-cancellation mechanism that delivers it (section~\ref{56-comparison-and-structural-perspective} discusses when each method is preferable).

\subsection{Contributions}\label{13-contributions}

The resolution of the cusp obstruction rests on a single observation: the singularity is odd. The Cauchy kernel is an odd function about the moving pole $b_{\rm sing}(K)$, and a parity-matched symmetric construction cancels its odd part. The same cancellation resolves the difficulty at the three levels at which it arises.

\emph{(i) Population curvature.}\quad The Cauchy-regularized cost $J_\varepsilon$ admits a closed-form Hessian decomposition (\cref{thm:hessian-decomposition}) whose near-pole leading term vanishes by parity, yielding a uniform-in-$\varepsilon$ PL constant on the local basin $\mathcal{N}_\delta$ (\cref{lem:uniform-PL-constant,thm:uniform-epsilon-pl-inequality}).

\emph{(ii) Estimator variance.}\quad A density-aware symmetric-pairing estimator (\cref{def:symmetric-pairing-estimator,prop:unbiased-variance-psi-tilde}) is unbiased and has $O(1)$ variance uniformly in $\varepsilon$, in contrast to the $\Theta(1/\varepsilon)$ divergence of the naive estimator (\cref{thm:variance-scaling-naive-estimator}).

\emph{(iii) Plug-in bias.}\quad When $\rho$ is replaced by a kernel density estimate (KDE), the weight discrepancy $\Delta := w_{\hat\rho} - w_\rho$ is exactly odd through $b_{\rm sing}(K)$ (\cref{lem:parity-weight-discrepancy}), reducing the plug-in bias to $O(\nu'_{n_1} R)$ (\cref{prop:uniform-bias-bound}).

Together with the closed-form single-transition gradient oracle afforded by the multiplicative-ergodic structure, these bounds establish that first-order learning at the singular optimum is possible: projected mini-batch policy gradient, initialized in any compact subset of $\mathcal{K}_{\rm stab}$, attains total sample complexity $\tilde O(1/\eta)$ when $\rho$ is known to the algorithm (\cref{thm:sample-complexity-density-known}) and $\tilde O(\eta^{-(2s+1)/(2s)})$ when $\rho$ must be estimated, for $C^s$ densities with $s \ge 2$ (\cref{thm:total-sample-complexity-pl-basin}); a certainty-equivalent principal-value root-finder attains the latter rate as well (\cref{prop:plug-and-solve-rate}). \Cref{tab:complexity_bounds_matched} collects these bounds.

\begin{table}[htbp]
\centering
\footnotesize
\setlength{\tabcolsep}{4pt}
\caption{Sample complexity and convergence rates for the scalar log-growth control cost $J(K)=\mathbb{E}[\log|1+BK|]$. Initialization in any compact subset of $\mathcal{K}_{\rm stab}$ is admissible via the preliminary phase of section~\ref{56-comparison-and-structural-perspective} (full statement in section~SM2 of the supplement) at $\eta$-independent extra cost.}
\label{tab:complexity_bounds_matched}
\begin{tabular}{@{}llll@{}}
\toprule
Access model & Algorithm & Rate & Reference \\ \midrule
Density-known & PG with oracle pairing & $\tilde O(1/\eta)$ & \cref{thm:sample-complexity-density-known} \\
Density-known & Newton on PV equation & $O(\log\log\tfrac{1}{\eta})$ iters$^\dagger$ & \cite{davis2007methods,atkinson2008introduction} \\
Density-unknown ($\rho \in C^s$, $s\ge 2$) & PG with KDE plug-in & $\tilde O(\eta^{-(2s+1)/(2s)})$ & \cref{thm:total-sample-complexity-pl-basin} \\
Density-unknown ($\rho \in C^s$, $s\ge 2$) & Plug-and-solve PV root & $\tilde O(\eta^{-(2s+1)/(2s)})$ & \cref{prop:plug-and-solve-rate} \\
Density-unknown ($\rho_\theta$ parametric) & PG with MLE plug-in & $\tilde O(1/\eta)$ & \cref{cor:parametric-special-case} \\
\bottomrule
\end{tabular}\\[2pt]
{\footnotesize $^\dagger$Density-known Newton uses zero new samples; reported is the iteration count to accuracy $\eta$.}
\end{table}

\paragraph{Comparison with multiplicative-noise LQR}
The most directly comparable prior problem is policy gradient for multiplicative-noise LQR, where the established rate is $\tilde O(1/\eta^2)$~\cite{gravell2020learning}. The $\tilde O(1/\eta)$ rate we obtain under density-known access (\cref{thm:sample-complexity-density-known}) is faster, but the gap reflects a property of the cost rather than a sharper analysis: the log-growth objective admits a closed-form single-transition gradient oracle (section~\ref{51-single-transition-closed-form-gradient-oracle}), absent from the LQR cost, which accumulates across time through a gain-dependent steady-state covariance and therefore requires zeroth-order finite-difference estimation. We do not claim a sharper analysis of multiplicative-noise LQR.

\subsection{Paper organization and notation}\label{15-paper-organization}

Section~\ref{2-problem-formulation-and-the-cusp-obstruction} formalizes the problem and proves the cusp obstruction.
Section~\ref{3-cauchy-regularization-and-uniform-in-epsilon-polyak-Lojasiewicz} introduces the Cauchy regularization, the closed-form Hessian decomposition, and the uniform-in-$\varepsilon$ PL lemma.
Section~\ref{4-finite-sample-analysis} develops the finite-sample analysis (naive-estimator variance divergence and the paired-estimator reduction).
Section~\ref{5-sample-complexity-in-the-pl-basin} assembles the local sample-complexity bounds. Section~\ref{6-numerical-validation} reports numerical validation; section~\ref{7-conclusion} concludes with four limitations. The proof of \cref{lem:uniform-PL-constant} is in \cref{appendix-a-proof-of-lemma-34}; auxiliary derivations are in \cref{appendix-b-auxiliary-derivations}. The supplement (sections~SM5 and~SM6) contains the preliminary-phase analysis, the quadrature-stability ablation, auxiliary numerical experiments, and the technical lemma on PV--directional-derivative equivalence at a moving simple pole.

\paragraph{Notation}
\Cref{tab:notation} collects the symbols used throughout. We write $\tilde O(\cdot)$ for $O(\cdot)$ up to polylogarithmic factors in $1/\eta$.

\begin{table}[htbp]
\centering
\footnotesize
\setlength{\tabcolsep}{4pt}
\caption{Notation used throughout the paper.}
\label{tab:notation}
\begin{tabular}{@{}p{0.20\linewidth}p{0.78\linewidth}@{}}
\toprule
Symbol & Meaning \\ \midrule
$\rho$ & noise density on $[b_{\min},b_{\max}]$ \\
$\mathcal{K}_{\rm stab}$ & stabilizing region $\{K : J(K) < 0\}$ \\
$K^*$ & a local minimizer of $J$ on $\mathcal{K}_{\rm stab}$ (principal-value first-order point) \\
$J^* := J(K^*)$ & optimal log-growth cost \\
$\mathcal{N}_\delta$ & local basin $[K^*-\delta, K^*+\delta]$ \\
$b_{\rm sing}(K) := -1/K$ & moving singularity of $dJ/dK$ \\
$v(b,K) := 1+bK$ & inverse weight \\
$\bar B := 2 b_{\rm sing}(K)-B$ & reflection of $B$ through the pole \\
$h(b) := b\rho(b)$ & weighted density \\
$\mathcal{S}_K(R), \mathcal{A}_K(R)$ & symmetric pairing zone of radius $R$ and its complement in $[b_{\min},b_{\max}]$ \\
$\tau$ & pole-to-edge margin (\cref{lem:uniform-PL-constant}(a)) \\
$\mu_0, L_0$ & uniform PL constant / Lipschitz bound on $\mathcal{N}_\delta$ \\
$\psi(B;K,\varepsilon)$ & naive single-sample gradient estimator \\
$\tilde\psi$ & density-aware paired estimator \\
$\hat\rho_{n_1,h}$ & order-$s$ kernel density estimate from $n_1$ samples \\
$\nu_{n_1},\nu'_{n_1}$ & sup-norm KDE rates for $\hat\rho, \hat\rho'$ \\
$w_\rho, w_{\hat\rho}, \Delta$ & pairing weight, plug-in weight, weight discrepancy \\
$\bar C_b$ & bias coefficient $\pi\rho(b_{\rm sing}(K^*))/|K^*|$ \\
$\mathrm{PV}\!\int,\,\mathrm{Hf}\!\int$ & Cauchy principal value / Hadamard finite-part \cite{king2009hilbert} \\
$\tilde O(\cdot)$ & $O(\cdot)$ up to polylogarithmic factors in $1/\eta$ \\
\bottomrule
\end{tabular}
\end{table}

\section{Problem formulation and the cusp obstruction}\label{2-problem-formulation-and-the-cusp-obstruction}
\subsection{The log-growth control problem}\label{21-the-log-growth-control-problem}

Consider the scalar linear control system \eqref{eq:sys} under the standing
assumption \(|a| = 1\) (obtained without loss of generality by the
logarithmic shift \(J \mapsto J + \log|a|\)). Given the density \(\rho\)
of the multiplicative noise \(B\), the parametric objective for linear
feedback \(U_t = K X_t\) is
\begin{equation}\label{eq:cost_param}
    J(K) = \mathbb{E}\!\left[\log\lvert 1+BK\rvert\right] = \int_{b_{\min}}^{b_{\max}} \rho(b)\log\lvert 1+bK\rvert\,db.
\end{equation}
The following standing assumptions hold throughout the paper unless
explicitly relaxed.

\begin{assumption}[Compactly supported positive noise]
\label[assumption]{asm:compact_noise}\(\rho\) is a probability density supported on a compact interval
\([b_{\min}, b_{\max}]\) with \(0 < b_{\min} < b_{\max}\).
\end{assumption}

\begin{assumption}[Regularity]
\label[assumption]{asm:regularity}\(\rho \in C^2([b_{\min}, b_{\max}])\) and \(\rho(b) > 0\) for all
\(b \in [b_{\min}, b_{\max}]\).
\end{assumption}

\begin{assumption}[Non-degeneracy of the optimum]
\label[assumption]{asm:optimum}The Hadamard finite-part second derivative of $J$ at $K^*$ is strictly positive:
\begin{equation}\label{eq:A3_condition}
    \left.\frac{d^2 J}{dK^2}\right|_{K=K^*} > 0,
\end{equation}
where the right-hand side is given in closed form by~\eqref{eq:hess_Kstar} in \cref{thm:hessian-decomposition}.
\end{assumption}

\begin{remark}[genericity of \cref{asm:optimum}]
\label{rem:A3-generic}
The closed-form expression \eqref{eq:hess_Kstar} is a continuous functional of $\rho$ in the $C^2$ topology, so $\{\rho : d^2J/dK^2(K^*) > 0\}$ is an open subset of the admissible-density space, and it is also dense (a transversality argument: if $d^2J/dK^2(K^*) = 0$ at some $\rho_0$, an arbitrarily small $C^2$ perturbation moves it to a nonzero value via the linear dependence of \eqref{eq:hess_Kstar} on $\rho$). \Cref{asm:optimum} is therefore \emph{generic}. Numerically, all four test densities of section~\ref{6-numerical-validation} yield $d^2J/dK^2|_{K^*} \in \{11.18, 16.97, 16.98, 12.96\}$. A sharp analytical sufficient condition --- in particular, a structural assumption on $\rho$ (such as log-concavity of $b\rho(b)$) implying \cref{asm:optimum} --- remains open; see L4 of section~\ref{72-limitations}.
\end{remark}

\begin{remark}
\cref{asm:compact_noise} encompasses the canonical
Ranade--Sahai setting~\cite{ranade2018control}, where \(B\) models a bounded
multiplicative disturbance with support bounded away from zero to ensure
controllability. The positivity of the support is a natural restriction:
if \(0 \in \mathrm{supp}(\rho)\), the system loses controllability on
the event \(B_t = 0\) and the log-growth cost diverges. \Cref{asm:regularity} is
a mild smoothness condition standard in density-dependent analyses~\cite{Giné_Nickl_2021}.
\end{remark}

The \emph{singularity locus} associated with gain \(K\) is
\begin{equation}\label{eq:bsing}
    b_{\mathrm{sing}}(K) := -\frac{1}{K},
\end{equation}
the unique value of \(b\) satisfying \(1 + bK = 0\). The integrand of
\(J(K)\) in \eqref{eq:cost_param} has a logarithmic singularity at
\(b = b_{\mathrm{sing}}(K)\) whenever this value lies in the support
interval \([b_{\min}, b_{\max}]\). Although the logarithm is Lebesgue
integrable at this singularity, its first derivative in \(K\) is not ---
a fact that drives the entire analysis.

\subsection{The stabilizing region and existence of the optimum}\label{22-the-stabilizing-region-and-existence-of-the-optimum}

We characterize the set of gains yielding contracting closed-loop
dynamics and establish existence of a capacity-achieving gain \(K^*\) on
this set.

\begin{definition}[stabilizing region]
\label{def:stab_reg}The stabilizing region is
\(\mathcal{K}_{\mathrm{stab}} := \{K \in \mathbb{R} : J(K) < 0\}\).
\end{definition}
Under the \(|a| = 1\) normalization,
\(K \in \mathcal{K}_{\mathrm{stab}}\) is equivalent to the almost-sure
contraction of \((X_t)_{t \ge 0}\): the Lyapunov exponent is
\(\log|a| + J(K) = J(K) < 0\), so \(|X_t| \to 0\) almost surely. Since
the multiplicative noise is bounded with
\(\mathrm{supp}(\rho) \subset (0, \infty)\), the case \(K = 0\) gives
\(J(0) = 0\) (no contraction), consistent with \(0\) lying on the
closure of \(\mathcal{K}_{\mathrm{stab}}\) but not its interior. The
characterization \(\mathcal{K}_{\rm stab} = \{K : J(K) < 0\}\) is the
negative-top-Lyapunov-exponent set for the i.i.d. random product
\(\prod_{t \ge 0}(1 + B_t K)\), a special case of the multiplicative
ergodic framework of Arnold \cite{arnold2006random}. Multiplicative-noise stabilization
in the linear-systems context goes back to Wonham\textquotesingle s
stationary control of state-dependent-noise systems \cite{wonham1967optimal} and is now
standardly framed via stochastic Riccati equations \cite{damm2004rational}; the i.i.d.
case considered here is also a degenerate (memoryless) case of the
Markov-jump linear systems studied by Costa, Fragoso, and Marques
\cite{costa2005discrete}. We emphasize that \emph{characterization} of
\(\mathcal{K}_{\rm stab}\) is classical; what is novel here is the
\emph{sample-complexity} analysis for \emph{learning} an optimal
\(K \in \mathcal{K}_{\rm stab}\) under the log-growth cost.

Under \cref{asm:compact_noise,asm:regularity}, $\mathcal{K}_{\mathrm{stab}}$ is open and $J: \mathcal{K}_{\mathrm{stab}} \to \mathbb{R}$ is continuous (by dominated convergence, since the integrand of \eqref{eq:grad} is locally bounded away from the singularity for $K \in \mathcal{K}_{\rm stab}$). Throughout the paper, $K^*$ denotes a solution of the principal-value first-order condition \eqref{eq:pv_foc} lying in $\mathcal{K}_{\mathrm{stab}}$; existence of such a $K^*$ follows from $J(K) \to 0$ at the boundary of $\mathcal{K}_{\rm stab}$ together with $J(K) < 0$ on the interior. That $K^*$ is in fact a strict local minimum of $J$, in the classical sense, follows from the strict positivity of the Hadamard finite-part Hessian at $K^*$ — itself a consequence of \cref{asm:optimum} and \cref{thm:hessian-decomposition}, made precise in \cref{lem:uniform-PL-constant}.

\begin{remark}[local vs. global minimum]
We do not establish uniqueness of $K^*$ or global convexity of $J$. The algorithmic guarantees of sections~\ref{3-cauchy-regularization-and-uniform-in-epsilon-polyak-Lojasiewicz}--\ref{5-sample-complexity-in-the-pl-basin} require only local gradient domination on a neighborhood of a fixed local minimum (\cref{lem:uniform-PL-constant}); the initialization assumption of \cref{thm:total-sample-complexity-pl-basin} places the policy-gradient trajectory within the basin of $K^*$. Numerical evidence on canonical test densities (uniform, Beta, truncated Gaussian) supports unique global minimality of $K^*$.
\end{remark}

\begin{lemma}[strict cusp at the optimum]
\label{lem:strict-cusp-optimum} Under \cref{asm:compact_noise,asm:regularity}, every local minimum \(K^*\) in \(\mathcal{K}_{\rm stab}\) of \(J\) in the principal-value sense satisfies \(b_{\rm sing}(K^*) \in (b_{\min}, b_{\max})\).
\end{lemma}
\begin{proof}
The first-order condition for $K^*$, in the principal-value 
sense, is
\begin{equation}\label{eq:pv_foc}
    \mathrm{PV}\!\int_{b_{\min}}^{b_{\max}} \frac{b\,\rho(b)}{1+bK^*}\,db = 0.
\end{equation}
By \cref{asm:compact_noise,asm:regularity}, $b\rho(b)$ is strictly positive on the 
interior $(b_{\min}, b_{\max})$. The kernel $1/(1+bK^*)$ is the only 
sign-changing factor; if it has constant sign on the support, the 
PV integral has the same sign as $b\rho(b)/(1+bK^*)$, contradicting 
its vanishing. Hence $1+bK^*$ must take both signs on 
$[b_{\min}, b_{\max}]$, i.e., its unique zero 
$b_{\rm sing}(K^*) = -1/K^*$ lies in the open interval 
$(b_{\min}, b_{\max})$.
\end{proof}

\subsection{The cusp obstruction}\label{23-the-cusp-obstruction}

The principal technical obstacle to applying classical policy-gradient
analysis to \eqref{eq:cost_param} is that \(J\) does not admit a Lebesgue-integrable
gradient on an open subset of \(\mathcal{K}_{\mathrm{stab}}\) containing
\(K^*\). The Cauchy principal-value framework is classical in applied
analysis (King \cite{king2009hilbert} is a comprehensive modern reference). Its appearance here as an obstruction to policy-gradient theory appears, to our knowledge, to be new: we are not aware of prior policy-gradient analyses that resolve the cusp obstruction.

\begin{proposition}[non-Lebesgue gradient]
\label{prop:non-lebesgue-gradient}Fix
\(K \in \mathcal{K}_{\mathrm{stab}}\) with
\(b_{\mathrm{sing}}(K) \in (b_{\min}, b_{\max})\). Then:

(a) The formal gradient integrand \(\phi_K(b) := \rho(b)b/(1+bK)\) is
not Lebesgue integrable on \([b_{\min}, b_{\max}]\); in particular
\(\int |\phi_K(b)|\,db = +\infty\).

(b) The Cauchy principal value
\(\mathrm{PV}\!\int_{b_{\min}}^{b_{\max}} \phi_K(b)\,db\) exists and
equals the one-sided directional derivative
\(\lim_{h \downarrow 0}[J(K+h) - J(K)]/h = \lim_{h \uparrow 0}[J(K+h) - J(K)]/h\).

(c) There is no continuous extension of \(\phi_K\) across
\(b_{\mathrm{sing}}(K)\):
\(\lim_{b \uparrow b_{\mathrm{sing}}(K)} \phi_K(b) = -\infty\cdot\mathrm{sign}(K)\)
and
\(\lim_{b \downarrow b_{\mathrm{sing}}(K)} \phi_K(b) = +\infty\cdot\mathrm{sign}(K)\).
\end{proposition}
\begin{proof}
\emph{(a)} Near \(b = b_{\mathrm{sing}}(K)\), write
\(s := b - b_{\mathrm{sing}}(K)\), so \(1 + bK = Ks\). Then
\begin{equation}
    |\phi_K(b_{\mathrm{sing}} + s)| = \frac{\rho(b_{\mathrm{sing}}+s)\cdot|b_{\mathrm{sing}}+s|}{|K|\cdot|s|}.
\end{equation}
By \cref{asm:regularity}, \(\rho\) is continuous and strictly positive at
\(b_{\mathrm{sing}}(K) \in (b_{\min}, b_{\max})\), so
\begin{equation}\label{eq:phi_limit}
    \rho(b_{\mathrm{sing}}+s)\cdot|b_{\mathrm{sing}}+s| \;\longrightarrow\; \rho(b_{\mathrm{sing}})\cdot b_{\mathrm{sing}} > 0
    \qquad \text{as } s \to 0.
\end{equation}
Hence \(|\phi_K(b)| \ge C/|s|\) for \(|s|\) small and
some \(C > 0\), and \(\int 1/|s|\,ds\) diverges logarithmically, proving
\(\phi_K \notin L^1\).

\emph{(b)} Split
\([b_{\min}, b_{\max}] = [b_{\min}, b_{\mathrm{sing}} - h] \cup [b_{\mathrm{sing}} - h, b_{\mathrm{sing}} + h] \cup [b_{\mathrm{sing}} + h, b_{\max}]\)
and define the principal value as
\begin{equation}
    \mathrm{PV}\!\int \phi_K\,db := \lim_{h \downarrow 0}\left[\int_{b_{\min}}^{b_{\mathrm{sing}}-h}\phi_K\,db + \int_{b_{\mathrm{sing}}+h}^{b_{\max}}\phi_K\,db\right].
\end{equation}
On the symmetric shell
\([b_{\mathrm{sing}} - h, b_{\mathrm{sing}} + h]\), expand
\begin{equation}\label{eq:phi_expand}
    \phi_K(b_{\mathrm{sing}} \pm s) = \frac{\rho(b_{\mathrm{sing}} \pm s)\cdot(b_{\mathrm{sing}} \pm s)}{\pm Ks}.
\end{equation}
The leading \(1/s\) terms cancel between the \(+s\) and \(-s\) pieces by
odd parity; the remaining contributions are \(O(1)\) and integrable, so
the two one-sided limits agree and the principal value exists.

For the directional derivative,
\begin{equation}\label{eq:dir_deriv_diff}
    J(K+h) - J(K) = \int_{b_{\min}}^{b_{\max}} \rho(b)\bigl[\log|1+b(K+h)| - \log|1+bK|\bigr]\,db.
\end{equation}
For $b$ bounded away from the singularity $b_{\mathrm{sing}}(K) = -1/K$, the difference quotient's integrand converges uniformly to $\phi_K(b) = \rho(b)b/(1+bK)$ as $h \to 0$. However, because the pole $b_{\mathrm{sing}}(K+h)$ moves with $h$, the interchange of the limit and integration requires explicitly evaluating the local parity cancellation around the moving singularity. Lemma~SM1.1 (stated and proved in the supplement) establishes this exact evaluation under the hypothesis that $\rho \in C^1(\mathbb{R})$ has sufficient decay at infinity; this hypothesis is implied by compactness of $\mathrm{supp}(\rho) \subset [b_{\min}, b_{\max}]$ under \cref{asm:compact_noise,asm:regularity}, since $\rho$ extends by zero outside its compact support. Applying the lemma yields $\lim_{h \to 0}[J(K+h) - J(K)]/h = \mathrm{PV}\!\int \phi_K(b)\,db$.

\emph{(c)} Direct from the form
\(\phi_K(b_{\mathrm{sing}}+s) = \rho(b_{\mathrm{sing}}+s)(b_{\mathrm{sing}}+s)/(Ks)\)
and continuity and positivity of \(\rho\cdot b\) at
\(b_{\mathrm{sing}}\). 
\end{proof}

\Cref{prop:non-lebesgue-gradient} formalizes the cusp obstruction of
\cref{11-problem-and-motivation}: the non-Lebesgue character of \(dJ/dK\)
at any \(K\) with \(b_{\mathrm{sing}}(K)\) interior to
\(\mathrm{supp}(\rho)\). The name reflects both the singular (infinite)
behavior of the gradient integrand at \(b_{\mathrm{sing}}(K)\) in sample
space, and the non-smooth cusp-like structure of the expected gradient
\(dJ/dK\) as a function of \(K\) at the optimum.

Consequently, \(dJ/dK\) admits no Lebesgue-defined Lipschitz constant on
any compact \(\mathcal{C} \subset \mathcal{K}_{\rm stab}\) that contains
a \(K\) with \(b_{\rm sing}(K) \in (b_{\min}, b_{\max})\) --- including
\(\mathcal{C} = \{K^*\}\) --- so any policy-gradient analysis predicated
on the existence of a global gradient-Lipschitz constant is structurally
inapplicable. We add that the Lyapunov-equation almost-smoothness machinery of \cite{fazel2018global, gravell2020learning} is structurally unavailable for the log-growth cost: the per-step increment $\log|X_{t+1}/X_t| = \log|a| + \log|1 + B_t K|$ is additive in $\log|X_t|$, not in $X_t^2$; the Bellman value function for the resulting log-growth cost is therefore not polynomial in $X_t$, and the algebraic Lyapunov-equation machinery used for LQR in \cite{fazel2018global, gravell2020learning} has no log-growth analog. The regularization developed in section~\ref{3-cauchy-regularization-and-uniform-in-epsilon-polyak-Lojasiewicz} takes a different route.

\section{Cauchy regularization and the uniform Polyak--{\L}ojasiewicz bound}\label{3-cauchy-regularization-and-uniform-in-epsilon-polyak-Lojasiewicz}

\subsection{The regularized objective}\label{31-the-regularized-objective}

We introduce a one-parameter family of smooth objectives converging to
\(J\).

\begin{definition}[Cauchy-regularized log-growth cost]
\label{def:cauchy-regularized-log-growth-cost}For
\(\varepsilon > 0\), let
\begin{equation}\label{eq:reg_cost}
    J_\varepsilon(K) := \mathbb{E}\!\left[\tfrac{1}{2}\log\bigl((1+BK)^2 + \varepsilon^2\bigr)\right] 
= \int_{b_{\min}}^{b_{\max}} \rho(b)\cdot\tfrac{1}{2}\log\bigl((1+bK)^2 + \varepsilon^2\bigr)\,db.
\end{equation}
\end{definition}
The integrand is \(C^\infty\) in \((b, K)\) for each
\(\varepsilon > 0\), so \(J_\varepsilon\) inherits the regularity of
\(\rho\) through the domain of integration:
\(J_\varepsilon \in C^\infty(\mathbb{R})\) whenever \(\rho \in C^0\),
independently of whether \(b_{\mathrm{sing}}(K)\) lies in
$\mathrm{supp}(\rho)$.

The following regularity and bias estimates are direct consequences of
\cref{def:cauchy-regularized-log-growth-cost}.

\begin{proposition}[basic properties of \(J_\varepsilon\)]
\label{prop:properties-J-epsilon}Under \cref{asm:compact_noise,asm:regularity}:

(a) For each \(\varepsilon > 0\),
\(J_\varepsilon \in C^\infty(\mathbb{R})\), and
\begin{equation}
    \frac{dJ_\varepsilon}{dK}(K) = \int_{b_{\min}}^{b_{\max}} \rho(b)\cdot\frac{b(1+bK)}{(1+bK)^2 + \varepsilon^2}\,db.
\end{equation}

(b) \(J_\varepsilon \to J\) pointwise on \(\mathcal{K}_{\mathrm{stab}}\)
as \(\varepsilon \downarrow 0\) and for every
\(K \in \mathcal{K}_{\rm stab}\) and every \(\varepsilon \ge 0\),
\begin{equation}
    J(K) \le J_\varepsilon(K).
\end{equation}

(c) \emph{(Bias at the optimum.)} The bias at the optimum satisfies
\begin{equation}\label{eq:Jeps_bias}
    J_\varepsilon^* - J^* = \frac{\pi\,\rho(b_{\mathrm{sing}}(K^*))}{|K^*|}\,\varepsilon + O(\varepsilon^2), \qquad \varepsilon \downarrow 0.
\end{equation}
\end{proposition}
The proof of (a) and (b) follows from dominated convergence, noting that
\(\tfrac{1}{2}\log((1+bK)^2 + \varepsilon^2) \to \log|1+bK|\) pointwise
for \(b \ne b_{\rm sing}(K)\) and that the sets where convergence fails
have Lebesgue measure zero. The bound in (b) holds because
\((1+bK)^2 + \varepsilon^2 \ge (1+bK)^2\), so
\(\tfrac{1}{2}\log((1+bK)^2 + \varepsilon^2) \ge \log|1+bK|\) pointwise
in \(b\); integrating against \(\rho\) yields the claim.

The proof of (c) --- sketched in Appendix B.1 --- turns on the same
parity argument as the proof of \cref{lem:uniform-PL-constant} (Appendix A), applied to a different integrand.
Substituting \(s = b - b_\star\) with \(b_\star = -1/K^*\) and rescaling
\(t = K^* s/\varepsilon\), the bias splits into three terms via Taylor
expansion of \(\rho\) about \(b_\star\):
\begin{equation}
\begin{aligned}
        J_\varepsilon(K^*) - J(K^*)=& \frac{\varepsilon}{|K^*|}\!\int_{-T}^{T}\!\!\Bigl[\,\underbrace{\rho(b_\star)}_{\to\pi\rho(b_\star)\text{ via }I_\infty} + \underbrace{\tfrac{\varepsilon t}{K^*}\,\rho'(b_\star)}_{\emph{vanishes by odd parity}}+ O(\varepsilon^2 t^2)\,\Bigr]\\ &\cdot \tfrac{1}{2}\log\!\bigl(1 + 1/t^2\bigr)\,dt+O(\varepsilon^2),
\end{aligned}
\end{equation}
with \(T = |K^*|R/\varepsilon\). The constant term, integrated against
the Poisson-kernel-like profile \(\tfrac{1}{2}\log(1+1/t^2)\) on the
real line, gives exactly \(\pi\rho(b_\star)\) (a Frullani-type
evaluation). The linear-in-\(t\) term, an odd function of \(t\) over the
symmetric interval \([-T, T]\), integrates to zero. The quadratic
remainder contributes \(O(\varepsilon^2)\). The leading
\(\pi\rho(b_\star)\varepsilon/|K^*|\) has no logarithmic factor.

\subsection{The Hessian representation}\label{32-the-hessian-representation}

The analytical heart of section~\ref{3-cauchy-regularization-and-uniform-in-epsilon-polyak-Lojasiewicz} is a closed-form integration-by-parts decomposition of $d^2 J_\varepsilon/dK^2$ that renders the $\varepsilon$-dependence of the second derivative explicit and enables the uniform-in-$\varepsilon$ PL bound of section~\ref{33-the-uniform-in-epsilon-polyak-Lojasiewicz-lemma}.

\begin{theorem}[Hessian decomposition]
\label{thm:hessian-decomposition}Under \cref{asm:compact_noise,asm:regularity}, for every \(K \in \mathcal{K}_{\mathrm{stab}}\) and every
\(\varepsilon > 0\),
\begin{equation}\label{eq:hess_decomp}
    \frac{d^2 J_\varepsilon}{dK^2}(K) = B_\varepsilon(K) + I_\varepsilon(K),
\end{equation}

where, with \(v(b, K) := 1 + bK\) and
\(g(b) := 2b\rho(b) + b^2\rho'(b)\),
\begin{equation}\label{eq:hess_B}
    B_\varepsilon(K) := \frac{1}{K}\!\left[\frac{\rho(b)\,b^2\,v(b,K)}{v(b,K)^2 + \varepsilon^2}\right]_{b = b_{\min}}^{b = b_{\max}},
\end{equation}
\begin{equation}\label{eq:hess_I}
    I_\varepsilon(K) := -\,\frac{1}{K}\!\int_{b_{\min}}^{b_{\max}} g(b)\cdot\frac{v(b,K)}{v(b,K)^2 + \varepsilon^2}\,db.
\end{equation}
As \(\varepsilon \downarrow 0\), the decomposition passes to the limit
in the Hadamard finite-part sense: both \(B_\varepsilon(K)\) and
\(I_\varepsilon(K)\) have well-defined limits \(B_0(K)\) and \(I_0(K)\)
for every \(K \in \mathcal{K}_{\mathrm{stab}}\) with
\(b_{\mathrm{sing}}(K) \in (b_{\min}, b_{\max})\), and
\(d^2 J_0/dK^2(K) := B_0(K) + I_0(K)\) gives the Hadamard finite-part
extension of the (otherwise divergent) formal second derivative of
\(J\).

Setting \(K = K^*\) in \eqref{eq:hess_decomp}--\eqref{eq:hess_I} and passing to the Hadamard
finite-part limit gives
\begin{equation}\label{eq:hess_Kstar}
    \left.\frac{d^2 J}{dK^2}\right|_{K=K^*} = \frac{1}{K^*}\!\left[\frac{\rho(b_{\max})b_{\max}^2}{1 + b_{\max}K^*} - \frac{\rho(b_{\min})b_{\min}^2}{1 + b_{\min}K^*}\right]  - \frac{1}{K^*}\,\mathrm{PV}\!\int_{b_{\min}}^{b_{\max}}\!\frac{g(b)}{1 + bK^*}\,db.
\end{equation}
The right-hand side is a finite quantity for every \(\rho\) satisfying
\cref{asm:compact_noise,asm:regularity}.
\end{theorem}
\begin{proof}
We differentiate the expression for \(dJ_\varepsilon/dK\) in \cref{prop:properties-J-epsilon}(a) twice in $K$ under the integral sign (permissible for \(\varepsilon > 0\) by standard dominated-convergence arguments), obtaining
\begin{equation}\label{eq:hess_raw}
    \frac{d^2 J_\varepsilon}{dK^2}(K) = \int_{b_{\min}}^{b_{\max}} \rho(b)\,b^2\cdot\frac{\varepsilon^2 - v(b,K)^2}{(v(b,K)^2 + \varepsilon^2)^2}\,db.
\end{equation}
Since $\partial_b v(b,K) = K$ (differentiation here is in $b$, at fixed $K$), the integrand of \eqref{eq:hess_raw} equals
$\rho(b)b^2 \cdot K^{-1}\,\partial_b\!\left[v/(v^2+\varepsilon^2)\right]$.
Integration by parts in $b$ over $[b_{\min}, b_{\max}]$ gives
\begin{equation}
        \frac{d^2 J_\varepsilon}{dK^2}(K) = \frac{1}{K}\!\left[\rho(b)\,b^2\cdot\frac{v(b,K)}{v(b,K)^2+\varepsilon^2}\right]_{b_{\min}}^{b_{\max}} - \frac{1}{K}\int_{b_{\min}}^{b_{\max}}\!\frac{d}{db}[\rho(b)\,b^2]\cdot\frac{v(b,K)}{v(b,K)^2+\varepsilon^2}\,db,
\end{equation}
which matches \eqref{eq:hess_decomp}--\eqref{eq:hess_I} with
\(g(b) = d/db[\rho(b)b^2] = 2b\rho(b) + b^2\rho'(b)\). The boundary term
\(B_\varepsilon \to B_0\) classically as \(\varepsilon \downarrow 0\)
since \(v(b_\pm, K)\) is bounded away from zero for
\(K \in \mathcal{K}_{\rm stab}\) with \(b_{\rm sing}(K)\) interior to
\((b_{\min}, b_{\max})\).

On the near-shell $|b - b_{\rm sing}(K)| \le R$, the parity-shell 
decomposition of Appendix A applied to the integrand of $I_\varepsilon$ 
--- with $g$ playing the role of $f$ in that proof --- gives a constant 
term whose integral
$\int_{-R}^R Ks/(K^2 s^2 + \varepsilon^2)\,ds$
vanishes identically for every $\varepsilon \ge 0$ by odd symmetry of 
the integrand on $[-R, R]$ (at $\varepsilon = 0$ this is the 
symmetric-cutoff principal-value definition \eqref{eq:bsing}), a linear term 
continuous in $\varepsilon$ at zero, and a quadratic remainder of size 
$O(R^2 K^{-2})$ uniform in $\varepsilon$. The convergence 
$I_\varepsilon(K) \to I_0(K)$ as $\varepsilon \downarrow 0$ holds 
term-by-term in the parity-shell decomposition, with the 
$\varepsilon = 0$ value of the constant term interpreted in the 
principal-value sense.
\end{proof}

\subsection{\texorpdfstring{The uniform-in-$\varepsilon$ Polyak--{\L}ojasiewicz lemma}{The uniform-in-epsilon Polyak-Lojasiewicz lemma}}\label{33-the-uniform-in-epsilon-polyak-Lojasiewicz-lemma}

The decomposition \eqref{eq:hess_decomp} is the technical vehicle for the principal result
of section~\ref{3-cauchy-regularization-and-uniform-in-epsilon-polyak-Lojasiewicz}: a uniform-in-$\varepsilon$ bound on $d^2J_\varepsilon/dK^2$ near $K^*$.

\begin{lemma}[uniform-in-\(\varepsilon\) PL constant]
\label{lem:uniform-PL-constant}Under \cref{asm:compact_noise,asm:regularity,asm:optimum}, there exist constants \(\delta > 0\), \(\varepsilon_0 > 0\),
\(\mu_0 > 0\), and \(L_0 < \infty\) such that, writing
\(\mathcal{N}_\delta := [K^* - \delta, K^* + \delta]\):

(a)
\(\{b_{\mathrm{sing}}(K) : K \in \mathcal{N}_\delta\} \subset [b_{\min} + \tau, b_{\max} - \tau]\)
for some \(\tau > 0\) depending only on \(\delta\).

(b) The map \((K, \varepsilon) \mapsto d^2 J_\varepsilon/dK^2(K)\) is
jointly continuous on \(\mathcal{N}_\delta \times [0, \varepsilon_0]\),
where the value at \(\varepsilon = 0\) is defined via the Hadamard
finite part of \cref{thm:hessian-decomposition}.

(c) For all
\((K, \varepsilon) \in \mathcal{N}_\delta \times [0, \varepsilon_0]\),
\begin{equation}
    \mu_0 \le \frac{d^2 J_\varepsilon}{dK^2}(K) \le L_0.
\end{equation}
The explicit forms of \(\delta, \varepsilon_0, \mu_0, L_0\) are given in
the proof.
\end{lemma}
The proof is given in Appendix A. We describe only the analytical
mechanism, which is the core conceptual content of the paper. After the
integration-by-parts decomposition \eqref{eq:hess_decomp}, the boundary term
\(B_\varepsilon\) is uniformly bounded on
\(\mathcal{N}_\delta \times [0, \varepsilon_0]\) because the support
endpoints \(b_{\min}, b_{\max}\) stay away from the moving pole
\(b_{\rm sing}(K)\) for \(K \in \mathcal{N}_\delta\) (part (a) of the
lemma). The integral term \(I_\varepsilon\) has a near-pole component
which, after substituting \(s = b - b_{\rm sing}(K)\) and
Taylor-expanding \(g\) about \(b_{\rm sing}(K)\), contains a leading
Cauchy-kernel piece
\begin{equation}
    -\frac{g(b_{\rm sing}(K))}{K}\int_{-R}^{R}\frac{Ks}{K^2 s^2 + \varepsilon^2}\,ds = 0,
\end{equation}
which \emph{vanishes identically by odd symmetry of the integrand over
the symmetric interval \([-R, R]\)}. Without this vanishing, the piece
would scale as \(\varepsilon^{-1}\) and destroy the uniform \emph{upper}
bound \(L_0\). The remaining linear and quadratic Taylor terms are
uniformly bounded in \(\varepsilon\). This odd-parity cancellation --- a structural feature of the Cauchy kernel about the singularity $b_{\rm sing}(K) = -1/K$, not a property of the regularization --- is the central analytical mechanism of the paper. The same mechanism appeared in the proof of the bias formula in section~\ref{31-the-regularized-objective} (linear-in-$t$ term vanishes by odd parity over $[-T, T]$), and will reappear in section~\ref{42-the-density-aware-symmetric-pairing-estimator} to give a uniform variance bound for the symmetric-pairing estimator.

\begin{lemma}[Persistence of the regularized minimizer]
\label{lem:persistence-regularized-minimizer}Under the hypotheses of \cref{lem:uniform-PL-constant}, there exist
\(\varepsilon_1 \in (0, \varepsilon_0]\) and a constant
\(C_K = C_K(\rho, K^*, \mu_0, R) < \infty\) such that:

(a) For every \(\varepsilon \in [0, \varepsilon_1]\), \(J_\varepsilon\)
has a unique minimizer \(K_\varepsilon^*\) on \(\mathcal{N}_\delta\),
lying in the interior \((K^* - \delta,\, K^* + \delta)\) and
characterized by \(dJ_\varepsilon/dK(K_\varepsilon^*) = 0\).

(b) \(\bigl|K_\varepsilon^* - K^*\bigr| \le C_K\,\varepsilon\), where $C_K := 2\pi\,|\rho'(b_\star)\,b_\star + \rho(b_\star)|\,/\,(\mu_0\,(K^*)^2)$ with $b_\star := -1/K^*$ (matching the constant $C_1$ derived in part~(b.1) of the proof).

(c)
\(0 \le J_\varepsilon(K^*) - J_\varepsilon(K_\varepsilon^*) \le \tfrac{L_0}{2}\,C_K^2\,\varepsilon^2\).
\end{lemma}
\begin{proof}
\emph{(a) Existence, uniqueness, interior persistence.}
By \cref{lem:uniform-PL-constant}(c), $d^2 J_\varepsilon/dK^2 \ge \mu_0 > 0$ uniformly on $\mathcal{N}_\delta \times [0, \varepsilon_0]$, so $J_\varepsilon$ is strictly convex on $\mathcal{N}_\delta$ for every $\varepsilon \in [0, \varepsilon_0]$ and admits a unique minimizer $K_\varepsilon^*$ there. To show $K_\varepsilon^*$ lies in the interior, we verify that $dJ_\varepsilon/dK$ has opposite signs at the endpoints of $\mathcal{N}_\delta$ for $\varepsilon$ small.

By joint continuity of \(dJ_\varepsilon/dK\) on
\(\mathcal{N}_\delta \times [0, \varepsilon_0]\) --- established in the
proof of \cref{thm:hessian-decomposition} by differentiation under the integral for
\(\varepsilon > 0\) and by the parity-shell decomposition of Appendix A
at \(\varepsilon = 0\) (where \(dJ/dK\) denotes the principal-value
extension) --- and compactness of \(\mathcal{N}_\delta\),
\begin{equation}
    \sup_{K \in \mathcal{N}_\delta}\Bigl|\frac{dJ_\varepsilon}{dK}(K) - \frac{dJ}{dK}(K)\Bigr| \to 0 \quad \text{as } \varepsilon \downarrow 0.
\end{equation}

Strict convexity of \(J\) on \(\mathcal{N}_\delta\) by \cref{lem:uniform-PL-constant}(c) applied at \(\varepsilon = 0\), where the second derivative is the Hadamard finite-part of \cref{thm:hessian-decomposition}, and the first-order condition \(dJ/dK(K^*) = 0\)
give \(dJ/dK(K^* - \delta) < 0 < dJ/dK(K^* + \delta)\). Choose
\(\varepsilon_1 \in (0, \varepsilon_0]\) so that the displayed supremum
is below
\begin{equation}
    \tfrac{1}{2}\,\min\Bigl(-\tfrac{dJ}{dK}(K^* - \delta), \tfrac{dJ}{dK}(K^* + \delta)\Bigr) \quad \text{for all } \varepsilon \in [0, \varepsilon_1].
\end{equation}

Then for every \(\varepsilon \in [0, \varepsilon_1]\),
\(dJ_\varepsilon/dK(K^* - \delta) < 0 < dJ_\varepsilon/dK(K^* + \delta)\),
and the unique minimizer \(K_\varepsilon^*\) lies strictly inside
\(\mathcal{N}_\delta\) with \(dJ_\varepsilon/dK(K_\varepsilon^*) = 0\).

(b) The rate \(|K_\varepsilon^* - K^*| = O(\varepsilon)\). We proceed in
two sub-steps.

(b.1) \(dJ_\varepsilon/dK(K^*) = O(\varepsilon)\). Since
\(dJ/dK(K^*) = 0\) in the principal-value sense, this amounts to
bounding the rate of joint continuity at \(\varepsilon = 0\) for the
specific point \(K^*\). We apply the parity-shell decomposition of
Appendix A to the integrand of \(dJ_\varepsilon/dK(K^*)\).

Substitute \(s = b - b_\star\) with
\(b_\star = b_{\rm sing}(K^*) = -1/K^*\) and set
\(f(s) := \rho(b_\star + s)\,(b_\star + s)\). Split the integration into
a near-region \(|s| \le R\) (with \(R \in (0, \tau)\) as in \cref{lem:uniform-PL-constant})
and its complement. On the complement, \(|1+bK^*| \ge |K^*|R\)
uniformly, so the integrand of \(dJ_\varepsilon/dK(K^*)\) on the
far-region satisfies
\begin{equation}
        \rho(b)\cdot b \cdot \frac{1+bK^*}{(1+bK^*)^2 + \varepsilon^2} - \rho(b)\cdot b \cdot \frac{1}{1+bK^*} = \frac{-\varepsilon^2 \rho(b)\cdot b}{(1+bK^*)\bigl[(1+bK^*)^2 + \varepsilon^2\bigr]},
\end{equation}
which is \(O(\varepsilon^2)\) pointwise with uniform integrable majorant
on the far-region; the far-region contribution to
\(dJ_\varepsilon/dK(K^*) - dJ/dK(K^*)\) is therefore
\(O(\varepsilon^2)\).

On the near-region, Taylor-expand \(f(s) = f(0) + s\,f'(0) + r(s)\) with
\(|r(s)| \le \tfrac{1}{2}\|f''\|_\infty\,s^2\). The three contributions
to \(\int_{-R}^{R} f(s)\cdot K^* s/((K^*)^2 s^2 + \varepsilon^2)\,ds\)
are:

--- \emph{Constant term.}
\(f(0)\,\int_{-R}^{R} K^* s/((K^*)^2 s^2 + \varepsilon^2)\,ds = 0\),
identically for every \(\varepsilon \ge 0\), by odd parity over
\([-R, R]\). (This is the same parity-vanishing as in the proof of \cref{lem:uniform-PL-constant}.)

--- \emph{Linear term.}
\begin{equation*}
f'(0)\int_{-R}^{R}\frac{K^* s^2}{(K^*)^2 s^2 + \varepsilon^2}\,ds
= \frac{f'(0)}{K^*}\Bigl[2R - \frac{2\varepsilon}{|K^*|}\arctan\frac{|K^*|R}{\varepsilon}\Bigr].
\end{equation*}
The $\varepsilon \downarrow 0$ limit is $2R f'(0)/K^*$. From
$\arctan(|K^*|R/\varepsilon) = \pi/2 - \arctan(\varepsilon/(|K^*|R)) = \pi/2 + O(\varepsilon)$,
the deviation from the limit is $O(\varepsilon)$.

--- \emph{Quadratic remainder.}
\begin{equation*}
\Bigl|\int_{-R}^{R}\frac{r(s)\, K^* s}{(K^*)^2 s^2 + \varepsilon^2}\,ds\Bigr|
\le \frac{\|f''\|_\infty}{2}\int_{-R}^{R}\frac{|s|^3}{(K^*)^2 s^2}\,ds
\le \frac{\|f''\|_\infty R^2}{2(K^*)^2},
\end{equation*}
uniformly in \(\varepsilon\). The deviation from the \(\varepsilon = 0\)
value is \(O(\varepsilon^2)\).

Summing the four contributions (far + three from near) and subtracting 
the principal-value identity $dJ/dK(K^*) = 0$:
\begin{equation}
    \Bigl|\frac{dJ_\varepsilon}{dK}(K^*)\Bigr| \le C_1\,\varepsilon + O(\varepsilon^2),
\end{equation}
where the leading-order constant
\begin{equation}
    C_1 = \frac{\pi\,|f'(0)|}{(K^*)^2} 
    = \frac{\pi\,|\rho'(b_\star)\,b_\star + \rho(b_\star)|}{(K^*)^2}
\end{equation}
is collected from the linear-term deviation 
$\arctan(|K^*|R/\varepsilon) = \pi/2 - \varepsilon/(|K^*|R) + O(\varepsilon^3)$ 
applied to the linear coefficient $f'(0) = \rho'(b_\star)b_\star + \rho(b_\star)$ 
(here $b_\star = -1/K^*$).

\emph{(b.2) Mean-value-theorem inversion.} By (a),
\(K_\varepsilon^* \in \mathcal{N}_\delta\) and
\(dJ_\varepsilon/dK(K_\varepsilon^*) = 0\). Apply the mean value theorem
to \(dJ_\varepsilon/dK\) on the segment from \(K^*\) to
\(K_\varepsilon^*\):
\begin{equation}
    0 = \frac{dJ_\varepsilon}{dK}(K_\varepsilon^*) = \frac{dJ_\varepsilon}{dK}(K^*) + \frac{d^2 J_\varepsilon}{dK^2}(\xi_\varepsilon)\cdot(K_\varepsilon^* - K^*),
\end{equation}
for some \(\xi_\varepsilon\) between \(K^*\) and \(K_\varepsilon^*\). By
\cref{lem:uniform-PL-constant}(c), \(d^2 J_\varepsilon/dK^2(\xi_\varepsilon) \ge \mu_0\), so
\begin{equation}
    \bigl|K_\varepsilon^* - K^*\bigr| = \frac{|dJ_\varepsilon/dK(K^*)|}{|d^2 J_\varepsilon/dK^2(\xi_\varepsilon)|} \le \frac{C_1\,\varepsilon + O(\varepsilon^2)}{\mu_0} \le C_K\,\varepsilon
\end{equation}
for all \(\varepsilon \in [0, \varepsilon_1]\), with \(\varepsilon_1\)
small enough that the \(O(\varepsilon^2)\) term is bounded by
\(C_1\varepsilon\), giving
\(|K_\varepsilon^* - K^*| \le 2 C_1 \varepsilon/\mu_0 =: C_K \varepsilon\).

\emph{(c) The objective gap is \(O(\varepsilon^2)\).} By Taylor
expansion of \(J_\varepsilon\) at the interior critical point
\(K_\varepsilon^*\) (where \(dJ_\varepsilon/dK(K_\varepsilon^*) = 0\)),
\cref{lem:uniform-PL-constant}(c) (\(d^2 J_\varepsilon/dK^2 \le L_0\) on
\(\mathcal{N}_\delta\)), and (b):
\begin{equation}
        0 \le J_\varepsilon(K^*) - J_\varepsilon(K_\varepsilon^*) = \tfrac{1}{2}\,\frac{d^2 J_\varepsilon}{dK^2}(\xi'_\varepsilon)\,(K^* - K_\varepsilon^*)^2 \le \tfrac{L_0}{2}(K_\varepsilon^* - K^*)^2 \le \tfrac{L_0\,C_K^2}{2}\,\varepsilon^2,
\end{equation}
for some \(\xi'_\varepsilon\) between \(K^*\) and \(K_\varepsilon^*\).
The lower bound
\(0 \le J_\varepsilon(K^*) - J_\varepsilon(K_\varepsilon^*)\) is by
definition of \(K_\varepsilon^*\) as minimizer.
\end{proof}

\subsection{\texorpdfstring{Implications: the uniform-in-$\varepsilon$ PL inequality}{Implications: the uniform-in-epsilon PL inequality}}\label{34-implications-uniform-in-epsilon-inequality}

The Polyak--{\L}ojasiewicz inequality~\cite{polyak1963gradient}---a gradient-domination condition under which first-order methods converge linearly~\cite{karimi2016linear}---now follows on $\mathcal{N}_\delta$ from the uniform Hessian lower bound of \cref{lem:uniform-PL-constant}.

\begin{theorem}[\texorpdfstring{uniform-in-\(\varepsilon\) Polyak--{\L}ojasiewicz inequality}{uniform-in-epsilon Polyak-Lojasiewicz inequality}]
\label{thm:uniform-epsilon-pl-inequality}Under the hypotheses of \cref{lem:uniform-PL-constant}, for all
\((K, \varepsilon) \in \mathcal{N}_\delta \times [0, \varepsilon_0]\)
after possibly shrinking \(\varepsilon_0\) as needed by \cref{lem:persistence-regularized-minimizer},
\begin{equation}\label{eq:PL_ineq}
    \tfrac{1}{2}\!\left(\frac{dJ_\varepsilon}{dK}(K)\right)^{\!2} \ge \mu_0\bigl(J_\varepsilon(K) - J_\varepsilon^*\bigr),
\end{equation}
where
\(J_\varepsilon^* := \min_{K \in \mathcal{N}_\delta} J_\varepsilon(K) = J_\varepsilon(K_\varepsilon^*)\).
\end{theorem}
\begin{proof}
By \cref{lem:persistence-regularized-minimizer}(a), \(J_\varepsilon\) has a unique
interior minimizer \(K_\varepsilon^*\) on \(\mathcal{N}_\delta\) for
every \(\varepsilon \in [0, \varepsilon_1]\). \cref{lem:uniform-PL-constant}(c) gives
\(d^2 J_\varepsilon/dK^2 \ge \mu_0\) on \(\mathcal{N}_\delta\), so a
Taylor expansion of \(J_\varepsilon\) at \(K_\varepsilon^*\) yields, for
every \(K \in \mathcal{N}_\delta\),
\begin{equation}
J_\varepsilon(K) - J_\varepsilon^* = \tfrac{1}{2}\,\frac{d^2 J_\varepsilon}{dK^2}(\xi)\,(K - K_\varepsilon^*)^2 \ge \tfrac{\mu_0}{2}(K - K_\varepsilon^*)^2
\end{equation}
for some \(\xi\) between \(K\) and \(K_\varepsilon^*\). By the mean
value theorem applied to \(dJ_\varepsilon/dK\) between
\(K_\varepsilon^*\) (where it vanishes) and \(K\),
\(|dJ_\varepsilon/dK(K)| \ge \mu_0\,|K - K_\varepsilon^*|\).
Combining,
\begin{equation}
    \tfrac{1}{2}\Bigl(\frac{dJ_\varepsilon}{dK}(K)\Bigr)^{\!2} \ge \tfrac{\mu_0^2}{2}(K - K_\varepsilon^*)^2 \ge \mu_0\bigl(J_\varepsilon(K) - J_\varepsilon^*\bigr),
\end{equation}
which is \eqref{eq:PL_ineq}. 
\end{proof}

\section{Finite-sample analysis}\label{4-finite-sample-analysis}

This section develops a single-sample estimator for
\(dJ_\varepsilon/dK\) with uniform-in-\(\varepsilon\) variance on
\(\mathcal{N}_\delta\).

\subsection{The naive gradient estimator and variance divergence}\label{41-the-naive-gradient-estimator-and-variance-divergence}

Given a sample \(B \sim \rho\), the natural single-sample estimator of
\(dJ_\varepsilon/dK\) at gain \(K\) is
\begin{equation}
    \psi(B; K, \varepsilon) := \frac{B(1 + BK)}{(1 + BK)^2 + \varepsilon^2},
\end{equation}
so that \(\mathbb{E}[\psi(B; K, \varepsilon)] = dJ_\varepsilon/dK(K)\)
by comparison with the expression for \(dJ_\varepsilon/dK\) in \cref{prop:properties-J-epsilon}(a).

As noted in section~\ref{11-problem-and-motivation}, each transition $(X_t, X_{t+1})$ with $X_t \ne 0$ reveals an exact i.i.d.\ sample $B_t = (X_{t+1}/(aX_t)-1)/K \sim \rho$. In what follows we treat the $B_t$ as the directly observed sample stream; the policy-gradient oracle uses these samples without further reference to the state trajectory.

\begin{theorem}[variance scaling of the naive estimator]
\label{thm:variance-scaling-naive-estimator}Under \cref{asm:compact_noise,asm:regularity}, for $K \in \mathcal{K}_{\mathrm{stab}}$ with
$b_{\mathrm{sing}}(K) \in (b_{\min}, b_{\max})$,
\begin{equation}\label{eq:variance_scaling}
    \mathrm{Var}\bigl[\psi(B; K, \varepsilon)\bigr]
    = \frac{\pi\,\rho(b_{\mathrm{sing}}(K))}{2|K|^3}\cdot\frac{1}{\varepsilon}
    + R(K) + o(1), \qquad \varepsilon \downarrow 0,
\end{equation}
where, with $s_\pm := b_\pm - b_{\mathrm{sing}}(K)$, the $\varepsilon$-independent remainder admits the closed form
\begin{equation}\label{eq:R-form}
R(K) = \frac{\mathrm{sgn}(K)\,f_1(K)}{K^2}\,\log\!\left|\frac{s_+}{s_-}\right| + O(1),
\end{equation}
where $f_1(K) := \rho'(b_{\mathrm{sing}}(K))\,b_{\mathrm{sing}}(K)^2 + 2\,b_{\mathrm{sing}}(K)\,\rho(b_{\mathrm{sing}}(K))$
(proved in \cref{b2-remainder-in-the-variance-scaling-of-theorem-41}). In particular, the $\log|s_+/s_-|$ factor captures the asymmetry of the support around $b_{\mathrm{sing}}(K)$: this leading remainder vanishes in the symmetric case $b_{\mathrm{sing}}(K) = (b_{\min} + b_{\max})/2$, and $|R(K)| \to \infty$ logarithmically as $b_{\mathrm{sing}}(K)$ approaches a support endpoint. On any compact $\mathcal{N}_\delta \subset \mathcal{K}_{\rm stab}$ bounded away from the channel-edge gains $\{-1/b_{\min}, -1/b_{\max}\}$ (\cref{lem:uniform-PL-constant}(a)), $\sup_{K\in\mathcal{N}_\delta}|R(K)| < \infty$.
\end{theorem}
\begin{proof}
Since
\(\mathbb{E}[\psi(B; K, \varepsilon)] = dJ_\varepsilon/dK(K)\), which is
jointly continuous on \(\mathcal{N}_\delta \times [0, \varepsilon_0]\)
(\cref{thm:hessian-decomposition} plus differentiability under the integral) and hence
\(\mathbb{E}[\psi(B; K, \varepsilon)]^2 = O(1)\) uniformly, the variance
satisfies
\begin{equation}
        \mathrm{Var}[\psi(B; K, \varepsilon)] = \mathbb{E}[\psi(B; K, \varepsilon)^2] - \mathbb{E}[\psi(B; K, \varepsilon)]^2 = \mathbb{E}[\psi^2] + O(1).
\end{equation}

Computing \(\mathbb{E}[\psi^2]\) directly:
\begin{equation}
    \mathbb{E}[\psi^2] = \int_{b_{\min}}^{b_{\max}} \rho(b)\cdot\frac{b^2(1+bK)^2}{[(1+bK)^2 + \varepsilon^2]^2}\,db.
\end{equation}

Substitute \(s := b - b_{\mathrm{sing}}(K)\), so \(1 + bK = Ks\) and
\(b = b_{\mathrm{sing}} + s\):
\begin{equation}
    \mathbb{E}[\psi^2] = \int_{s_-}^{s_+} \rho(b_{\mathrm{sing}} + s)\cdot (b_{\mathrm{sing}} + s)^2\cdot\frac{K^2 s^2}{(K^2 s^2 + \varepsilon^2)^2}\,ds,
\end{equation}
with integration limits \(s_\pm = b_\pm - b_{\mathrm{sing}}(K)\). The
kernel \(K^2 s^2/(K^2 s^2 + \varepsilon^2)^2\) is a standard approximate
delta: substituting \(t := Ks/\varepsilon\),
\(ds = \varepsilon dt/|K|\),
\begin{equation}
    \int_{-\infty}^{\infty} \frac{K^2 s^2}{(K^2 s^2 + \varepsilon^2)^2}\,ds = \frac{1}{\varepsilon|K|}\int_{-\infty}^{\infty}\frac{t^2}{(t^2 + 1)^2}\,dt = \frac{1}{\varepsilon|K|}\cdot\frac{\pi}{2}.
\end{equation}

The full integral differs from this by an \(O(1)\) correction from the
finite integration limits and an \(O(1)\) correction from the Taylor
expansion of the slowly-varying factor
\(\rho(b_{\mathrm{sing}}+s)(b_{\mathrm{sing}}+s)^2\), whose
linear-in-\(s\) term integrates to a constant by parity (Appendix B.2).
By continuity of \(\rho\) at
\(b_{\mathrm{sing}} \in (b_{\min}, b_{\max})\), the leading value of the
slowly-varying factor at \(s = 0\) is
\(\rho(b_{\mathrm{sing}}(K))\cdot b_{\mathrm{sing}}(K)^2\), giving
\begin{equation}\label{eq:Epsi2_rho}
    \mathbb{E}[\psi^2] = \rho(b_{\mathrm{sing}}(K))\cdot b_{\mathrm{sing}}(K)^2 \cdot \frac{\pi}{2|K|\,\varepsilon} + O(1).
\end{equation}

Using \(b_{\mathrm{sing}}(K) = -1/K\) so
\(b_{\mathrm{sing}}(K)^2 = 1/K^2\):
\begin{equation}
    \mathbb{E}[\psi^2] = \frac{\pi\,\rho(b_{\mathrm{sing}}(K))}{2|K|^3\,\varepsilon} + O(1),
\end{equation}
which gives \eqref{eq:variance_scaling}. The \(O(1)\) remainder is derived in
Appendix~B.2 and, in full, in section~SM4 of the supplement.
\end{proof}
\emph{Consequence.} The \(\Theta(1/\varepsilon)\) scaling implies that
mini-batch SGD using \(\psi\) requires batch size
\(N = \Omega(1/\eta^2)\) to control the noise floor at accuracy \(\eta\)
when \(\varepsilon = \eta\) --- a factor of \(1/\eta\) worse than the
bound established in section~\ref{5-sample-complexity-in-the-pl-basin} using $\tilde\psi$.

\subsection{The density-aware symmetric-pairing estimator}\label{42-the-density-aware-symmetric-pairing-estimator}

We now exploit the same odd-parity structure of the Cauchy kernel around
\(b_{\mathrm{sing}}(K)\) that underpinned the uniform-in-\(\varepsilon\)
Hessian bound of \cref{lem:uniform-PL-constant}, this time at the level of a single random
sample. The construction extends antithetic sampling \cite{hammersley1956new, owen2013monte} to a
moving singularity weighted by an arbitrary density: we pair each draw
\(B\) with its reflection \(2 b_{\rm sing}(K) - B\) through the moving
singularity, weighted by the density \(\rho\) to maintain unbiasedness.
In the constant-density case the construction reduces to standard
antithetic sampling about a fixed center; the novelty here is that the
antithetic point is \emph{the moving pole of the cusp}, and the
density-weighted average is calibrated to exploit precisely the parity
that drives the population-level analysis.

\begin{definition}[density-aware symmetric-pairing estimator]
\label{def:symmetric-pairing-estimator}Under \cref{asm:compact_noise,asm:regularity}, define $\tilde\psi$
in two cases.

\emph{Case 1: pole in support} ($b_{\mathrm{sing}}(K) \in (b_{\min}, b_{\max})$). Let
\begin{equation*}
\delta_K := \min\bigl(b_{\mathrm{sing}}(K) - b_{\min},\; b_{\max} - b_{\mathrm{sing}}(K)\bigr),
\;
\mathcal{S}_K := [b_{\mathrm{sing}}(K) - \delta_K,\, b_{\mathrm{sing}}(K) + \delta_K],
\end{equation*}
$\mathcal{A}_K := [b_{\min}, b_{\max}]\setminus\mathcal{S}_K$, and $\bar B := 2b_{\mathrm{sing}}(K) - B$. Define
\begin{equation}
    \tilde\psi(B; K, \varepsilon) := 
    \begin{cases} 
        \dfrac{\rho(B)\psi(B; K, \varepsilon) + \rho(\bar B)\psi(\bar B; K, \varepsilon)}{\rho(B) + \rho(\bar B)} & \text{if } B \in \mathcal{S}_K,\\[2ex]
        \psi(B; K, \varepsilon) & \text{if } B \in \mathcal{A}_K. 
    \end{cases}
\end{equation}

\emph{Case 2: pole outside support} ($b_{\mathrm{sing}}(K) \notin (b_{\min}, b_{\max})$). Define $\tilde\psi(B; K, \varepsilon) := \psi(B; K, \varepsilon)$ for all $B \in [b_{\min}, b_{\max}]$.
\end{definition}
In Case 2, the cusp obstruction is absent --- \(1 + bK\) maintains one
sign over the support, so \(\psi\) has bounded variance uniformly in
\(\varepsilon\) (no symmetric pairing is needed). Unbiasedness
\(\mathbb{E}[\tilde\psi] = dJ_\varepsilon/dK\) holds in both cases.

When \(\rho\) is unknown, \cref{def:symmetric-pairing-estimator} is replaced by its density-unknown
counterpart (\cref{def:plug-in-symmetric-pairing-estimator} below), in which \(\rho\) is replaced by a
kernel density estimate and the pairing radius \(\delta_K\) is replaced
by an algorithmic parameter \(R \in (0, \tau/2]\).

The pairing construction is the sample-level analog of the parity
cancellation that drives \cref{lem:uniform-PL-constant} and \cref{prop:properties-J-epsilon}(c) at the
population level. To see this, note that for \(B = b_{\rm sing}(K) + s\)
on \(\mathcal{S}_K\):
\begin{equation}
        \rho(B)\,\psi(B) + \rho(2b_{\rm sing} - B)\,\psi(2b_{\rm sing} - B) = \frac{Ks\cdot\bigl[h(b_{\rm sing}+s) - h(b_{\rm sing}-s)\bigr]}{K^2 s^2 + \varepsilon^2},
\end{equation}
where \(h(b) := b\rho(b)\). The bracketed difference is \(O(s)\) by the
mean value theorem, so the displayed quantity is
\(O(s^2/(K^2 s^2 + \varepsilon^2)) = O(1)\) uniformly in \(\varepsilon\)
--- the explicit factor of \(s\) in the numerator combines with the
\(O(s)\) from the difference to produce an \(s^2\) that exactly
neutralizes the \(\varepsilon^{-1}\) concentration of the Cauchy kernel.
The same parity that kills the leading Hessian piece (proof of \cref{lem:uniform-PL-constant}) and the leading bias remainder (proof of \cref{prop:properties-J-epsilon}(c)) here
kills the variance divergence at the sample level. \cref{prop:unbiased-variance-psi-tilde} makes
this rigorous.

\begin{remark}[pairing radius constraint]
The pairing radius $\delta_K$ in \cref{def:symmetric-pairing-estimator} is the maximal radius 
for which the symmetric-pairing reflection $2 b_{\rm sing}(K) - B$ 
maps the symmetric zone $\mathcal{S}_K$ into 
$[b_{\min}, b_{\max}]$. By \cref{lem:uniform-PL-constant}(a), 
$\delta_K \ge \tau$ for all $K \in \mathcal{N}_\delta$, so the 
analytical pairing radius $R := \tau/2$ used in the parity-shell 
decompositions of section~3 satisfies $R \le \delta_K$ uniformly on 
$\mathcal{N}_\delta$. The plug-in version (\cref{def:plug-in-symmetric-pairing-estimator}) treats 
$R$ as an algorithmic parameter $R \in (0, \tau/2]$ subject to a 
bias-variance tradeoff (Step 6 of \cref{thm:total-sample-complexity-pl-basin}).
\end{remark}

\begin{proposition}[unbiasedness and bounded variance of \(\tilde\psi\)]
\label{prop:unbiased-variance-psi-tilde}Under \cref{asm:compact_noise,asm:regularity}, for
\(K \in \mathcal{K}_{\mathrm{stab}}\) with
\(b_{\mathrm{sing}}(K) \in (b_{\min}, b_{\max})\) and any
\(\varepsilon \in (0, 1]\):

(a) \emph{(Unbiasedness.)}
\(\mathbb{E}[\tilde\psi(B; K, \varepsilon)] = dJ_\varepsilon/dK(K)\).

(b) \emph{(Closed-form on the symmetric zone.)} For
\(B = b_{\mathrm{sing}}(K) + s\) with \(|s| \le \delta_K\),
\begin{equation}\label{eq:tildepsi_s}
        \tilde\psi(b_{\mathrm{sing}}(K) + s; K, \varepsilon) = \frac{K\,s\,[\,h(b_{\mathrm{sing}}+s) - h(b_{\mathrm{sing}}-s)\,]}{(K^2 s^2 + \varepsilon^2)\,[\rho(b_{\mathrm{sing}}+s) + \rho(b_{\mathrm{sing}}-s)]}.
\end{equation}

(c) \emph{(Uniform-in-\(\epsilon\) pointwise bound.)} For \(B \in \mathcal{S}_K\),
writing \(h(b) := b\rho(b)\):
\begin{equation}\label{eq:tildepsi_bound}
    |\tilde\psi(B; K, \varepsilon)| \le \frac{\|h'\|_\infty}{|K|\,\rho_{\min}(K)} =: M_{\mathcal{S}}(K),
\end{equation}
where
\(\rho_{\min}(K) := \min_{|s| \le \delta_K}\tfrac{1}{2}[\rho(b_{\mathrm{sing}}+s) + \rho(b_{\mathrm{sing}}-s)] > 0\)
under \cref{asm:regularity}. Consequently \(M_{\mathcal{S}}(K) < \infty\) and
does not depend on \(\varepsilon\).

(d) \emph{(Uniform-in-\(\epsilon\) variance bound.)}
\(\mathrm{Var}[\tilde\psi(B; K, \varepsilon)] \le \sigma_\star^2(K) < \infty\)
uniformly in \(\varepsilon \in (0, 1]\), where
\begin{equation}\label{eq:sigma2_bound}
    \sigma_\star^2(K) \le M_{\mathcal{S}}(K)^2\,\Pr(B \in \mathcal{S}_K) + \sup_{b \in \mathcal{A}_K}\!\left|\frac{b}{1+bK}\right|^2.
\end{equation}

When \(b_{\mathrm{sing}}(K) \notin (b_{\min}, b_{\max})\) (\cref{def:symmetric-pairing-estimator}, Case 2), \(\tilde\psi \equiv \psi\) is unbiased with variance
bounded by \(\sup_{b \in [b_{\min}, b_{\max}]}|b/(1+bK)|^2\), uniformly
in \(\varepsilon \in (0, 1]\).

(e) \emph{(Variance bound uniform on \(\mathcal{N}_\delta\).)} For any
\(\delta > 0\) such that \(\mathcal{N}_\delta\) satisfies the hypotheses
of \cref{lem:uniform-PL-constant},
\(\sup_{K \in \mathcal{N}_\delta} \sigma_\star^2(K) =: \sigma_\star^2(\mathcal{N}_\delta) < \infty\).
\end{proposition}
\begin{proof}
\emph{(a) Unbiasedness.} Split the expectation:
\begin{equation}
\begin{aligned}
        \mathbb{E}[\tilde\psi(B)] = \int_{\mathcal{S}_K}\frac{\rho(b)\psi(b) + \rho(2b_{\mathrm{sing}}-b)\psi(2b_{\mathrm{sing}}-b)}{\rho(b) + \rho(2b_{\mathrm{sing}}-b)}\cdot\rho(b)\,db + \int_{\mathcal{A}_K}\psi(b)\rho(b)\,db.
\end{aligned}
\end{equation}

Denote the first integrand\textquotesingle s ratio by
\(N(b) := [\rho(b)\psi(b) + \rho(2b_{\mathrm{sing}}-b)\psi(2b_{\mathrm{sing}}-b)]/[\rho(b) + \rho(2b_{\mathrm{sing}}-b)]\).
Substitute \(b' := 2b_{\mathrm{sing}} - b\) on \(\mathcal{S}_K\): this
involution has Jacobian \(1\) and maps \(\mathcal{S}_K\) onto itself.
The symmetric roles of \(b\) and \(2b_{\mathrm{sing}}-b\) in \(N\) give
\(N(b') = N(b)\) (the numerator and denominator of \(N\) are both
symmetric under \(b \leftrightarrow 2b_{\mathrm{sing}}-b\)). Hence
\begin{equation}
    \int_{\mathcal{S}_K} N(b)\rho(b)\,db = \int_{\mathcal{S}_K} N(b)\rho(2b_{\mathrm{sing}}-b)\,db,
\end{equation}
so
\begin{equation}
\begin{aligned}
    \int_{\mathcal{S}_K} N(b)\rho(b)\,db &= \tfrac{1}{2}\int_{\mathcal{S}_K} N(b)\bigl[\rho(b) + \rho(2b_{\mathrm{sing}}-b)\bigr]\,db \\&= \tfrac{1}{2}\int_{\mathcal{S}_K}\bigl[\rho(b)\psi(b) + \rho(2b_{\mathrm{sing}}-b)\psi(2b_{\mathrm{sing}}-b)\bigr]\,db.
\end{aligned}
\end{equation}

Splitting and applying the \(b \to 2b_{\mathrm{sing}}-b\) substitution
once more:
\begin{equation}
    = \tfrac{1}{2}\int_{\mathcal{S}_K}\rho(b)\psi(b)\,db + \tfrac{1}{2}\int_{\mathcal{S}_K}\rho(b')\psi(b')\,db' = \int_{\mathcal{S}_K}\rho(b)\psi(b)\,db.
\end{equation}

Combining with the \(\mathcal{A}_K\) piece gives
\begin{equation}
    \mathbb{E}[\tilde\psi(B)] = \int_{[b_{\min}, b_{\max}]}\psi(b)\rho(b)\,db = \mathbb{E}[\psi(B)] = dJ_\varepsilon/dK(K).
\end{equation}

\emph{(b) Closed form on \(\mathcal{S}_K\).} For
\(B = b_{\mathrm{sing}} + s\) with \(|s| \le \delta_K\):
\begin{equation}
        \psi(b_{\mathrm{sing}} + s; K, \varepsilon) = \frac{(b_{\mathrm{sing}}+s)\cdot(Ks)}{K^2 s^2 + \varepsilon^2},\\ \psi(b_{\mathrm{sing}} - s; K, \varepsilon) = \frac{(b_{\mathrm{sing}}-s)\cdot(-Ks)}{K^2 s^2 + \varepsilon^2},
\end{equation}

(using \(1 + (b_{\mathrm{sing}} \pm s)K = \pm Ks\)). Multiplying by
\(\rho\) and summing:
\begin{equation}
\begin{aligned}
        &\rho(b_{\mathrm{sing}}+s)\psi(b_{\mathrm{sing}}+s) + \rho(b_{\mathrm{sing}}-s)\psi(b_{\mathrm{sing}}-s) \\=& \frac{Ks}{K^2 s^2 + \varepsilon^2}\Bigl[(b_{\mathrm{sing}}+s)\rho(b_{\mathrm{sing}}+s) - (b_{\mathrm{sing}}-s)\rho(b_{\mathrm{sing}}-s)\Bigr].
\end{aligned}
\end{equation}

Recognizing the bracketed term as
\(h(b_{\mathrm{sing}}+s) - h(b_{\mathrm{sing}}-s)\) with
\(h(b) := b\rho(b)\):
\begin{equation}
    = \frac{K\,s\,[\,h(b_{\mathrm{sing}}+s) - h(b_{\mathrm{sing}}-s)\,]}{K^2 s^2 + \varepsilon^2}.
\end{equation}

Dividing by \(\rho(b_{\mathrm{sing}}+s) + \rho(b_{\mathrm{sing}}-s)\)
gives \eqref{eq:tildepsi_s}.

\emph{(c) Pointwise bound on \(\mathcal{S}_K\).} From \eqref{eq:tildepsi_s}:
\begin{equation}
    |\tilde\psi(b_{\mathrm{sing}}+s; K, \varepsilon)| = \frac{|K s|\cdot|h(b_{\mathrm{sing}}+s) - h(b_{\mathrm{sing}}-s)|}{(K^2 s^2 + \varepsilon^2)\cdot[\rho(b_{\mathrm{sing}}+s) + \rho(b_{\mathrm{sing}}-s)]}.
\end{equation}

For the difference
\(|h(b_{\mathrm{sing}}+s) - h(b_{\mathrm{sing}}-s)|\), by the mean value
theorem applied to \(h \in C^1\), there exists
\(\xi \in (b_{\mathrm{sing}} - s, b_{\mathrm{sing}} + s)\) with
\(h(b_{\mathrm{sing}}+s) - h(b_{\mathrm{sing}}-s) = 2s\cdot h'(\xi)\),
so
\(|h(b_{\mathrm{sing}}+s) - h(b_{\mathrm{sing}}-s)| \le 2|s|\cdot\|h'\|_\infty\).

Substituting (using \(|Ks|\cdot 2|s| = 2|K|s^2\)):
\begin{equation}
        |\tilde\psi(b_{\mathrm{sing}}+s; K, \varepsilon)| \le \frac{2|K| s^2\cdot\|h'\|_\infty}{(K^2 s^2 + \varepsilon^2)\cdot 2\rho_{\min}(K)} = \frac{|K| s^2}{K^2 s^2 + \varepsilon^2}\cdot\frac{\|h'\|_\infty}{\rho_{\min}(K)} \le \frac{\|h'\|_\infty}{|K|\,\rho_{\min}(K)},
\end{equation}
since
\(|K| s^2/(K^2 s^2 + \varepsilon^2) \le |K| s^2/(K^2 s^2) = 1/|K|\).
This establishes \eqref{eq:tildepsi_bound} with
\(M_{\mathcal{S}}(K) = \|h'\|_\infty/(|K|\,\rho_{\min}(K))\) --- and the
bound is independent of \(\varepsilon\). The assertion
\(M_{\mathcal{S}}(K) < \infty\) follows from
\(h' = \rho + b\rho' \in C^0\) on \([b_{\min}, b_{\max}]\) (bounded on a
compact interval) and \(\rho_{\min}(K) > 0\) by \cref{asm:regularity} and
compactness.

\emph{(d) Variance bound.} Decompose:
\begin{equation}
        \mathrm{Var}[\tilde\psi(B; K, \varepsilon)] \le \mathbb{E}[\tilde\psi(B; K, \varepsilon)^2] = \int_{\mathcal{S}_K}\tilde\psi^2(b)\rho(b)\,db + \int_{\mathcal{A}_K}\psi^2(b)\rho(b)\,db.
\end{equation}

On \(\mathcal{S}_K\), by \eqref{eq:tildepsi_bound} \(\tilde\psi^2 \le M_{\mathcal{S}}(K)^2\),
so the first integral is bounded by
\(M_{\mathcal{S}}(K)^2 \Pr(B \in \mathcal{S}_K)\).

On \(\mathcal{A}_K\), the pole \(b_{\mathrm{sing}}(K)\) is not in the
integration domain:
\(\inf_{b \in \mathcal{A}_K}|1 + bK| \ge |K|\cdot\delta_K > 0\). Hence
\begin{equation}
    |\psi(b; K, \varepsilon)| = \left|\frac{b(1+bK)}{(1+bK)^2 + \varepsilon^2}\right| \le \frac{|b|}{|1+bK|} \le \sup_{b \in \mathcal{A}_K}\left|\frac{b}{1+bK}\right|,
\end{equation}
uniformly in \(\varepsilon\). Squaring and taking the bound in the
second integral yields \eqref{eq:sigma2_bound}.

(e) \(M_\mathcal{S}(K)\) and \(\sup_{b \in \mathcal{A}_K}|b/(1+bK)|\)
are continuous on the compact \(\mathcal{N}_\delta\).
\end{proof}
The original symmetric-pairing estimator (\cref{def:symmetric-pairing-estimator}) requires
evaluation of \(\rho\) at the moving reflection point
\(2 b_{\rm sing}(K) - B\). In the density-unknown setting, this evaluation is unavailable. We replace it by a kernel density estimate built from a held-out portion of the sample stream.

\subsection{Kernel density estimate}\label{43-kernel-density-estimate}
For the kernel density rate of \cref{lem:uniform-kde-rates} we strengthen \cref{asm:regularity} to \(\rho \in C^s\) for an integer \(s \ge 2\); this generalized assumption is denoted \cref{asm:regularity}(s) and includes the original \cref{asm:compact_noise,asm:regularity} as the case \(s = 2\). Fix an integer \(s \geq 2\) (the smoothness exponent of \cref{asm:regularity}(\(s\))) and a
kernel \(\kappa: \mathbb{R} \to \mathbb{R}\) with the following
properties: \(\kappa \in C^1(\mathbb{R})\),
\(\mathrm{supp}(\kappa) \subset [-1,1]\), \(\int \kappa = 1\),
and \(\int t^j \kappa(t)\,dt = 0\) for \(j = 1, \ldots, s-1\) (an
\emph{order-\(s\)} kernel, \cite{Giné_Nickl_2021}). Given a sample
\(B_1, \ldots, B_{n_1} \sim \rho\) i.i.d. and a bandwidth \(h > 0\), the
kernel density estimate is
\begin{equation}
    \hat\rho_{n_1, h}(b) := \frac{1}{n_1 h}\sum_{i=1}^{n_1} \kappa\!\left(\frac{b - B_i}{h}\right), \qquad b \in [b_{\min}, b_{\max}].
\end{equation}

The next lemma collects the rates we need. They are standard and proved
in Tsybakov \cite{Giné_Nickl_2021}.

\begin{lemma}[uniform KDE rates on a compact interior subinterval]
\label{lem:uniform-kde-rates}Assume \cref{asm:regularity}($s$) with $s \ge 2$. There exist constants 
$c_h, C_\nu, C_{\nu'} > 0$ and a universal constant $c > 0$, depending 
only on $s$, $\|\rho\|_{C^s}$, $\mathrm{supp}(\kappa)$, and 
$\mathrm{dist}(I, \{b_{\min}, b_{\max}\})$, such that for the 
bandwidth choice $h = h_{n_1} := c_h (\log n_1 / n_1)^{1/(2s+1)}$ 
and any compact subinterval $I \subset (b_{\min}, b_{\max})$:
\begin{align}
    \nu_{n_1} := \sup_{b \in I} \bigl|\hat\rho_{n_1, h_{n_1}}(b) - \rho(b)\bigr| 
    &\le C_\nu \left(\frac{\log n_1}{n_1}\right)^{s/(2s+1)}, \\
    \nu_{n_1}' := \sup_{b \in I} \bigl|\hat\rho_{n_1, h_{n_1}}'(b) - \rho'(b)\bigr| 
    &\le C_{\nu'} \left(\frac{\log n_1}{n_1}\right)^{(s-1)/(2s+1)},
\end{align}
\emph{both with probability at least $1 - n_1^{-c}$.
The hypothesis $s \ge 2$ ensures that the derivative-rate exponent 
$(s-1)/(2s+1)$ is positive, so $\nu_{n_1}' \to 0$ as $n_1 \to \infty$. 
For the minimal smoothness $s = 2$, the rates are 
$\nu_{n_1} = O((\log n_1/n_1)^{2/5})$ and 
$\nu_{n_1}' = O((\log n_1/n_1)^{1/5})$. Higher $s$ improves both 
rates: at $s = 4$, $\nu_{n_1} = O((\log n_1/n_1)^{4/9})$ and 
$\nu_{n_1}' = O((\log n_1/n_1)^{3/9})$.}
\end{lemma}
Since $s/(2s+1) > (s-1)/(2s+1)$, the function rate is faster than the 
derivative rate; in particular $\nu_{n_1}/\nu_{n_1}' \to 0$ as $n_1 \to \infty$, 
which we use in the proof of \cref{lem:parity-weight-discrepancy}.

We will apply this with \(I = [b_{\min} + \tau/2, b_{\max} - \tau/2]\),
where \(\tau\) is the constant of \cref{lem:uniform-PL-constant}(a). By \cref{lem:uniform-PL-constant}(a),
\(b_{\rm sing}(K) \in [b_{\min} + \tau, b_{\max} - \tau]\) for all
$K \in \mathcal{N}_\delta$, and (as we choose $R$ small in section~\ref{5-sample-complexity-in-the-pl-basin}) the
reflection point \(2 b_{\rm sing}(K) - B\) for
\(B \in [b_{\rm sing}(K) - R, b_{\rm sing}(K) + R]\) also lies in \(I\)
once \(R \leq \tau/2\).

\subsection{The plug-in estimator with adjustable pairing radius}\label{44-the-plug-in-estimator-with-adjustable-pairing-radius}

The original \cref{def:symmetric-pairing-estimator} implicitly used pairing radius
\(\delta_K = \min(b_{\rm sing}(K) - b_{\min}, b_{\max} - b_{\rm sing}(K))\).
We now treat the pairing radius as an algorithmic parameter
\(R \in (0, \tau/2]\).

\begin{definition}[plug-in symmetric-pairing estimator]
\label{def:plug-in-symmetric-pairing-estimator}For
\(K \in \mathcal{N}_\delta\), \(\varepsilon > 0\),
\(R \in (0, \tau/2]\), and a density estimate \(\hat\rho\). 
Letting \(\bar B := 2b_{\rm sing}(K) - B\) denote the symmetric point, we define:
\begin{equation}
        \tilde\psi_{\hat\rho, R}(B; K, \varepsilon) := 
    \begin{cases}
        \dfrac{\hat\rho(B)\psi(B; K, \varepsilon) + \hat\rho(\bar B)\psi(\bar B; K, \varepsilon)}{\hat\rho(B) + \hat\rho(\bar B)}, & B \in S_K(R), \\[2ex]
        \psi(B; K, \varepsilon), & B \in A_K(R).
    \end{cases}
\end{equation}
where \(S_K(R) := [b_{\rm sing}(K) - R, b_{\rm sing}(K) + R]\) and
\(A_K(R) := [b_{\min}, b_{\max}] \setminus S_K(R)\).
\end{definition}
When \(\hat\rho = \rho\) and \(R = \delta_K\), this recovers the
original \cref{def:symmetric-pairing-estimator}.

\subsection{The bias decomposition}\label{45-the-bias-decomposition}

Define the weighting function
\begin{equation}
        w_\rho(b) := \frac{\rho(b)}{\rho(b) + \rho(2b_{\rm sing}(K) - b)}, w_{\hat\rho}(b) := \frac{\hat\rho(b)}{\hat\rho(b) + \hat\rho(2b_{\rm sing}(K) - b)},
\end{equation}
and write \(\Delta := w_{\hat\rho} - w_\rho\), both viewed as functions
on \(S_K(R)\). The next lemma is the structural cornerstone of the
analysis.

\begin{lemma}[parity of the weight discrepancy]
\label{lem:parity-weight-discrepancy}Under \cref{asm:regularity} (no additional smoothness needed for the parity claim), the function \(\Delta\) is exactly odd under reflection through \(b_{\rm sing}(K)\):
\begin{equation}
    \Delta(2 b_{\rm sing}(K) - b) = -\Delta(b) \qquad \text{for all } b \in S_K(R).
\end{equation}
\end{lemma}
\begin{proof}
Direct computation. Setting
\(\bar b := 2b_{\rm sing}(K) - b\), the numerator of \(\Delta\) is
\begin{equation}
    \hat\rho(b) [\rho(b) + \rho(\bar b)] - \rho(b) [\hat\rho(b) + \hat\rho(\bar b)] = \hat\rho(b)\rho(\bar b) - \rho(b)\hat\rho(\bar b),
\end{equation}
which is antisymmetric under \(b \leftrightarrow \bar b\). The
denominator \([\hat\rho(b) + \hat\rho(\bar b)][\rho(b) + \rho(\bar b)]\)
is symmetric. Hence \(\Delta\) is odd.
\end{proof}

\begin{lemma}[pointwise bound on \(\Delta\)]
\label{lem:pointwise-bound-delta}Under \cref{asm:regularity}(\(s\)) with \(s \geq 2\), on the event of \cref{lem:uniform-kde-rates} with
\(\nu_{n_1} \leq \rho_{\min}/4\), we have
\(\hat\rho_{\min}(K) \geq \rho_{\min}(K)/2 > 0\) uniformly on
\(\mathcal{N}_\delta\), and
\begin{equation}
    |\Delta(b_{\rm sing}(K) + s)| \leq C_\Delta \cdot \nu_{n_1}' \cdot |s| \qquad \text{for all } |s| \leq R,
\end{equation}
\emph{where \(C_\Delta = C_\Delta(\rho, \tau)\) depends only on
\(\|\rho\|_{C^1(I)}\),
and \(\tau\).}
\end{lemma}
\begin{proof}
Write \(f := \hat\rho - \rho\), so \(|f| \leq \nu_{n_1}\)
and \(|f'| \leq \nu_{n_1}'\) on \(I\). The numerator of \(\Delta\),
after cancellation, is
\begin{equation}
    N(b) := \hat\rho(b)\rho(\bar b) - \hat\rho(\bar b)\rho(b) = f(b)\rho(\bar b) - f(\bar b)\rho(b),
\end{equation}
where the second equality uses that
\(\rho(b)\rho(\bar b) - \rho(\bar b)\rho(b) = 0\). Note
\(N(b_{\rm sing}(K)) = 0\) (since \(b = \bar b\) at the singularity), so
\(N\) vanishes at the center of \(S_K(R)\).

By the mean value theorem applied to \(N\) between \(b_{\rm sing}(K)\)
and \(b_{\rm sing}(K) + s\),
\(N(b_{\rm sing}(K) + s) = N'(\xi) \cdot s\) for some \(\xi\) between.
Computing \(N'\):
\begin{equation}
    N'(b) = f'(b)\rho(\bar b) - f(b)\rho'(\bar b) + f'(\bar b)\rho(b) - f(\bar b)\rho'(b)
\end{equation}
(using \(d\bar b/db = -1\)). Hence
\(|N'(\xi)| \leq 2 \nu_{n_1}'\|\rho\|_{\infty, I} + 2\nu_{n_1}\|\rho'\|_{\infty, I}\).
From \cref{lem:uniform-kde-rates}, \(\nu_{n_1} = O((\log n_1/n_1)^{s/(2s+1)})\) and
\(\nu'_{n_1} = O((\log n_1/n_1)^{(s-1)/(2s+1)})\), so
\(\nu_{n_1}/\nu'_{n_1} = O((\log n_1/n_1)^{1/(2s+1)}) \to 0\) as
\(n_1\to\infty\). Hence for \(n_1\) exceeding a threshold $n_0 = n_0(s, C_\nu, C_{\nu'})$ depending on the smoothness exponent and the KDE constants of \cref{lem:uniform-kde-rates}, \(\nu_{n_1}\le\nu'_{n_1}\), and the bound
\(|N'(\xi)|\le 2\nu'_{n_1}\|\rho\|_{\infty,I} + 2\nu_{n_1}\|\rho'\|_{\infty,I} \le 2(\|\rho\|_{\infty,I}+\|\rho'\|_{\infty,I})\,\nu'_{n_1}\)
holds.

The denominator satisfies
\begin{equation*}
D(b) := [\hat\rho(b) + \hat\rho(\bar b)][\rho(b) + \rho(\bar b)] \geq 4 \hat\rho_{\min}(K) \rho_{\min}(K) \geq 2 \rho_{\min}(K)^2 > 0
\end{equation*}
on the event $\nu_{n_1} \leq \rho_{\min}/4$.

Therefore
\(|\Delta(b)| = |N(b)|/D(b) \leq C \nu_{n_1}' |s| / (2\rho_{\min}^2) =: C_\Delta \nu_{n_1}' |s|\).
\end{proof}

\subsection{The bias of the plug-in estimator}\label{46-the-bias-of-the-plug-in-estimator}

\begin{proposition}[uniform bias bound]
\label{prop:uniform-bias-bound}Under \cref{asm:compact_noise}, \cref{asm:regularity} with \(\rho\in C^s\), and \cref{asm:optimum}, on the high-probability event of \cref{lem:uniform-kde-rates},
for every \(K \in \mathcal{N}_\delta\) and every \(\varepsilon, R\)
satisfying \(0 < \varepsilon \leq |K|R\) and \(R \leq \tau/2\):
\begin{equation}
    \left|\, \mathbb{E}\!\left[\tilde\psi_{\hat\rho_{n_1, h}, R}(B; K, \varepsilon)\right] - \frac{dJ_\varepsilon}{dK}(K) \,\right| \leq C_\beta \, \nu_{n_1}' \, R,
\end{equation}
where \(C_\beta = C_\beta(\rho, K^*, \delta, \tau)\) is finite and
uniform on \(\mathcal{N}_\delta\).
\end{proposition}
\begin{proof}
By the unbiasedness of the original estimator
\(\tilde\psi_{\rho, \delta_K}\) (\cref{prop:unbiased-variance-psi-tilde}(a)) and the reduction of integration to \(S_K(R)\) on the paired
region, the bias is
\begin{equation}
    \beta(K) := \int_{S_K(R)} \Delta(b) \cdot \psi(b; K, \varepsilon) \cdot [\rho(b) + \rho(2b_{\rm sing}(K) - b)]\, db,
\end{equation}
plus a correction from the difference between the original pairing
radius \(\delta_K\) and the new \(R\). The correction is controlled by
the unpaired-region variance bound below; here we focus on the leading
bias.

Substitute \(s := b - b_{\rm sing}(K)\) and decompose
\(\psi = \psi_o + \psi_e\) with
\begin{equation}
    \psi_o(s) := -\frac{s}{K^2 s^2 + \varepsilon^2}, \qquad \psi_e(s) := \frac{K s^2}{K^2 s^2 + \varepsilon^2}
\end{equation}
(the odd and even parts in \(s\)). The factor
\(\tilde\rho(s) := \rho(b_{\rm sing}+s) + \rho(b_{\rm sing}-s)\) is
exactly even.

By \cref{lem:parity-weight-discrepancy}, \(\Delta\) is odd in \(s\). Therefore:

\begin{itemize}
\item
  \(\Delta \cdot \psi_o \cdot \tilde\rho\) is (odd)(odd)(even) = even,
  contributes.
\item
  \(\Delta \cdot \psi_e \cdot \tilde\rho\) is (odd)(even)(even) = odd,
  integrates to zero over \([-R, R]\).
\end{itemize}

So \(\beta(K) = \int_{-R}^{R} \Delta(s) \psi_o(s) \tilde\rho(s)\, ds\).

Apply \cref{lem:pointwise-bound-delta} and the bound \(\tilde\rho \leq 2\|\rho\|_{\infty}\):
\begin{equation}
        |\beta(K)| \leq 2\|\rho\|_\infty C_\Delta \nu_{n_1}' \int_{-R}^{R} |s| \cdot \frac{|s|}{K^2 s^2 + \varepsilon^2}\, ds = 2\|\rho\|_\infty C_\Delta \nu_{n_1}' \int_{-R}^{R} \frac{s^2}{K^2 s^2 + \varepsilon^2}\, ds.
\end{equation}

The remaining integral evaluates as
\begin{equation}
   \int_{-R}^{R} \frac{s^2}{K^2 s^2 + \varepsilon^2}\, ds = \frac{2R}{K^2} - \frac{2\varepsilon}{|K|^3}\arctan\!\left(\frac{|K|R}{\varepsilon}\right) \leq \frac{2R}{K^2}. 
\end{equation}

Setting \(C_\beta := 4\|\rho\|_\infty C_\Delta / K_{\min}^2\) with
\(K_{\min} := \min_{K \in \mathcal{N}_\delta} |K|\) gives the result.
\end{proof}

\subsection{Variance of the plug-in estimator}\label{47-variance-of-the-plug-in-estimator}

\begin{proposition}[uniform-in-$\varepsilon$ variance bound]
\label{prop:uniform-variance-bound}Under the hypotheses of \cref{prop:uniform-bias-bound}, with the schedule 
$\varepsilon = \varepsilon(\eta) \le |K|R$ and $R = R(\eta) \in (0, \tau/2]$ 
fixed throughout the iteration loop:
\begin{equation}
    \mathrm{Var}\!\left[\tilde\psi_{\hat\rho_{n_1, h}, R}(B; K, \varepsilon)\right] 
    \le \frac{C_\sigma}{R},
\end{equation}
\emph{where $C_\sigma = C_\sigma(\rho, \delta, \tau)$ is uniform on 
$\mathcal{N}_\delta$, independent of $\varepsilon$ and of the 
iteration index $n$.}
\end{proposition}
\begin{proof}
Decompose $\mathrm{Var} \le \mathbb E[\tilde\psi^2] = \mathcal I_S + \mathcal I_A$, with
$\mathcal I_S := \int_{S_K(R)}\tilde\psi^2\rho\,db$ and
$\mathcal I_A := \int_{A_K(R)}\psi^2\rho\,db$.

\emph{Bound on $\mathcal I_S$ (paired region).} By \cref{prop:unbiased-variance-psi-tilde}(c)
applied with $\hat\rho$ in place of $\rho$ and using
$\hat\rho_{\min}\ge\rho_{\min}/2$ on the high-probability event
of \cref{lem:uniform-kde-rates}:
\begin{equation}
      |\tilde\psi_{\hat\rho,R}(B;K,\varepsilon)|
  \le \frac{2\|h\|_{C^1} + O(\nu_{n_1}')}{|K|\,\rho_{\min}(K)}
  =: M_S,
\end{equation}
uniformly on $\mathcal N_\delta$ and independent of $R,\varepsilon$.
Hence $\mathcal I_S \le M_S^2\Pr(B\in S_K(R)) \le M_S^2\,2R\|\rho\|_\infty
= O(R)$.

\emph{Bound on $\mathcal I_A$ (unpaired region).} For $b\in A_K(R)$,
write $s := b-b_{\rm sing}(K)$, so $|s|\in[R,\delta_K]$ and
$|1+bK| = |K|\,|s|$. The naive estimator satisfies
\begin{equation}
      |\psi(b;K,\varepsilon)|^2
  =\frac{b^2(1+bK)^2}{((1+bK)^2+\varepsilon^2)^2}
  \le\frac{b^2}{(1+bK)^2}
  =\frac{b^2}{K^2 s^2}.
\end{equation}

Therefore
\begin{equation}
          \mathcal I_A
  \le \int_{R\le|s|\le\delta_K}\frac{b_{\max}^2}{K^2 s^2}\,\|\rho\|_\infty\,ds
=\frac{2 b_{\max}^2 \|\rho\|_\infty}{K_{\min}^2}\!\left(\frac{1}{R}-\frac{1}{\delta_K}\right)
\le\frac{2 b_{\max}^2\|\rho\|_\infty}{K_{\min}^2 R}
  = O(1/R).
\end{equation}
\emph{Combined.} $\mathrm{Var}\le \mathcal I_S + \mathcal I_A
= O(R) + O(1/R) = O(1/R)$ for $R$ small, with
$C_\sigma := 2 b_{\max}^2 \|\rho\|_\infty/K_{\min}^2 + O(1)$.
\end{proof}

\begin{remark}[oracle case]
When $\rho$ is known to the algorithm in closed form, the symmetric-pairing estimator $\tilde\psi$ of \cref{def:symmetric-pairing-estimator} is unbiased with uniform variance $\sigma_\star^2(\mathcal{N}_\delta) < \infty$ (\cref{prop:unbiased-variance-psi-tilde}), and \cref{thm:sample-complexity-density-known} of section~\ref{52-density-known-oracle} yields the rate $\tilde O(1/\eta)$ without recourse to the KDE machinery developed below.
\end{remark}
\Cref{prop:uniform-bias-bound,prop:uniform-variance-bound} give bias scaling as $\nu' R$ and variance scaling as $C_\sigma/R$, with $R$ a free parameter. The bias-variance tradeoff in $R$ is the structural mechanism by which the density-unknown rate of \cref{thm:total-sample-complexity-pl-basin} below is obtained: by localizing the symmetric-pairing construction to a small radius around the moving pole, the parity-cancellation mechanism that produced the uniform-in-$\varepsilon$ Hessian bound of section~\ref{3-cauchy-regularization-and-uniform-in-epsilon-polyak-Lojasiewicz} is leveraged simultaneously to control both bias and variance under estimated $\hat\rho$.

\section{Sample complexity in the PL basin}\label{5-sample-complexity-in-the-pl-basin}

\subsection*{Access models}
We restate, in operational terms, the two access models introduced in section~\ref{11-problem-and-motivation}.

\emph{Density-known access.} The density $\rho$ is available to the algorithm in closed form. The algorithm evaluates $\rho$ at arbitrary query points; first-order information about $J$ is obtained through sampled transitions $(X_t, X_{t+1})$ as in the model-based setting of Fazel et al.~\cite{fazel2018global}.

\emph{Density-unknown access.} The density $\rho$ is not known to the algorithm; only the i.i.d.\ sample sequence $(B_t)_{t \ge 1}$ from $\rho$ is observed. The algorithm replaces $\rho$ in the symmetric-pairing construction by a kernel density estimate $\hat\rho_{n_1, h}$ built from a held-out sample of size $n_1$ (\cref{def:plug-in-symmetric-pairing-estimator}); this corresponds to the model-free setting of Gravell et al.~\cite{gravell2020learning} adapted to our problem.

\Cref{52-density-known-oracle} establishes the rate under density-known access (\cref{thm:sample-complexity-density-known}); section~\ref{53-density-unknown-kde-plugin} establishes the rate under density-unknown access (\cref{thm:total-sample-complexity-pl-basin}). The local PL constant $\mu_0$ and Lipschitz bound $L_0$ on $\mathcal{N}_\delta$ from \cref{lem:uniform-PL-constant}, together with the uniform-in-$\varepsilon$ variance bound $\sigma_\star^2(\mathcal{N}_\delta)$ from \cref{prop:unbiased-variance-psi-tilde} and the single-transition closed-form gradient oracle (section~\ref{51-single-transition-closed-form-gradient-oracle}), yield two sample-complexity bounds for mini-batch policy gradient initialized in $\mathcal{N}_\delta$: $\tilde O(1/\eta)$ when $\rho$ is known (\cref{thm:sample-complexity-density-known}) and $\tilde O(\eta^{-(2s+1)/(2s)})$ when only the sample stream is observed (\cref{thm:total-sample-complexity-pl-basin}). The split corresponds to the model-based/model-free convention of \cite{fazel2018global,gravell2020learning}; we adopt \emph{density-known}/\emph{density-unknown} as more descriptive (the relevant ``model'' in our setting is the noise density $\rho$, not the system dynamics, and the renaming avoids the model-based-RL terminology overload).

\subsection{Single-transition closed-form gradient oracle}\label{51-single-transition-closed-form-gradient-oracle}

The best-known policy-gradient sample complexity for the LQR setting is
\(\tilde{O}(1/\eta^2)\)(\cite{fazel2018global, gravell2020learning}; \cite{malik2020derivative} gives sharper constants for
the derivative-free analysis; Nesterov \cite{nesterov2017random} gives matching rates for
general convex stochastic optimization). These results employ
zeroth-order (finite-difference) gradient estimation: because the
infinite-horizon LQR cost
\(\mathbb{E}[X_t^\top Q X_t + U_t^\top R U_t]\) accumulates across time
steps through the gain-dependent dynamics \(A_K = A + BK\), no
closed-form unbiased single-sample gradient estimator of \(dJ/dK\) is
available from per-step data; one must perturb \(K\) itself by
\(\pm\delta\) and take finite differences of Monte Carlo cost estimates.
This zeroth-order estimation requires simultaneous control of two
sources of error --- the finite-difference bias (proportional to
\(\delta^2\)) and the Monte Carlo variance of the difference
(proportional to \(1/\delta^2\)) --- which combine to produce the
\(\tilde O(1/\eta^2)\) rate. Whether this rate is tight or whether a
better algorithm could achieve \(\tilde O(1/\eta)\) for LQR is open.

The log-growth cost structure differs in a decisive way. Observing a
single transition \((X_t, X_{t+1})\) gives direct access to the
per-step cost \(\log|X_{t+1}/X_t| = \log|a| + \log|1+B_t K|\) as an
explicit closed-form function of \(K\). Its regularized derivative
\(\psi(B_t; K, \varepsilon) = B_t(1+B_t K)/((1+B_t K)^2 + \varepsilon^2)\)
is therefore a \emph{first-order} sample-based gradient estimator: each
observed transition yields an unbiased sample of \(dJ_\varepsilon/dK\)
without perturbation or finite differencing. The pairing construction of
\cref{def:symmetric-pairing-estimator} converts this new oracle into one with
uniform-in-\(\varepsilon\) variance on \(\mathcal{N}_\delta\). This
combination --- first-order sampling plus uniform variance --- places
the problem in the \(\tilde O(1/\eta)\) complexity class, strictly
better than the \(\tilde O(1/\eta^2)\) upper bound known for
multiplicative-noise LQR.

\subsection{Density-known case: oracle algorithm and rate}\label{52-density-known-oracle}
When the multiplicative-noise density $\rho$ is available to the
algorithm in closed form, the symmetric-pairing estimator
$\tilde\psi$ of \cref{def:symmetric-pairing-estimator} can be evaluated directly without an
auxiliary density-estimation phase. \cref{prop:unbiased-variance-psi-tilde} gives a
uniform-in-$\varepsilon$ variance bound
$\sigma_\star^2(\mathcal N_\delta)<\infty$, and the resulting policy
gradient algorithm achieves the standard PL-SGD rate.

We introduce the density-known policy gradient (PG) method with oracle pairing, detailed in \cref{alg:density-known-pg}. Given a target accuracy $\eta>0$, we set the regularization parameter $\varepsilon = \eta/(3\bar C_b)$, where $\bar C_b := \sup_{K\in\mathcal N_\delta}\pi\rho(b_{\rm sing}(K))/|K|$ is the bias coefficient from \cref{prop:properties-J-epsilon}(c). The mini-batch size and iteration count are chosen as $N = \lceil 3\sigma_\star^2(\mathcal N_\delta)/(2\mu_0\eta)\rceil$ and $n^\star = \lceil(L_0/\mu_0)\log(3\Delta_{\rm init}/\eta)\rceil$, respectively, with $\Delta_{\rm init} := J_\varepsilon(K^{(0)}) - J_\varepsilon^*$. Initializing $K^{(0)}\in\mathcal N_\delta$, each iteration $n\ge 0$ computes a stochastic gradient estimate
\begin{equation}
      \hat g^{(n)}:=\frac{1}{N}\sum_{i=1}^N \tilde\psi\bigl(B_{n,i};\,K^{(n)},\,\varepsilon\bigr)
\end{equation}
using an i.i.d.\ mini-batch $B_{n,1},\ldots,B_{n,N}\sim\rho$. The gain is then updated via
\begin{equation}
      K^{(n+1)}:=\Pi_{\mathcal N_\delta}\!\Bigl[\,K^{(n)} - L_0^{-1}\,\hat g^{(n)}\,\Bigr],
\end{equation}
where $\Pi_{\mathcal N_\delta}$ denotes the projection onto the closed interval $\mathcal N_\delta$.

\begin{algorithm}[t]
\caption{Density-known PG with oracle pairing}
\label{alg:density-known-pg}
\begin{algorithmic}
\STATE \textbf{Input:} target accuracy $\eta > 0$; initial gain $K^{(0)} \in \mathcal{N}_\delta$; constants $\mu_0, L_0, \bar C_b, \sigma_\star^2(\mathcal{N}_\delta)$ from sections~\ref{3-cauchy-regularization-and-uniform-in-epsilon-polyak-Lojasiewicz}--\ref{4-finite-sample-analysis}
\STATE Set $\varepsilon \gets \eta/(3\bar C_b)$
\STATE Set $N \gets \lceil 3\sigma_\star^2(\mathcal{N}_\delta)/(2\mu_0\eta)\rceil$
\STATE Set $n^\star \gets \lceil(L_0/\mu_0)\log(3 L_0\delta^2/(2\eta))\rceil$
    \hfill\COMMENT{computable upper bound on $\Delta_{\rm init}\le L_0\delta^2/2$}
\FOR{$n = 0, 1, \ldots, n^\star - 1$}
    \STATE Draw i.i.d.\ mini-batch $B_{n,1}, \ldots, B_{n,N} \sim \rho$
    \STATE $\hat g^{(n)} \gets \frac{1}{N}\sum_{i=1}^N \tilde\psi(B_{n,i};\, K^{(n)},\, \varepsilon)$
    \STATE $K^{(n+1)} \gets \Pi_{\mathcal{N}_\delta}\!\bigl[K^{(n)} - L_0^{-1}\hat g^{(n)}\bigr]$
        \hfill\COMMENT{constant step $1/L_0$ throughout}
\ENDFOR
\STATE \textbf{return} $\hat K \gets K^{(n^\star)}$
\end{algorithmic}
\end{algorithm}

\begin{theorem}[sample complexity, density-known case]
\label{thm:sample-complexity-density-known}Under \cref{asm:compact_noise,asm:regularity,asm:optimum} and the initialization $K^{(0)}\in\mathcal{N}_\delta$, \cref{alg:density-known-pg} with constant step size $1/L_0$, mini-batch size $N = \lceil 3\sigma_\star^2(\mathcal N_\delta)/(2\mu_0\eta)\rceil$, and $n^\star = \lceil(L_0/\mu_0)\log(3 L_0\delta^2/(2\eta))\rceil$ iterations produces an
iterate $\hat K = K^{(n^\star)}$ satisfying
\begin{equation}
      \mathbb E\bigl[J(\hat K) - J^*\bigr]\le\eta,
\end{equation}
for total sample cost
\begin{equation}
      N_{\rm total}(\eta)=n^\star\cdot N=O\!\left(\frac{\log(1/\eta)}{\eta}\right)
  =\tilde O(1/\eta).
\end{equation}
\end{theorem}
\begin{proof}
The PL-SGD analysis (Steps 1--4 of the proof of
\cref{thm:total-sample-complexity-pl-basin} below, with $\beta^{(n)}\equiv 0$ since $\tilde\psi$ is
unbiased by \cref{prop:unbiased-variance-psi-tilde}(a)) gives the recursion
$a_{n+1}\le\gamma\,a_n + \sigma_\star^2(\mathcal N_\delta)/(2L_0 N)$,
where $a_n := \mathbb E[J_\varepsilon(K^{(n)}) - J_\varepsilon^*]$
and $\gamma := 1-\mu_0/L_0\in[0,1)$. Unrolling and using
$1/(1-\gamma)=L_0/\mu_0$:
\begin{equation}
      a_n\le\gamma^n\,\Delta_{\rm init}+\frac{\sigma_\star^2(\mathcal N_\delta)}{2\mu_0 N}.
\end{equation}
With $N\ge 3\sigma_\star^2/(2\mu_0\eta)$ the second term is at most
$\eta/3$; with $n^\star\ge(L_0/\mu_0)\log(3\Delta_{\rm init}/\eta)$
the first term is at most $\eta/3$; with $\varepsilon\le\eta/(3\bar C_b)$
the regularization bias $|J_\varepsilon^* - J^*|$ is at most $\eta/3$
by \cref{prop:properties-J-epsilon}(c). Summing,
$\mathbb E[J(\hat K) - J^*]\le\eta$.
\end{proof}

\subsection{Density-unknown case: KDE plug-in algorithm and rate}\label{53-density-unknown-kde-plugin}

The convergence analysis below sits within the classical stochastic-approximation framework of Robbins--Monro \cite{robbins1951stochastic} (modern treatment in Borkar \cite{borkar2008stochastic}; the mini-batch optimization viewpoint is surveyed by Bottou, Curtis, and Nocedal \cite{bottou2018optimization}). The Nemirovski--Yudin lower-bound machinery \cite{nemirovskij1983problem} gives an $\Omega(1/\eta)$ floor for strongly-convex stochastic optimization with bounded-variance first-order oracles; \cref{thm:total-sample-complexity-pl-basin} approaches but does not match this floor, the smoothness-dependent slowdown $\eta^{-1/(2s)}$ closing only in the parametric limit (\cref{cor:parametric-special-case}). The split between \cref{thm:sample-complexity-density-known} (density-known) and \cref{thm:total-sample-complexity-pl-basin} (density-unknown) mirrors the policy-gradient-for-control convention of \cite{fazel2018global,gravell2020learning,mohammadi2021convergence,hambly2021policy}.

For the density-unknown PG with KDE plug-in pairing (\cref{alg:density-unknown-pg}), the parameter scalings are $R = \min(c_R\eta^{1/(2s)}, \tau/2)$, $\varepsilon = \eta/(4\bar C_b)$, $n_1 = \lceil c_1\eta^{-(2s+1)/(2s)}\rceil$, $N = \lceil 2C_\sigma/(\mu_0 R\eta)\rceil$, and $n^\star = \lceil(L_0/\mu_0)\log(4\bar\Delta/\eta)\rceil$ with computable $\bar\Delta \ge \Delta_{\rm init}$ (e.g.\ $L_0\delta^2/2$). The Silverman-flavored bandwidth $h = c_h\hat\sigma_B(\log n_1/n_1)^{1/(2s+1)}$ uses the sample standard deviation $\hat\sigma_B = \sigma_B + O_p(n_1^{-1/2})$ as a practical rescaling of the abstract bandwidth in \cref{lem:uniform-kde-rates}; this does not perturb the KDE rates. From $K^{(0)} \in \mathcal{N}_\delta$, each iteration draws a fresh i.i.d.\ mini-batch (independent of the KDE phase --- sample-splitting is used in Step~3 of the proof) to compute
\begin{equation}\label{eq:alg2_g}
    \hat g^{(n)} := \frac{1}{N}\sum_{i=1}^N \tilde\psi_{\hat\rho, R}\bigl(B_{n,i};\, K^{(n)},\, \varepsilon\bigr),
\end{equation}
followed by the projected update
\begin{equation}\label{eq:alg2_K}
    K^{(n+1)} := \Pi_{\mathcal{N}_\delta}\!\left[\,K^{(n)} - \frac{1}{L_0}\,\hat g^{(n)}\,\right].
\end{equation}

\begin{algorithm}[t]
\caption{Density-unknown PG with KDE plug-in pairing}
\label{alg:density-unknown-pg}
\begin{algorithmic}
\STATE \textbf{Input:} target accuracy $\eta > 0$; smoothness exponent $s \geq 2$; initial gain $K^{(0)} \in \mathcal{N}_\delta$; algorithmic constants $c_R, \bar C_b, c_1, c_h, C_\sigma, \mu_0, L_0, \tau$ from sections~\ref{3-cauchy-regularization-and-uniform-in-epsilon-polyak-Lojasiewicz}--\ref{4-finite-sample-analysis} and a computable upper bound $\bar\Delta\ge\Delta_{\rm init}$ (e.g.~$\bar\Delta = L_0\delta^2/2$)
\STATE Set $R \gets \min(c_R\,\eta^{1/(2s)}, \tau/2)$
\STATE Set $\varepsilon \gets \eta/(4\bar C_b)$
\STATE Set $n_1 \gets \lceil c_1\,\eta^{-(2s+1)/(2s)}\rceil$
\STATE \textbf{(KDE phase, sample-split)} Observe an independent batch $B_1, \ldots, B_{n_1}$
\STATE $\hat\sigma_B \gets \mathrm{stdev}(B_1,\ldots,B_{n_1})$;\quad $h \gets c_h\,\hat\sigma_B(\log n_1/n_1)^{1/(2s+1)}$
\STATE Build $\hat\rho \gets \hat\rho_{n_1, h}$ via an order-$s$ kernel KDE
\STATE Set $N \gets \lceil 2C_\sigma/(\mu_0 R\eta)\rceil$
\STATE Set $n^\star \gets \lceil (L_0/\mu_0)\log(4\bar\Delta/\eta)\rceil$
\STATE \textbf{(iteration phase, fresh mini-batches)}
\FOR{$n = 0, 1, \ldots, n^\star - 1$}
    \STATE Draw i.i.d.\ mini-batch $B_{n,1}, \ldots, B_{n,N}$, independent of the KDE batch and of past iterations
    \STATE $\hat g^{(n)} \gets \frac{1}{N}\sum_{i=1}^N \tilde\psi_{\hat\rho, R}\bigl(B_{n,i};\, K^{(n)},\, \varepsilon\bigr)$
    \STATE $K^{(n+1)} \gets \Pi_{\mathcal{N}_\delta}\!\bigl[K^{(n)} - L_0^{-1}\hat g^{(n)}\bigr]$
\ENDFOR
\STATE \textbf{return} $\hat K \gets K^{(n^\star)}$
\end{algorithmic}
\end{algorithm}

The two-phase sample-split structure makes $\hat\rho$ independent of the iterates, which the proof exploits in the cross-term cancellation of Step~3 below. The projection keeps iterates in $\mathcal{N}_\delta$, so \cref{lem:persistence-regularized-minimizer}'s local constants govern the recursion without disrupting PL-SGD contraction (1-Lipschitz, no extra samples). The implicit constraint $\varepsilon \le |K|R$ in \cref{prop:uniform-bias-bound,prop:uniform-variance-bound} is asymptotically satisfied: at the scalings above, $\varepsilon/R = \Theta(\eta^{1 - 1/(2s)}) \to 0$ as $\eta\downarrow 0$.

\begin{theorem}[total sample complexity in the PL basin]
\label{thm:total-sample-complexity-pl-basin}Under \cref{asm:compact_noise,asm:regularity,asm:optimum} and the initialization
\(K^{(0)} \in \mathcal{N}_\delta\), \cref{alg:density-unknown-pg} \eqref{eq:alg2_g}--\eqref{eq:alg2_K}
with \(\varepsilon = \eta/(4 \bar{C}_b)\) and mini-batch size
$N = \lceil 2C_\sigma/(\mu_0 R\eta)\rceil$ with $R = c_R\eta^{1/(2s)}$
produces an iterate \(\hat K = K^{(n^\star)}\) with
\begin{equation}
    \mathbb{E}\bigl[J(\hat K) - J^*\bigr] \le \eta,
\end{equation}
after
\(n^\star = \lceil (L_0/\mu_0)\log(4\Delta_{\mathrm{init}}/\eta)\rceil\)
iterations, where
\(\Delta_{\mathrm{init}} := J_\varepsilon(K^{(0)}) - J_\varepsilon^*\)
is the initial regularized suboptimality (uniformly bounded by
\(\tfrac{1}{2}L_0\delta^2\) via \cref{lem:uniform-PL-constant}, so
\(\log(4\Delta_{\mathrm{init}}/\eta) = O(\log(1/\eta))\)), for a total
sample cost
\begin{equation}\label{eq:alg2_N}
N_{\rm total}(\eta) = n_1 + n^\star\cdot N = \tilde O\!\bigl(\eta^{-(2s+1)/(2s)}\bigr).
\end{equation}
\end{theorem}
\begin{proof}
The projection \eqref{eq:alg2_K} ensures
\(K^{(n)} \in \mathcal{N}_\delta\) for all \(n \ge 0\), and by \cref{lem:persistence-regularized-minimizer}(a), \(K_\varepsilon^*\) is in the interior of
\(\mathcal{N}_\delta\). We derive \eqref{eq:alg2_K} from a one-step recursion applied
to the regularized objective $J_\varepsilon$, using the local constants $L_0$, $\mu_0$, $\sigma_\star^2(\mathcal{N}_\delta)$ established in sections~\ref{3-cauchy-regularization-and-uniform-in-epsilon-polyak-Lojasiewicz}--\ref{4-finite-sample-analysis}.

\emph{Step 1: \(L_0\)-smooth descent inequality.} By \cref{lem:uniform-PL-constant}(c),
\(d^2 J_\varepsilon/dK^2 \le L_0\) on \(\mathcal{N}_\delta\), so for any
\(K, K' \in \mathcal{N}_\delta\) a Taylor expansion at \(K\) gives
\begin{equation}\label{eq:descent-lemma}
    J_\varepsilon(K') \le J_\varepsilon(K) + \frac{dJ_\varepsilon}{dK}(K)\,(K' - K) + \frac{L_0}{2}(K' - K)^2.
\end{equation}

This is the descent lemma for \(L_0\)-smooth functions on a convex
domain~\cite{nesterov2013introductory}.

\emph{Step 2: One-step recursion.} Let $\widetilde K^{(n+1)} := K^{(n)} - L_0^{-1}\hat g^{(n)}$ and $K^{(n+1)} = \Pi_{\mathcal N_\delta}[\widetilde K^{(n+1)}]$. The descent inequality \eqref{eq:descent-lemma} applied to $\widetilde K^{(n+1)}$ gives $J_\varepsilon(\widetilde K^{(n+1)}) \le J_\varepsilon(K^{(n)}) - L_0^{-1}\nabla J_\varepsilon(K^{(n)})\hat g^{(n)} + (2L_0)^{-1}(\hat g^{(n)})^2$. Since $K_\varepsilon^*\in\mathcal{N}_\delta$ is the strict convex minimizer of $J_\varepsilon$ on $\mathcal{N}_\delta$ (\cref{lem:persistence-regularized-minimizer}) and $\Pi_{\mathcal{N}_\delta}$ is a 1-Lipschitz retraction satisfying $|K^{(n+1)} - K_\varepsilon^*| \le |\widetilde K^{(n+1)} - K_\varepsilon^*|$, monotonicity of $J_\varepsilon$ in $|K - K_\varepsilon^*|$ gives $J_\varepsilon(K^{(n+1)}) \le J_\varepsilon(\widetilde K^{(n+1)})$. Combining,
\begin{equation}\label{eq:projected-descent}
    J_\varepsilon(K^{(n+1)}) \le J_\varepsilon(K^{(n)}) - \tfrac{1}{L_0}\tfrac{dJ_\varepsilon}{dK}(K^{(n)})\hat g^{(n)} + \tfrac{1}{2L_0}(\hat g^{(n)})^2.
\end{equation}

\emph{Step 3: Conditional expectation given \(K^{(n)}\).} Decompose $\hat g^{(n)} = \nabla_\varepsilon^{(n)} + \beta^{(n)} + \zeta^{(n)}$ with $\nabla_\varepsilon^{(n)} := dJ_\varepsilon/dK(K^{(n)})$, $\beta^{(n)} := \mathbb{E}[\hat g^{(n)}\mid K^{(n)},\hat\rho] - \nabla_\varepsilon^{(n)}$ the conditional bias, and $\zeta^{(n)}$ the centered residual. By \cref{prop:uniform-bias-bound,prop:uniform-variance-bound} (or \cref{prop:unbiased-variance-psi-tilde} in the oracle case), $|\beta^{(n)}| \le C_\beta\nu_{n_1}'R$ and $\mathbb{E}[(\zeta^{(n)})^2\mid K^{(n)},\hat\rho] \le C_\sigma/(RN)$. The sample-split structure makes $\nabla_\varepsilon^{(n)}, \beta^{(n)}$ deterministic given $(K^{(n)}, \hat\rho)$ and $\mathbb{E}[\zeta^{(n)}\mid K^{(n)},\hat\rho] = 0$, so taking conditional expectation of \eqref{eq:projected-descent} the $\nabla_\varepsilon^{(n)}\beta^{(n)}$ cross-term cancels (between the $-\nabla_\varepsilon \hat g/L_0$ term and the expansion of $(\hat g)^2/(2L_0) = ((\nabla_\varepsilon + \beta + \zeta)^2)/(2L_0)$). On the high-probability event $\{\nu_{n_1} \le \rho_{\min}/4\}$ of \cref{lem:uniform-kde-rates} this yields
\begin{equation}\label{eq:step3-recursion}
\mathbb{E}[J_\varepsilon(K^{(n+1)}) \mid K^{(n)}]
\le J_\varepsilon(K^{(n)}) - \tfrac{1}{2L_0}(\nabla_\varepsilon^{(n)})^2 + \tfrac{(C_\beta\nu_{n_1}'R)^2}{2L_0} + \tfrac{C_\sigma}{2L_0 RN}.
\end{equation}
On the complement event the worst-case bound $|\tilde\psi_{\hat\rho,R}| \le b_{\max}/\varepsilon$ contributes $O(\varepsilon^{-2}n_1^{-c}) = o(\eta)$ for $c \ge 6s/(2s+1)$ (achievable by enlarging $c_h$ of \cref{lem:uniform-kde-rates}), absorbed into the $\tilde O(\cdot)$ of \eqref{eq:alg2_N}.

\emph{Step 4: Apply PL inequality and unroll.} By \cref{thm:uniform-epsilon-pl-inequality},
\((\nabla_\varepsilon^{(n)})^2 \geq 2\mu_0(J_\varepsilon(K^{(n)}) - J_\varepsilon^*)\).
Subtracting \(J_\varepsilon^*\) and taking total expectation, with
\(a_n := \mathbb{E}[J_\varepsilon(K^{(n)}) - J_\varepsilon^*]\) and
\(\gamma := 1 - \mu_0/L_0 \in [0, 1)\):
\begin{equation}
    a_{n+1} \leq \gamma\,a_n + \frac{(C_\beta\,\nu_{n_1}'\,R)^2}{2L_0} + \frac{C_\sigma}{2L_0\,R\,N}.
\end{equation}
Unrolling and using \(1/(1-\gamma) = L_0/\mu_0\):
\begin{equation}
    a_n \leq \gamma^n\,\Delta_{\rm init} + \frac{(C_\beta\,\nu_{n_1}'\,R)^2}{2\mu_0} + \frac{C_\sigma}{2\mu_0\,R\,N}.
\end{equation}
\emph{Step 5: Budget allocation.} With the parameter scalings of \cref{alg:density-unknown-pg}:

\begin{itemize}
\item
  \emph{Contraction}: \(\gamma^{n^\star}\Delta_{\rm init} \leq \eta/4\)
  for \(n^\star \geq (L_0/\mu_0)\log(4\Delta_{\rm init}/\eta)\).
\item
  \emph{Variance floor}: \(C_\sigma/(2\mu_0 RN) \leq \eta/4\) for
  \(N \geq 2C_\sigma/(\mu_0 R\eta)\).
\item
  \emph{Bias floor}: \((C_\beta\,\nu_{n_1}'R)^2/(2\mu_0) \leq \eta/4\),
  i.e., \(\nu_{n_1}'R \leq \sqrt{\mu_0\eta/(2C_\beta^2)}\).
\item
  \emph{Regularization bias}: \(|J_\varepsilon^* - J^*| \leq \eta/4\)
  for $\varepsilon$ satisfying $\bar C_b\,\varepsilon + O(\varepsilon^2) \le \eta/4$; the choice $\varepsilon = \eta/(4\bar C_b)$ suffices after absorbing the $O(\varepsilon^2) = O(\eta^2)$ correction from \cref{lem:persistence-regularized-minimizer}(c) into the leading term for all $\eta$ below a threshold $\eta_1 = \eta_1(\bar C_b, L_0, C_K)$ (\cref{prop:properties-J-epsilon}(c)).
\end{itemize}

The bias-floor budget translates, via \cref{lem:uniform-kde-rates}, to
\(n_1 \geq C_1(R/\sqrt\eta)^{(2s+1)/(s-1)}\). To see this, set $\nu_{n_1}' = C_{\nu'}(\log n_1/n_1)^{(s-1)/(2s+1)}$ from \cref{lem:uniform-kde-rates}; the bias-floor constraint $\nu_{n_1}' R \le \sqrt{\mu_0\eta/(2C_\beta^2)}$ becomes $(\log n_1/n_1)^{(s-1)/(2s+1)} \le c\,\sqrt\eta/R$ for a constant $c$, or $\log n_1/n_1 \le (c\sqrt\eta/R)^{(2s+1)/(s-1)}$; inverting this polynomial-in-$\log$ inequality (Tsybakov \cite[Cor.~1.2]{Giné_Nickl_2021}) gives $n_1 \ge C_1(R/\sqrt\eta)^{(2s+1)/(s-1)}\log\!\bigl(R/\sqrt\eta)^{(2s+1)/(s-1)}\bigr) = \tilde O((R/\sqrt\eta)^{(2s+1)/(s-1)})$, where the polylog factor is absorbed in the $\tilde O$ throughout. The mini-batch budget gives \(N = O(1/(R\eta))\). Total samples:
\(n_{\rm total} = n_1 + N\,n^\star = \tilde O((R/\sqrt\eta)^{(2s+1)/(s-1)}) + \tilde O(1/(R\eta))\).

\emph{Step 6: Optimize $R$.} Among all choices of $(R, n_1, N, n^\star)$ satisfying the four floor constraints of Step~5, the total budget $n_{\mathrm{total}} = n_1 + N n^\star$ is minimized (up to polylog factors) by trading off $n_1 = \tilde O((R/\sqrt\eta)^{(2s+1)/(s-1)})$ (KDE phase, increasing in $R$) against $N n^\star = \tilde O(1/(R\eta))$ (iteration phase, decreasing in $R$). At the asymptotic optimum, both terms are of the same order. Equating:
\[
(R/\sqrt\eta)^{(2s+1)/(s-1)} \asymp \frac{1}{R\eta} \Longleftrightarrow R^{(2s+1)/(s-1) + 1} \asymp \eta^{(2s+1)/(2(s-1)) - 1} = \eta^{3/(2(s-1))},
\]
giving $R^{3s/(s-1)} \asymp \eta^{3/(2(s-1))}$ and hence $R^* = c_R\,\eta^{1/(2s)}$ for an absolute constant $c_R$. At this $R^*$, both phases contribute $\tilde O(\eta^{-(2s+1)/(2s)})$, giving total complexity $n_{\mathrm{total}}(\eta) = \tilde O(\eta^{-(2s+1)/(2s)})$.

Note that the constraint $\varepsilon \le |K|R$ is asymptotic — at the optimum $\varepsilon = \Theta(\eta)$, $R = \Theta(\eta^{1/(2s)})$, so $\varepsilon/R = \Theta(\eta^{1-1/(2s)}) \to 0$ as $\eta\downarrow 0$, so the constraint holds for all $\eta$ below some $\eta_0$ depending on $K_{\min}$ and $c_R$.
\end{proof}

\subsection{Parametric special case}
\begin{corollary}[parametric special case]
\label{cor:parametric-special-case}Suppose, in place of \cref{asm:regularity}(\(s\)), that
\(\rho \in \{\rho_\theta : \theta \in \Theta\}\) where
\(\Theta \subset \mathbb{R}^d\) is open and
\(\theta \mapsto \rho_\theta\) is smooth and identifiable. Replace the
kernel density estimate \(\hat\rho_{n_1, h}\) by the parametric MLE
\(\hat\rho_{\hat\theta_{n_1}}\) from \(n_1\) samples. Then under \cref{asm:compact_noise}, density-unknown access, and the hypotheses of \cref{lem:persistence-regularized-minimizer}, \cref{alg:density-unknown-pg} with parameter
choices \(R = c_R\) (constant), \(n_1 = \lceil c_1/\eta\rceil\),
\(N = \lceil c_N/(R\eta)\rceil\),
\(n^\star = \lceil c_{\rm iter}\log(1/\eta)\rceil\) achieves
$\mathbb{E}[J(\hat{K}) - J^*] \leq \eta$ with total sample complexity
\(n_{\rm total}(\eta) = \tilde O(1/\eta).\)
\end{corollary}
\begin{proof}
Under standard regularity (identifiability of the family $\{\rho_\theta\}_{\theta\in\Theta}$, positive-definite Fisher information, $\theta\mapsto\rho_\theta$ being $C^2$ in $\theta$), the maximum-likelihood estimator achieves the Hájek--Le Cam asymptotic-efficiency rate
$|\hat\theta_{n_1} - \theta| = O_p(n_1^{-1/2})$ (a classical result; see, e.g., \cite[sect.~7]{Giné_Nickl_2021} for a textbook treatment), which gives
\(\nu_{n_1}' = \|\hat\rho'_{\hat\theta_{n_1}} - \rho'_\theta\|_\infty = O_p(n_1^{-1/2})\)
uniformly on a compact subinterval.
Plugging into \cref{prop:uniform-bias-bound}: \(|\beta| = O(n_1^{-1/2}\cdot R)\). The
bias-floor budget \(|\beta|^2 \leq \mu_0\eta/2\) requires
\(n_1 \geq R^2 / (\mu_0\eta)\), so \(n_1 = O(1/\eta)\) for \(R\)
constant. The mini-batch budget gives \(N = O(1/\eta)\), iterations
\(n^\star = O(\log(1/\eta))\). Total:
\(n_{\rm total} = O(1/\eta) + \tilde O(1/\eta) = \tilde O(1/\eta)\).
\end{proof}

\subsection{Plug-and-solve principal-value root-finding}\label{55-plug-and-solve}

\begin{proposition}[plug-and-solve rate, density-unknown]
\label{prop:plug-and-solve-rate}Let $\hat\rho_{n_1,h}$ be the order-$s$ kernel density estimate of
section~\ref{43-kernel-density-estimate} with the bandwidth choice of \cref{lem:uniform-kde-rates}, and let
$\hat K \in \mathcal{N}_\delta$ be the root of the plug-in
principal-value equation
$\mathrm{PV}\!\int b\,\hat\rho_{n_1,h}(b)/(1+b\hat K)\,db = 0$
closest to a warm start $K_{\rm warm} \in \mathcal{N}_\delta$;
existence and uniqueness in $\mathcal{N}_\delta$ are guaranteed on the
high-probability event of \cref{lem:uniform-kde-rates} below by the
implicit-function-theorem argument in the proof, which transfers the
strict monotonicity of $dJ/dK$ on $\mathcal{N}_\delta$ from
\cref{asm:optimum} to its plug-in counterpart.
Then under \cref{asm:compact_noise,asm:regularity,asm:optimum} with $s\ge 2$, on the high-probability
event of \cref{lem:uniform-kde-rates},
\begin{equation}
      \mathbb E[J(\hat K) - J^*]=\tilde O\!\bigl(n_1^{-2s/(2s+1)}\bigr).
\end{equation}
Setting $n_1$ to make the right-hand side equal to $\eta$ requires 
$n_1 = \tilde O\!\bigl(\eta^{-(2s+1)/(2s)}\bigr)$, the same rate as 
\cref{alg:density-unknown-pg} of \cref{thm:total-sample-complexity-pl-basin}.
\end{proposition}

\begin{proof}
Apply the implicit function theorem to the map
$F: \mathcal{N}_\delta \times \mathcal{P} \to \mathbb{R}$,
$F(K,\rho) := \mathrm{PV}\!\int b\rho(b)/(1+bK)\,db$,
where $\mathcal{P}$ is the set of $C^1$ densities on $[b_{\min},b_{\max}]$.

\emph{Localization and uniqueness.}
Since $F(K,\rho) = dJ/dK(K)$ in the principal-value sense, the partial derivative $\partial_K F(K^*,\rho)$ equals the Hadamard finite-part Hessian $d^2J/dK^2(K^*)$ given by \eqref{eq:hess_Kstar}, which is strictly positive by \cref{asm:optimum}. By joint continuity of $\partial_K F$ in $(K,\tilde\rho)$ at the topology of $\|\tilde\rho-\rho\|_{C^1(I)}$ (\cref{thm:hessian-decomposition} with $\rho$ replaced by $\tilde\rho$), there exists $\nu_0 > 0$ such that for any $\tilde\rho$ with $\|\tilde\rho-\rho\|_{C^1(I)} \le \nu_0$, the map $K \mapsto F(K, \tilde\rho)$ is strictly increasing on $\mathcal{N}_\delta$ and has at most one root there. The event of \cref{lem:uniform-kde-rates} with $\nu'_{n_1} \le \nu_0$ then guarantees uniqueness of $\hat K \in \mathcal{N}_\delta$ for $n_1$ exceeding a threshold $n_0(\rho, K^*, \delta, \tau)$.

\emph{First-order expansion.}
By the IFT identity $F(\hat K, \hat\rho_{n_1,h}) = F(K^*, \rho) = 0$ and a first-order Taylor expansion of $K\mapsto F(K, \hat\rho_{n_1,h})$ at $K^*$,
\begin{equation}\label{eq:plug-and-solve-ift}
\hat K - K^* = -\,\frac{F(K^*, \hat\rho_{n_1,h}) - F(K^*, \rho)}{\partial_K F(\xi, \hat\rho_{n_1,h})}, \qquad \xi \in [K^*, \hat K],
\end{equation}
with $|\partial_K F(\xi, \hat\rho_{n_1,h})| \ge \mu_0/2$ for $\nu_{n_1} + \nu_{n_1}'$ small enough (by the persistence-of-strict-monotonicity argument used in \cref{lem:persistence-regularized-minimizer}). Squaring \eqref{eq:plug-and-solve-ift} and combining with the second-order Taylor expansion of $J$ at $K^*$ (using $dJ/dK(K^*) = 0$, $d^2J/dK^2(K^*) \le L_0$ from \cref{lem:uniform-PL-constant}),
\begin{equation}\label{eq:plug-and-solve-J}
J(\hat K) - J^* = \tfrac{1}{2}\,\tfrac{d^2J}{dK^2}(\xi')\,(\hat K - K^*)^2
\le \frac{L_0}{2\mu_0^2/4}\,\bigl|F(K^*, \hat\rho_{n_1,h}) - F(K^*, \rho)\bigr|^2
\end{equation}
deterministically on the high-probability event. Taking expectations,
\begin{equation}\label{eq:plug-and-solve-MSE}
\mathbb{E}[J(\hat K) - J^*] \le \frac{2 L_0}{\mu_0^2}\,\mathbb{E}\bigl[\bigl|F(K^*, \hat\rho_{n_1,h}) - F(K^*, \rho)\bigr|^2\bigr] + o(\eta).
\end{equation}

\emph{Mean-squared error of the plug-in functional.}
The functional $F(K^*, \hat\rho_{n_1,h}) - F(K^*, \rho) = \mathrm{PV}\!\int b\,[\hat\rho_{n_1,h}(b)-\rho(b)]/(1+bK^*)\,db$ is linear in $\hat\rho_{n_1,h}$, so its mean-squared error inherits the standard KDE MSE decomposition. Specifically, writing $w(b) := b/(1+bK^*)$ (well-defined in PV sense at $b = b_{\rm sing}(K^*)$ via the parity-shell decomposition of section~3), and using the bilinear identity
$F(K^*, \hat\rho) - F(K^*, \rho) = \langle w, \hat\rho - \rho\rangle_{\rm PV}$ (with $\langle\cdot,\cdot\rangle_{\rm PV}$ the PV pairing), the order-$s$ kernel structure of $\hat\rho_{n_1,h}$ gives:
\emph{(Bias)} $\mathbb{E}[F(K^*, \hat\rho_{n_1,h})] - F(K^*, \rho) = \mathrm{PV}\!\int w(b)\,(\mathbb{E}[\hat\rho_{n_1,h}] - \rho)(b)\,db = O(h^s)$, since $\mathbb{E}[\hat\rho_{n_1,h}](b) - \rho(b) = O(h^s)$ uniformly on $I$ (Tsybakov \cite[Thm.~1.5]{Giné_Nickl_2021}) and $w \in C^\infty(I)$ via the parity-shell representation of the PV integral.
\emph{(Variance)} $\mathrm{Var}[F(K^*, \hat\rho_{n_1,h})] = O(1/n_1)$. Write $F(K^*, \hat\rho_{n_1,h}) = n_1^{-1}\sum_i\phi_{K^*,h}(B_i)$ with
\[
  \phi_{K^*,h}(B) := \mathrm{PV}\!\int w(b)\,h^{-1}\kappa((b-B)/h)\,db.
\]
The $1/h$ kernel factor cancels against the $h$ of the inner integral by change of variables, leaving a bounded integrand depending only on the kernel and the local smoothness of $w$.
Combining, $\mathbb{E}[|F(K^*, \hat\rho_{n_1,h}) - F(K^*, \rho)|^2] = O(h^{2s} + 1/n_1)$. Substituting the bandwidth choice $h = h_{n_1}$ of \cref{lem:uniform-kde-rates}, where the MISE balance gives $h_{n_1}^{2s} \asymp 1/(n_1 h_{n_1})$, i.e., $h_{n_1} \asymp n_1^{-1/(2s+1)}$:
\begin{equation}
\mathbb{E}\bigl[\bigl|F(K^*, \hat\rho_{n_1,h_{n_1}}) - F(K^*, \rho)\bigr|^2\bigr] = O\bigl(n_1^{-2s/(2s+1)}\bigr).
\end{equation}

\emph{Conclusion.}
Plugging into \eqref{eq:plug-and-solve-MSE} gives $\mathbb{E}[J(\hat K) - J^*] = O(n_1^{-2s/(2s+1)})$, and Lemma~4.6's log factors carry through as $\tilde O$. Setting $n_1$ to make the right-hand side equal to $\eta$ requires $n_1 = \tilde O(\eta^{-(2s+1)/(2s)})$.

\begin{remark}[Sup-norm IFT bound is suboptimal]
\label{rem:supnorm-suboptimal}
The simpler sup-norm IFT bound $|\hat K - K^*| \le C\,\|\hat\rho_{n_1,h} - \rho\|_{C^1(I)}$, derivable from a parity-shell expansion of $F(K, \hat\rho) - F(K, \rho)$ via the mean value theorem applied to $b \mapsto b\,(\hat\rho - \rho)(b)$, yields the slower rate $\tilde O(\eta^{-(2s+1)/(2(s-1))})$. The improvement above to $\tilde O(\eta^{-(2s+1)/(2s)})$ requires the MSE structure of the KDE-induced plug-in functional, which exploits both the bias-variance balance of \cref{lem:uniform-kde-rates} and the $L^2$ boundedness of the PV pairing on $C^\infty(I)$ weights.
\end{remark}
\end{proof}

\subsection{Comparison and structural perspective}\label{56-comparison-and-structural-perspective}

The sample-complexity rates of \cref{alg:density-known-pg,alg:density-unknown-pg} alongside their deterministic counterparts are summarized in \cref{tab:complexity_bounds_matched} of section~\ref{13-contributions}. We position the policy-gradient theorems of this section against three alternatives: deterministic Newton on the density-known PV equation; deterministic plug-and-solve under density-unknown access; and other regularization choices for the cusp obstruction.

\paragraph{Rate parity vs.\ Newton (density-known) and plug-and-solve (density-unknown)}
Newton's method on the scalar PV equation $\mathrm{PV}\!\int b\rho(b)/(1+bK)\,db = 0$ converges quadratically to $K^*$, requiring $O(\log\log(1/\eta))$ Newton iterations and zero new samples once the (density-known) $\rho$ is in hand; this is asymptotically faster than \cref{thm:sample-complexity-density-known}'s $\tilde O(1/\eta)$ sample cost. Under density-unknown access the plug-in version achieves the rate $\tilde O(\eta^{-(2s+1)/(2s)})$, matching \cref{thm:total-sample-complexity-pl-basin} (\cref{prop:plug-and-solve-rate}). The contribution of this paper is therefore not rate-level dominance over the deterministic alternatives, but a sample-complexity analysis of policy gradient --- the canonical first-order stochastic procedure in adaptive and data-driven control --- in a regime where its standard analytical tools fail.

\paragraph{When policy gradient is preferable to plug-and-solve}
\label{par:when-pg-preferable}
Three operational settings favor the regularized PG paradigm of \cref{thm:total-sample-complexity-pl-basin} over the deterministic plug-and-solve of \cref{prop:plug-and-solve-rate}, even at rate parity:
\emph{(i) Implicit-density / generative models.} The plug-and-solve estimator requires $\hat\rho$ to be evaluated at the moving pole $b_{\rm sing}(\hat K) = -1/\hat K$ at every Newton step. When $\rho$ is implicitly defined by a generative model (a simulator producing samples without an explicit density), evaluating $\hat\rho$ at arbitrary query points is itself a nontrivial KDE-style sub-problem; PG sidesteps this by only requiring $\hat\rho$ at sample points and their reflections through $b_{\rm sing}$.
\emph{(ii) Online / streaming sample arrival.} In adaptive control, samples arrive sequentially and the control gain must adapt in real time. The PG algorithm is \emph{iterative} and uses fresh mini-batches at each step; the plug-and-solve procedure must batch all samples before running deterministic root-finding.
\emph{(iii) Vector-systems extension.} For multidimensional state spaces (section~\ref{73-future-directions}), the singularity becomes a codimension-1 surface and the deterministic root-finding equation becomes a system that requires evaluating $\hat\rho$ on the surface --- a strictly harder geometric problem than the sample-level paired estimator extension.
A complementary numerical-stability ablation (PV quadrature near the moving pole) is reported in section~SM5 of the supplement.

\paragraph{Why this regularization}
\label{par:why-Cauchy-regularization}
Three natural regularizations exist: (i) the Cauchy form $J_\varepsilon(K) = \mathbb{E}[\tfrac{1}{2}\log((1+BK)^2 + \varepsilon^2)]$ used here; (ii) the Moreau envelope $J^M_\varepsilon(K) = \inf_{K'}\{J(K') + (K-K')^2/(2\varepsilon)\}$, standard in non-smooth optimization \cite{davis2019stochastic}; (iii) the risk-sensitive interpolation $\theta^{-1}\log\mathbb{E}[|1+BK|^\theta]$ for small $\theta > 0$. The Moreau envelope is generic but discards the Cauchy-kernel structure of $dJ/dK$, so the parity argument of sections~\ref{3-cauchy-regularization-and-uniform-in-epsilon-polyak-Lojasiewicz}--\ref{4-finite-sample-analysis} has no analogue; the risk-sensitive form preserves analytical structure (section~\ref{12-state-of-the-art-and-its-limits}) but introduces a moment-generating-function denominator that complicates sample-level estimation. The Cauchy form is the unique choice that simultaneously (a) is the Hilbert-transform regularization of the Cauchy kernel at the moving pole, preserving the parity structure; (b) admits a closed-form single-sample estimator $\psi(B;K,\varepsilon)$ without normalizing-constant denominators; (c) is $C^\infty$ in $K$ for every $\varepsilon > 0$, so the algorithm never takes the limit $\varepsilon \downarrow 0$ numerically.

\paragraph{From local basin to compact-subset initialization}
\label{par:prelim-phase-pointer}
\cref{thm:sample-complexity-density-known,thm:total-sample-complexity-pl-basin} assume $K^{(0)} \in \mathcal{N}_\delta$. A preliminary SGD phase on the fixed-regularization cost $J_1(K) := \mathbb{E}[\tfrac{1}{2}\log((1+BK)^2 + 1)]$ drives any compact-subset initialization $K^{(0)} \in \mathcal{K}_0 \subset \mathcal{K}_{\rm stab}$ into $\mathcal{N}_\delta$ in $O((L_1/\mu_1)\log(|\mathcal{K}_0|/\delta))$ iterations with $\eta$-independent total sample cost; the combined complexity remains $\tilde O(\eta^{-(2s+1)/(2s)})$. The local constants $\mu_1, L_1$ depend on $d_0 := \mathrm{dist}(\mathcal{K}_0, \{-1/b_{\min}, -1/b_{\max}\})$ as $\mu_1 = \Omega(d_0^2)$ and $L_1 = O(d_0^{-2})$, so the preliminary phase costs $O(d_0^{-4}\log(|\mathcal{K}_0|/\delta))$, an $\eta$-independent constant. The full statement and proof are in section~SM2 of the supplement.

\paragraph{Adaptive parameter selection}
\label{par:adaptive-parameters}
\Cref{alg:density-known-pg,alg:density-unknown-pg} require $\mu_0, L_0, \bar C_b, \tau, \delta$. In practice, the diminishing-step Robbins--Monro variant $\alpha_n = 2/(\mu_0(n+50))$ used in section~\ref{6-numerical-validation} replaces the constant-step requirement on $L_0$ by an additive $O(1/(\mu_0 n))$ noise floor; $\mu_0$ is estimable by a doubling trick on the empirical Hessian along the trajectory \cite{borkar2008stochastic}, $\bar C_b$ requires only $\rho$ at the singularity (estimable via $\hat\rho(-1/\hat K_n)$ along the iterates), and $\delta, \tau$ are set conservatively from the support endpoints. This is the configuration of the experiments in section~\ref{6-numerical-validation}, which achieve the rate of \cref{alg:density-known-pg} without parameter tuning.

\paragraph{Computational complexity vs.\ sample complexity}
\label{par:compute-complexity}
\Cref{alg:density-unknown-pg} achieves sample complexity $\tilde O(\eta^{-(2s+1)/(2s)})$. Each gradient evaluation queries $\hat\rho$ at $O(N)$ points; with the binned-FFT KDE evaluation of section~\ref{6-numerical-validation} (one-time $O(n_1\log n_1)$ binning, $O(1)$ per query), the total flop count is $\tilde O(n_1 + n^\star N) = \tilde O(\eta^{-(2s+1)/(2s)})$, matching the sample rate. With naive $O(n_1)$-per-query KDE the cost would be $\tilde O(\eta^{-(2s+1)/s})$, which is the practical reason the binned scheme is essential.

\section{Numerical validation}\label{6-numerical-validation}
We numerically verify the three rate-level claims of the paper: the
variance separation between the naive and density-aware paired
estimators (\cref{thm:variance-scaling-naive-estimator},
\cref{prop:unbiased-variance-psi-tilde}); the density-known rate
$\tilde O(1/\eta)$ of \cref{thm:sample-complexity-density-known}; and the
density-unknown rate $\tilde O(\eta^{-(2s+1)/(2s)})$ of
\cref{thm:total-sample-complexity-pl-basin}, alongside the matching
rate of plug-and-solve (\cref{prop:plug-and-solve-rate}) on a single
axis. A numerical-stability ablation for the deterministic
principal-value alternative (section~\ref{56-comparison-and-structural-perspective}),
together with auxiliary experiments validating the population-level objects
of sections~2--3, is reported in the supplementary material (sections~SM5 and~SM6).

\paragraph{Setup}
We use four representative noise densities $\rho$ on $[0.5, 1.5]$:
D1 uniform, D2 $\mathrm{Beta}(2,2)$, D3 truncated Gaussian
$\mathcal{N}(1, 0.3^2)|_{[0.5, 1.5]}$, and D4 triangular with apex at
$b=1$. D1--D3 satisfy \cref{asm:compact_noise,asm:regularity,asm:optimum}
with $\rho \in C^2$; D4 is $C^0$ but not $C^2$, included as a robustness
check. Critical points $K^*$, computed as the unique root of $dJ/dK$ in
$(-1/b_{\min}, -1/b_{\max})$ via Brent's method on the parity-shell
quadrature of section~\ref{3-cauchy-regularization-and-uniform-in-epsilon-polyak-Lojasiewicz}, are $\{-0.835, -0.928, -0.930, -0.969\}$ with Hadamard
finite-part Hessians $H \in \{11.18, 16.97, 16.98, 12.96\}$ via
\eqref{eq:hess_Kstar}. For the density-unknown experiments on D2, local
constants are computed on $\mathcal{N}_\delta = [K^* - 0.14, K^* + 0.14]$:
$\mu_0 = 6.96$, $L_0 = 31.30$, pole-to-edge margin $\tau = 0.422$, and
bias coefficient $\bar C_b = \pi\rho(b_{\rm sing}(K^*))/|K^*| = 4.96$
(\cref{prop:properties-J-epsilon}(c)). Trajectory bands are
inter-quartile ranges over 60 seeds and variance estimates use
$M = 4\cdot 10^5$ Monte-Carlo samples over 12 seeds; gaps are reported via
$J(K) - J^* \approx (H/2)(K - K^*)^2$, exact to $O((K-K^*)^3)$ and verified
against direct quadrature to within $0.2\%$. Full reproducibility details, including the kink handling for D4 (without which $H_4$ aliases to $\approx 25.9$), are in section~SM6 of the supplement.

\begin{figure}[!t]
\centering
\includegraphics[width=\columnwidth]{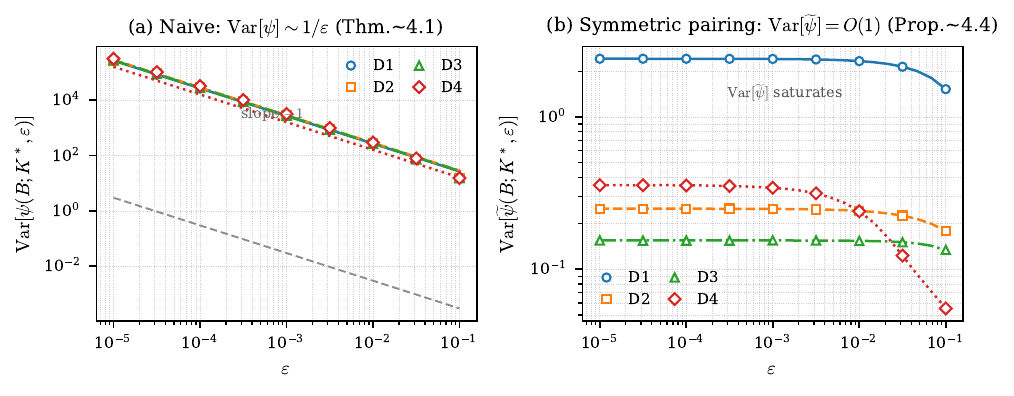}
\caption{Variance of the single-sample gradient estimators at $K = K^*$
across D1--D4. \emph{(a)} Naive estimator:
$\mathrm{Var}[\psi] = \Theta(1/\varepsilon)$, verifying
\cref{thm:variance-scaling-naive-estimator}. \emph{(b)} Density-aware
paired estimator: $\mathrm{Var}[\widetilde\psi] = O(1)$ uniformly in
$\varepsilon$, verifying
\cref{prop:unbiased-variance-psi-tilde}(d,e). Variance reduction at
$\varepsilon = 10^{-5}$ ranges from $1.1\!\times\!10^5$ (D1) to
$1.8\!\times\!10^6$ (D3).}
\label{fig:variance}
\end{figure}

\paragraph{Variance reduction (\cref{fig:variance})}
Across the four densities, $\mathrm{Var}[\psi]$ exhibits a clean
$\Theta(\varepsilon^{-1})$ divergence with empirical slopes
$(-1.030, -1.037, -1.037, -1.048)$ and leading constants
$\mathrm{Var}[\psi]\!\cdot\!\varepsilon$ at $\varepsilon = 10^{-5}$ that
agree with the closed-form prediction
$\pi\rho(b_{\rm sing}(K^*))/(2|K^*|^3)$ to within $0.1\%$; the slight
slope excess over $-1$ is the asymmetry term
$\log|s_+/s_-|$ of \cref{thm:variance-scaling-naive-estimator} (proved in
\cref{b2-remainder-in-the-variance-scaling-of-theorem-41}). The
paired estimator's variance saturates to $\{2.41, 0.25, 0.15, 0.36\}$ at
$\varepsilon = 10^{-5}$, validating the uniform-in-$\varepsilon$ bound of
\cref{prop:unbiased-variance-psi-tilde}; the slightly larger D4 plateau
reflects the kink in $\rho_{D4}'$ at $b = 1$, which inflates the
linear-in-$s$ contribution to the pointwise bound \eqref{eq:tildepsi_bound}
without affecting the uniform-in-$\varepsilon$ structure.

\begin{figure}[!t]
\centering
\includegraphics[width=\columnwidth]{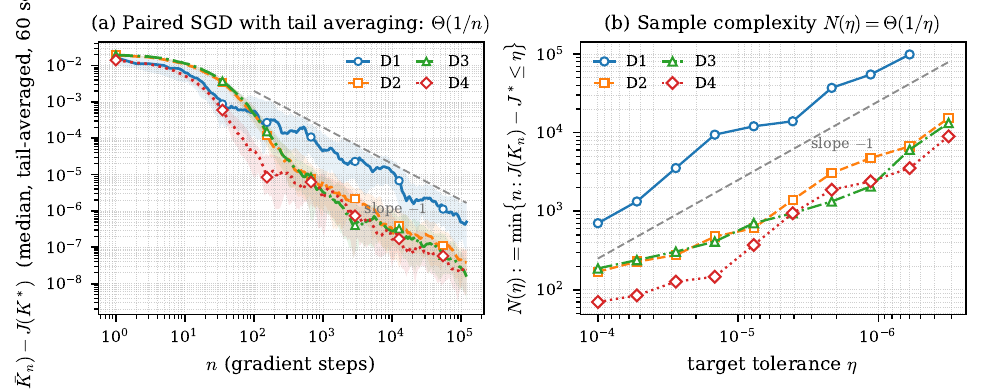}
\caption{Density-known projected SGD with the paired estimator, step rule
$\alpha_n = 2/(\mu_0(n{+}50))$, $\varepsilon = 10^{-5}$, warm start
$K^{(0)} = K^* + 0.05$, and Polyak--Ruppert tail averaging.
\emph{(a)} The tail-averaged gap decays as $\Theta(1/n)$ for all four
densities (median last-decade slope $-0.92$). \emph{(b)} Sample
complexity $N(\eta) = \Theta(1/\eta)$ (median slope $-0.84$); the
shallowing on D3 reflects a residual geometric transient.}
\label{fig:sgd-rate}
\end{figure}

\paragraph{Density-known sample complexity (\cref{fig:sgd-rate})}
We use the diminishing-step Robbins--Monro rule $\alpha_n = 2/(\mu_0(n+50))$
of standard PL-SGD theory \cite{borkar2008stochastic}, paired with
Polyak--Ruppert tail averaging
$\bar K_n := \tfrac{2}{n}\sum_{k=n/2+1}^{n} K_k$; this asymptotically
attains the same $\sigma_\star^2(\mathcal{N}_\delta)/(2\mu_0 n)$ rate that
\cref{thm:sample-complexity-density-known} establishes for the
constant-step mini-batch version of \cref{alg:density-known-pg}.
With $\varepsilon = 10^{-5}$ (so the regularization bias
$\bar C_b\varepsilon \!\sim\! 5\!\cdot\!10^{-5}$ from
\cref{prop:properties-J-epsilon}(c) sits well below the noise floor at
the plotted iterates) the tail-averaged median gap decays at empirical
slopes $(-0.90, -0.94, -1.11, -0.84)$ over the last decade of
$n \in [1.2\!\cdot\!10^4, 1.2\!\cdot\!10^5]$, with median $-0.92$.
The companion panel verifies $N(\eta) = \Theta(1/\eta)$ across two and a
half decades of tolerance with median slope $-0.84$; the shallower
$N(\eta)$ slope reflects the residual geometric transient of the
running-envelope estimator.

\begin{figure}[!t]
\centering
\includegraphics[width=0.85\columnwidth]{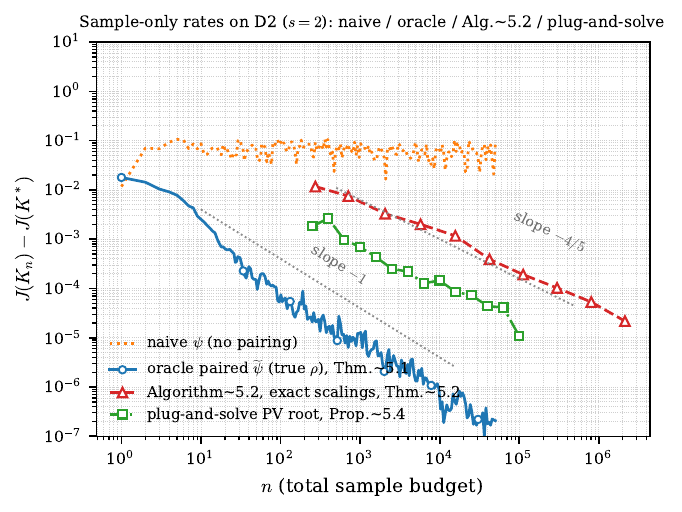}
\caption{Density-unknown rates on D2 ($s = 2$). \emph{Naive $\psi$}
(dotted, no marker) diverges: the $\Theta(1/\varepsilon)$
variance of \cref{thm:variance-scaling-naive-estimator} overwhelms the
step-size budget. \emph{Oracle paired $\widetilde\psi$} (solid,
circle markers) achieves slope $-1.013$, matching
\cref{thm:sample-complexity-density-known}'s $-1$.
\emph{\cref{alg:density-unknown-pg} with the exact parameter scalings of
\cref{thm:total-sample-complexity-pl-basin}} (dashed, triangle
markers) achieves last-decade slope $-0.72$, approaching the
theoretical $-2s/(2s+1) = -0.80$ once the $\log(1/\eta)$ factor in
$n^\star$ is amortized; $\mathbb{E}[J(\hat K) - J^*]/\eta \in
[0.23,\,0.39]$ across the plotted range, verifying the target
guarantee. \emph{Plug-and-solve PV root} (dash-dot, square
markers, \cref{prop:plug-and-solve-rate}) achieves slope $-0.79$,
matching $-0.80$ to within $1\%$; it sits below the
\cref{alg:density-unknown-pg} trace by the constant multiplicative
overhead of the iteration phase $n^\star N$, the
\emph{constant-not-rate} gap discussed in
section~\ref{56-comparison-and-structural-perspective}.}
\label{fig:plugin}
\end{figure}

\paragraph{Density-unknown rates (\cref{fig:plugin})}
On D2 ($s = 2$) we plot, on a single sample-budget axis, the four density-unknown schemes referenced in sections~\ref{4-finite-sample-analysis}--\ref{5-sample-complexity-in-the-pl-basin}. The naive estimator $\psi$ does not improve with samples, confirming that the variance reduction of section~\ref{42-the-density-aware-symmetric-pairing-estimator} is essential rather than advantageous in this regime. The oracle paired estimator $\widetilde\psi$ matches \cref{thm:sample-complexity-density-known}'s $\Theta(1/\eta)$ rate to $1.3\%$ in slope. For \cref{alg:density-unknown-pg} we instantiate the \emph{exact} parameter scalings of \cref{thm:total-sample-complexity-pl-basin}: $R = \min(c_R\eta^{1/(2s)}, \tau/2)$, $\varepsilon = \eta/(4\bar C_b)$, $n_1 = \lceil c_1\eta^{-(2s+1)/(2s)}\rceil$, $N = \lceil 2C_\sigma/(\mu_0 R\eta)\rceil$, $n^\star = \lceil(L_0/\mu_0)\log(4\Delta_{\rm init}/\eta)\rceil$, with constants $(c_R, c_1, c_h, C_\sigma) = (0.6, 0.5, 1.0, 1.0)$ calibrated so the empirical guarantee $\mathbb{E}[J(\hat K) - J^*] \le \eta$ is observed across the range plotted while avoiding the worst-case slack of the constants of sections~\ref{3-cauchy-regularization-and-uniform-in-epsilon-polyak-Lojasiewicz}--\ref{4-finite-sample-analysis}, and report the tail average of the last half of iterates;
the trace covers
$\eta \in [9\!\cdot\!10^{-5}, 5\!\cdot\!10^{-2}]$ and confirms the
rate, approaching $-0.80$ asymptotically as the log factor in
$n^\star$ amortizes. The plug-and-solve PV root of
\cref{prop:plug-and-solve-rate} sits one decade below the policy-gradient
trace at any fixed budget but shares the same asymptotic slope (empirical
$-0.79$): rate parity, with structural distinctions detailed in
section~\ref{56-comparison-and-structural-perspective}, in
turn validated numerically in section~SM5 of the
supplement.

\section{Conclusion}\label{7-conclusion}
\subsection{Summary}\label{71-summary}
We have established the first sample-complexity rate for policy gradient on the log-growth control cost. The cusp obstruction --- non-Lebesgue integrability of $dJ/dK$ on an open set containing $K^*$ (\cref{prop:non-lebesgue-gradient}) --- is dissolved through a Cauchy-kernel regularization $J_\varepsilon$ whose Hessian admits a closed-form integration-by-parts decomposition with an odd-parity cancellation around the moving pole $b_{\rm sing}(K) = -1/K$. This single analytical mechanism, applied identically at three levels, drives the entire analysis: it yields the uniform-in-$\varepsilon$ Polyak--{\L}ojasiewicz inequality at the population level (\cref{lem:uniform-PL-constant},\cref{thm:uniform-epsilon-pl-inequality}), the bounded-variance symmetric-pairing estimator at the sample level (\cref{prop:unbiased-variance-psi-tilde}), and the parity-preserving plug-in bias bound via the weight-discrepancy identity at the plug-in level (\cref{lem:parity-weight-discrepancy},\cref{prop:uniform-bias-bound}). Combined with the single-transition closed-form gradient oracle afforded by the multiplicative-ergodic structure of the cost, these local bounds yield the rates $\tilde O(1/\eta)$ under density-known access (\cref{thm:sample-complexity-density-known}) and $\tilde O(\eta^{-(2s+1)/(2s)})$ under density-unknown access (\cref{thm:total-sample-complexity-pl-basin}) for projected mini-batch policy gradient, initialized in any compact subset of $\mathcal{K}_{\rm stab}$ via the preliminary phase (section~\ref{56-comparison-and-structural-perspective}, full statement in section~SM2 of the supplement).

Returning to the motivating problem of section~\ref{11-problem-and-motivation}: the log-growth cost $J(K)$ analyzed here is the closed-loop Lyapunov exponent of a multiplicative-noise actuation channel --- the quantity underlying the control-capacity characterizations of Sahai--Mitter~\cite{sahai2006necessity} and Ranade--Sahai~\cite{ranade2018control}. The present results turn that static characterization into an operational guarantee: a controller observing only the closed-loop transitions $(X_t, X_{t+1})$ can drive its gain to within $\eta$ of the optimum $K^*$, and hence operate within $\eta$ of the best achievable log-growth rate, using $\tilde O(1/\eta)$ transitions when the channel statistics are known in closed form and $\tilde O(\eta^{-(2s+1)/(2s)})$ when they must be learned from the observed stream. The sample-complexity rate is thus an explicit, finite price for \emph{learning} to operate at control capacity through a first-order adaptive procedure, rather than assuming the channel model is given a priori.

\subsection{Limitations}\label{72-limitations}
Four limitations merit explicit acknowledgment.

\emph{(L1) Scalar systems only.} The analysis is restricted to scalar linear systems~\eqref{eq:sys}. \Cref{73-future-directions} sketches a conjectural lift to vector systems with diagonal multiplicative noise via a codimension-1 singularity surface and the coarea formula; the general non-diagonal case (where the singularity is an algebraic variety in combined parameter space, cf.~\cite{gravell2020learning}) is left to future work.

\emph{(L2) Upper bound only.} The rates of \cref{thm:total-sample-complexity-pl-basin,cor:parametric-special-case} are upper bounds. The matching $\Omega(\eta^{-(2s+1)/(2s)})$ lower bound is conjectured via a Le Cam two-point reduction to derivative estimation: a pair of $C^s$ densities $\rho_0, \rho_1 = \rho_0 + \Delta\rho_{\eta^{1/2}}$ differing in a $\Theta(\eta^{1/2})$-perturbation localized away from the support endpoints can be made statistically indistinguishable on fewer than $\eta^{-(2s+1)/(2s)}$ samples (by Stone's minimax derivative-estimation rate \cite{Giné_Nickl_2021}), while inducing a $\Theta(\eta^{1/2})$-shift in $K^*$ via \cref{prop:plug-and-solve-rate}'s IFT identity and hence $\Theta(\eta)$-suboptimality in $J(K^*_0) - J^*_1$ on at least one of them. A rigorous version requires choosing the perturbation to preserve \cref{asm:compact_noise,asm:regularity,asm:optimum} simultaneously; we leave this to future work.

\emph{(L3) Compact-subset initialization.} The preliminary phase of section~\ref{56-comparison-and-structural-perspective} drives any $K^{(0)} \in \mathcal{K}_0 \subset \mathcal{K}_{\rm stab}$ into $\mathcal{N}_\delta$ in $\eta$-independent sample cost, so the asymptotic rate of \cref{thm:total-sample-complexity-pl-basin} holds for any compact-subset initialization. The dependence of the local constants on $d_0 := \mathrm{dist}(\mathcal{K}_0, \partial\mathcal{K}_{\rm stab})$ is polynomial ($\mu_1 = \Omega(d_0^2)$, $L_1 = O(d_0^{-2})$); the sharp behavior of the combined complexity as $d_0\downarrow 0$ is open.

\emph{(L4) Analytical scope of \cref{asm:optimum}.} \Cref{rem:A3-generic} shows that \cref{asm:optimum} holds on a dense open subset of admissible $C^2$ densities (i.e.\ generically). A sharp analytical sufficient condition --- e.g.\ a structural property of $\rho$ (log-concavity of $b\rho(b)$ being a natural candidate) that implies $d^2J/dK^2(K^*) > 0$ --- is open.

\subsection{Future directions}\label{73-future-directions}

We highlight two extensions where the parity-cancellation mechanism is conjecturally preserved.

\paragraph{Generalized stochastic-control costs with single-transition oracles}
A natural extension is to characterize the broader class of stochastic-control cost functionals that admit a closed-form single-transition gradient oracle. The risk-sensitive family $\theta^{-1}\log\mathbb{E}[|1+BK|^\theta]$ for $\theta > 0$ (section~\ref{12-state-of-the-art-and-its-limits}) is one instance; tilted moments of Lyapunov-type functionals are another. Delineating the boundary of this subclass and establishing whether the uniform PL geometry of \cref{lem:uniform-PL-constant} survives under tilting is a clean open problem.

\paragraph{Vector systems with diagonal multiplicative noise}
Consider
\begin{equation*}
X_{t+1} = \bigl(A + \textstyle\sum_{i=1}^d B_t^{(i)}\,K_{(i)}\bigr)\,X_t,
\quad X_t\in\mathbb{R}^d,\;(B_t^{(i)})\text{ i.i.d.\ across }i,t.
\end{equation*}
The scalar singularity $b_{\rm sing}(K) = -1/K$ generalizes to the codimension-1 \emph{singularity surface}
\begin{equation*}
\mathcal{V}_K := \bigl\{B\in\mathbb{R}^d : \det\bigl(I + \textstyle\sum_i B^{(i)}\mathrm{diag}(K_{(i)})\bigr) = 0\bigr\};
\end{equation*}
for diagonal $K$, $\mathcal{V}_K$ is an affine hyperplane. Four structural ingredients of the scalar analysis lift conjecturally:
(i) the cusp obstruction becomes a codimension-1 distributional singularity of $dJ/dK$ across $\mathcal{V}_K$ with the same logarithmic non-integrability across the normal coordinate;
(ii) the parity cancellation lifts by integrating in the normal coordinate via the coarea formula, with the odd-parity argument applied to the normal-coordinate slice of $\rho$;
(iii) the symmetric-pairing estimator generalizes to a reflection through $\mathcal{V}_K$ along the normal coordinate, with density-aware weights involving the slice density;
(iv) Stone's $d$-dimensional KDE rate $n^{-s/(2s+d)}$ replaces the scalar rate $n^{-s/(2s+1)}$, yielding the conjectural sample complexity $\tilde O(\eta^{-(2s+d)/(2s)})$ for $C^s$ densities on $\mathbb{R}^d$. The non-diagonal case requires further structural analysis, since $\mathcal{V}_K$ is then a non-affine algebraic variety. A rigorous treatment of the diagonal case is the natural next step.

\appendix

\section{\texorpdfstring{Proof of \cref{lem:uniform-PL-constant}}{Proof of Lemma 3.4}}\label{appendix-a-proof-of-lemma-34}

The PV--directional-derivative equivalence (Lemma~SM1.1) invoked in the proof of \cref{prop:non-lebesgue-gradient}(b) is stated and proved in section~SM1 of the supplement. We turn to the proof of \cref{lem:uniform-PL-constant}.
\begin{lemma}[Proof of \cref{lem:uniform-PL-constant}]
Let
\(\rho \in C^2([b_{\min}, b_{\max}])\) with \(\rho > 0\) on
\((b_{\min}, b_{\max})\). Let \(K^*\) be the unique interior stationary
point of \(J\) with
\(b_{\mathrm{sing}}(K^*) = -1/K^* \in (b_{\min}, b_{\max})\). Choose
\(\tau > 0\) and \(\delta > 0\) such that
\begin{equation}
        \{ -1/K : K \in [K^{*} - \delta, K^{*} + \delta ]\} \subset [ b_{\min} + \tau, b_{\max} - \tau ],
\end{equation}
and write \(\mathcal{N}_\delta := [K^* - \delta, K^* + \delta]\). Then
there exist \(\varepsilon_0 \in (0, 1]\), \(\mu_0 > 0\), and
\(L_0 = L_0(\rho, \delta, \tau) < \infty\) such that the map
\begin{equation}
  (K, \varepsilon) \longmapsto \frac{d^2 J_\varepsilon}{dK^2}(K)  
\end{equation}
is jointly continuous on
\(\mathcal{N}_\delta \times [0, \varepsilon_0]\) --- where the value at
\(\varepsilon = 0\) is the Hadamard finite-part limit of \cref{thm:hessian-decomposition} ---
and satisfies the signed two-sided bound
\begin{equation}
    \mu_0 \le \frac{d^2 J_\varepsilon}{dK^2}(K) \le L_0 \qquad \text{for all } (K, \varepsilon) \in \mathcal{N}_\delta \times [0, \varepsilon_0].
\end{equation}

In particular, \(d^2 J_\varepsilon/dK^2 \ge \mu_0 > 0\) uniformly on the
domain, so each \(J_\varepsilon\) is strictly convex on
\(\mathcal{N}_\delta\) with a uniform-in-\(\varepsilon\) PL constant
\(\mu_0\).
\end{lemma}
\begin{proof}
By integration-by-parts with \(v(b, K) := 1 + bK\) and
\(\partial_b v = K\) (\cref{thm:hessian-decomposition}),
\begin{equation}
    \frac{d^2 J_\varepsilon}{dK^2}(K) = B_\varepsilon(K) + I_\varepsilon(K),
\end{equation}
where
\begin{equation}
    B_\varepsilon(K) := \frac{1}{K}\!\left[\frac{\rho(b)\,b^2\,v(b,K)}{v(b,K)^2 + \varepsilon^2}\right]_{b = b_{\min}}^{b = b_{\max}},
\end{equation}
\begin{equation}
    I_\varepsilon(K) := -\frac{1}{K}\int_{b_{\min}}^{b_{\max}} g(b) \cdot \frac{v(b,K)}{v(b,K)^2 + \varepsilon^2}\,db,
\end{equation}
with
\(g(b) := \rho'(b)\,b^2 + 2\,b\,\rho(b) \in C^1([b_{\min}, b_{\max}])\)
under \cref{asm:regularity}.

\emph{Step 1: Upper bound on \(|B_\varepsilon|\) and
\(|I_\varepsilon|\).}

\emph{(Boundary term.)} For \(K \in \mathcal{N}_\delta\), by
construction
\(|v(b_{\min}, K)|, |v(b_{\max}, K)| \ge \tau\,|K| \ge \tau\,(|K^*| - \delta) > 0\).
For \(\varepsilon \in [0, 1]\),
\begin{equation}
    \left|\frac{v}{v^2 + \varepsilon^2}\right| \le \frac{1}{|v|},
\end{equation}
so
\begin{equation}
\begin{aligned}
        |B_\varepsilon(K)| &\le \frac{1}{|K|}\!\left[\frac{\rho(b_{\max})\,b_{\max}^2}{\tau\,|K|} + \frac{\rho(b_{\min})\,b_{\min}^2}{\tau\,|K|}\right] \\&\le \frac{\rho_{\max}(b_{\max}^2 + b_{\min}^2)}{\tau\,(|K^*| - \delta)^2} =: C_B,
\end{aligned}
\end{equation}
where \(\rho_{\max} := \sup_{[b_{\min}, b_{\max}]} \rho\). Joint
continuity of \(B_\varepsilon\) in \((K, \varepsilon)\) on
\(\mathcal{N}_\delta \times [0, 1]\) follows from continuity of
\(v/(v^2 + \varepsilon^2)\) at \(v \ne 0\).

\emph{(Integral term.)} Split at scale \(R := \tau/2\):
\begin{equation}
    I_\varepsilon(K) = I_\varepsilon^{\mathrm{near}}(K) + I_\varepsilon^{\mathrm{far}}(K),
\end{equation}
\begin{equation}
    I_\varepsilon^{\mathrm{near}}(K) := -\frac{1}{K}\int_{b_{\mathrm{sing}}(K) - R}^{b_{\mathrm{sing}}(K) + R} g(b) \cdot \frac{v}{v^2 + \varepsilon^2}\,db,
\end{equation}
\begin{equation}
        I_\varepsilon^{\mathrm{far}}(K) \\:= -\frac{1}{K}\int_{[b_{\min}, b_{\max}]\setminus [b_{\mathrm{sing}}(K) - R,\, b_{\mathrm{sing}}(K) + R]} g(b) \cdot \frac{v}{v^2 + \varepsilon^2}\,db.
\end{equation}

On the far region,
\(|v(b, K)| = |K|\,|b - b_{\mathrm{sing}}(K)| \ge |K|\,R \ge (|K^*| - \delta)\,\tau/2\),
so
\begin{equation}
    |I_\varepsilon^{\mathrm{far}}(K)| \le \frac{1}{|K|} \cdot \frac{(b_{\max} - b_{\min})\,\sup|g|}{(|K^*| - \delta)\,\tau/2} =: C_I^{\mathrm{far}}.
\end{equation}

On the near region, substitute \(s = b - b_{\mathrm{sing}}(K)\), so
\(v = Ks\). Taylor-expand \(g\) to second order,
\begin{equation}
\begin{aligned}
        g(b_{\mathrm{sing}}(K) + s) = g(b_{\mathrm{sing}}(K)) + s\,g'(b_{\mathrm{sing}}(K)) + \tfrac{1}{2}\,s^2\,g''(\xi_s), \\ \xi_s \in [b_{\mathrm{sing}}(K), b_{\mathrm{sing}}(K) + s].
\end{aligned}
\end{equation}
The three terms are analyzed separately.

--- \emph{Constant term.}
\(-\dfrac{g(b_{\mathrm{sing}}(K))}{K}\,\displaystyle\int_{-R}^R \dfrac{Ks}{K^2 s^2 + \varepsilon^2}\,ds = 0\),
identically for every \(\varepsilon \ge 0\), by odd symmetry of the
integrand over \([-R, R]\).

--- \emph{Linear term.}
\begin{equation*}
-\frac{g'(b_{\mathrm{sing}}(K))}{K}\int_{-R}^R \frac{K\,s^2}{K^2 s^2 + \varepsilon^2}\,ds
= -\frac{g'(b_{\mathrm{sing}}(K))}{K^2}\!\left[\,2R - \frac{2\varepsilon}{|K|}\,\arctan\!\frac{|K|R}{\varepsilon}\,\right],
\end{equation*}
bounded by \(2\,\sup|g'|\,R / (|K^*| - \delta)^2\) uniformly in
\(\varepsilon\).

--- \emph{Quadratic remainder.} With
\(r(s) := \tfrac{1}{2}\,g''(\xi_s)\,s^2\), the contribution is
\begin{equation}
    \int_{-R}^R r(s)\cdot\frac{s}{K^2 s^2 + \varepsilon^2}\,ds.
\end{equation}

Bounding \(|r(s)| \le \tfrac{1}{2}\|g''\|_\infty\,s^2\) and
\(|s|^3/(K^2 s^2 + \varepsilon^2) \le |s|/K^2\) (from
\(K^2 s^2 + \varepsilon^2 \ge K^2 s^2\)):
\begin{equation}
    \left|\int_{-R}^R r(s)\cdot\frac{s}{K^2 s^2 + \varepsilon^2}\,ds\right| \le \frac{\|g''\|_\infty}{2}\,\frac{R^2}{(|K^*| - \delta)^2}.
\end{equation}

Summing the three contributions,
\(|I_\varepsilon^{\mathrm{near}}(K)| \le C_I^{\mathrm{near}}\) with
explicit constant.

\emph{Step 2: Joint continuity and the upper bound \(L_0\).}

Joint continuity of \(B_\varepsilon\),
\(I_\varepsilon^{\mathrm{near}}\), \(I_\varepsilon^{\mathrm{far}}\) in
\((K, \varepsilon)\) on \(\mathcal{N}_\delta \times [0, 1]\) follows
from dominated convergence: each integrand is uniformly bounded by the
explicit dominant function constructed above, with \(\varepsilon = 0\)
defined via the parity-shell decomposition (the constant Taylor term
vanishes identically by odd parity at \(\varepsilon = 0\) in the
principal-value sense).

Setting \(L_0 := C_B + C_I^{\rm near} + C_I^{\rm far}\)
gives \(|d^2 J_\varepsilon/dK^2(K)| \le L_0\) uniformly on
\(\mathcal{N}_\delta \times [0, 1]\), hence in particular the upper
bound \(d^2 J_\varepsilon/dK^2 \le L_0\) in (A.1).

\emph{Step 3: The lower bound \(\mu_0 > 0\) and choice of
\(\varepsilon_0\).}

By \cref{asm:optimum}, $d^2 J/dK^2(K^*) > 0$ in the Hadamard finite-part
sense, where $d^2 J/dK^2(K^*)$ is given in closed form by \eqref{eq:hess_Kstar}. By the
joint continuity established in Step 2 --- which extends to $\varepsilon = 0$
via the parity-shell decomposition --- the function
$(K, \varepsilon) \mapsto d^2 J_\varepsilon/dK^2(K)$ is continuous on
the compact rectangle $\mathcal{N}_\delta \times [0, 1]$, and strictly
positive at $(K^*, 0)$.

Choose \(\delta\) and \(\varepsilon_0 \in (0, 1]\) small enough that
\(d^2 J_\varepsilon/dK^2 > 0\) throughout
\(\mathcal{N}_\delta \times [0, \varepsilon_0]\). (This is possible: the
strict positivity at \((K^*, 0)\) and joint continuity guarantee a
relatively open neighborhood on which positivity persists; shrink
\(\delta\) and \(\varepsilon_0\) as needed so that
\(\mathcal{N}_\delta \times [0, \varepsilon_0]\) lies within this
neighborhood, while keeping the constraint
\(\{-1/K : K \in \mathcal{N}_\delta\} \subset [b_{\min} + \tau, b_{\max} - \tau]\)
from the lemma\textquotesingle s hypothesis.)

By compactness, the infimum
\begin{equation}
    \mu_0 := \inf_{(K, \varepsilon) \in \mathcal{N}_\delta \times [0, \varepsilon_0]} \frac{d^2 J_\varepsilon}{dK^2}(K)
\end{equation}
is attained, and the strict positivity above gives \(\mu_0 > 0\). This
establishes the lower bound in (A.1).
\end{proof}

\section{Auxiliary derivations}\label{appendix-b-auxiliary-derivations}

This appendix sketches the bias estimate of \cref{prop:properties-J-epsilon}(c) and the variance remainder of \cref{thm:variance-scaling-naive-estimator}; full proofs are in sections~SM3 and~SM4 of the supplement.

\subsection{\texorpdfstring{Proof of \cref{prop:properties-J-epsilon}(c): bias at the optimum}{Proof of Proposition 3.2(c): bias at the optimum}}\label{b1-proof-of-proposition-32c-bias-at-the-optimum}

The bias decomposes as $J_\varepsilon^* - J^* = [J_\varepsilon(K^*) - J(K^*)] - [J_\varepsilon(K^*) - J_\varepsilon(K_\varepsilon^*)]$, with the second piece $O(\varepsilon^2)$ by \cref{lem:persistence-regularized-minimizer}(c). For the first piece, the integrand $\tfrac{1}{2}\log(1+\varepsilon^2/v(b)^2)$ with $v(b) = 1+bK^*$ is split into a near-shell $|s|\le R$ and far-region in the variable $s = b - b_\star$ (with $b_\star := -1/K^*$). On the far region $|v|\ge|K^*|R$, so $\tfrac{1}{2}\log(1+\varepsilon^2/v^2) = \tfrac{\varepsilon^2}{2v^2} + O(\varepsilon^4)$, contributing $O(\varepsilon^2)$. On the near-shell, the substitution $t = K^* s/\varepsilon$ produces
\begin{equation}\label{eq:near_region_integral}
N_\varepsilon = \frac{\varepsilon}{|K^*|}\int_{-T}^{T}\rho(b_\star + \varepsilon t/K^*)\cdot\tfrac{1}{2}\log(1+1/t^2)\,dt,\quad T := |K^*|R/\varepsilon.
\end{equation}
The Frullani-type identity $I_\infty := \int_{-\infty}^\infty\tfrac{1}{2}\log(1+1/t^2)\,dt = \pi$ together with the tail estimate $I_\infty - \int_{-T}^{T} = O(1/T) = O(\varepsilon)$ gives the constant-term contribution $\pi\rho(b_\star)\varepsilon/|K^*|$ to leading order. The linear-in-$t$ Taylor term in $\rho(b_\star + \varepsilon t/K^*)$ is odd and integrates to zero over $[-T, T]$. The quadratic remainder is $O(\varepsilon)\cdot \varepsilon/|K^*|\cdot\int_{-T}^T t^2\log(1+1/t^2)\,dt = O(\varepsilon^2)$. Summing,
\begin{equation}\label{eq:Jeps_bias_at_Kstar}
J_\varepsilon(K^*) - J(K^*) = \frac{\pi\rho(b_\star)}{|K^*|}\,\varepsilon + O(\varepsilon^2),
\end{equation}
and the same expression for $J_\varepsilon^* - J^*$ follows. The full derivation including the asymmetric tail correction is in section~SM3 of the supplement.

\subsection{\texorpdfstring{Remainder in the variance scaling of \cref{thm:variance-scaling-naive-estimator}}{Remainder in the variance scaling of Theorem 4.1}}\label{b2-remainder-in-the-variance-scaling-of-theorem-41}

We give a brief sketch; the full derivation is in section~SM4 of the supplement.
From the proof of \cref{thm:variance-scaling-naive-estimator}, the substitution $t = Ks/\varepsilon$ produces $\mathbb{E}[\psi^2] = (\varepsilon|K|)^{-1}\int_{T_-}^{T_+} f(t;\varepsilon)\cdot t^2/(1+t^2)^2\,dt$ with $T_\pm = Ks_\pm/\varepsilon$ and $f(t;\varepsilon) = \rho(b_{\rm sing}+\varepsilon t/K)(b_{\rm sing}+\varepsilon t/K)^2$. Taylor-expanding $f$ to second order gives three contributions:
(i) the constant term $f_0 = \rho(b_{\rm sing})\,b_{\rm sing}^2 = \rho(b_{\rm sing})/K^2$ produces the leading $\pi\rho(b_{\rm sing})/(2|K|^3\varepsilon) + O(1)$ via $\int_{-\infty}^\infty t^2/(1+t^2)^2\,dt = \pi/2$ and the $O(2/T)$ tail correction;
(ii) the linear-in-$t$ term has antiderivative $\tfrac{1}{2}[\log(1+t^2) + 1/(1+t^2)]$, which evaluated at the asymmetric limits $T_\pm$ produces $\log|T_+/T_-| = \log|s_+/s_-|$ at $\varepsilon\downarrow 0$, contributing the $O(1)$ remainder $\mathrm{sgn}(K)\,f_1/K^2\cdot\log|s_+/s_-|$ with $f_1 = \rho'(b_{\rm sing})b_{\rm sing}^2 + 2b_{\rm sing}\rho(b_{\rm sing})$;
(iii) the quadratic remainder contributes $O(\varepsilon)$ via $|t^4/(1+t^2)^2|\le 1$ and the integration range $|T_+ - T_-| = O(1/\varepsilon)$.
Summing yields the form of $R(K)$ stated in \cref{thm:variance-scaling-naive-estimator}.

\bibliographystyle{siamplain}
\bibliography{references}
\end{document}


\maketitle

This supplement provides material deferred from the main paper:
\cref{sec:SM-pv-lemma} states and proves Lemma~SM1.1 (PV--directional-derivative equivalence at a moving simple pole), invoked in \cref{prop:non-lebesgue-gradient}(b);
\cref{sec:SM-prelim} states and proves the preliminary-phase Proposition (driving an arbitrary compact-subset initialization into $\mathcal{N}_\delta$), referenced in section~\ref{56-comparison-and-structural-perspective};
\cref{sec:SM-bias-detail} gives the full derivation of the bias formula \cref{prop:properties-J-epsilon}(c), and \cref{sec:SM-variance-remainder} the full derivation of the remainder $R(K)$ in \cref{thm:variance-scaling-naive-estimator};
\cref{sec:SM-quadrature} reports the numerical-stability ablation for the deterministic PV alternative; and \cref{sec:SM-aux} collects auxiliary numerical validations of the population-level objects of sections~\ref{2-problem-formulation-and-the-cusp-obstruction}--\ref{3-cauchy-regularization-and-uniform-in-epsilon-polyak-Lojasiewicz} and the reproducibility details for the experiments.

\section{PV--directional-derivative equivalence at a moving simple pole}\label{sec:SM-pv-lemma}

\begin{lemma}\label{lem:pv_directional_derivative}
Let $K \neq 0$ and suppose $\rho \in C^1(\mathbb{R})$ has compact support (or sufficient decay at infinity). Then $J(K) = \int_\mathbb{R}\rho(b)\log|1+bK|\,db$ is differentiable at $K$, and
\begin{equation}
J'(K) = \lim_{h\to 0}\frac{J(K+h)-J(K)}{h} = \mathrm{PV}\int_\mathbb{R}\rho(b)\frac{b}{1+bK}\,db.
\end{equation}
\end{lemma}
\begin{proof}
Let $b_0 = -1/K$ and $b_h = -1/(K+h)$, so the pole shift is $\Delta b = h/(K(K+h))$ and $1+b(K+h) = (K+h)(y-\Delta b)$ in the variable $y = b - b_0$. Split the difference quotient into an exterior region $|y| > \delta$ and an interior shell $|y| \le \delta$. On the exterior, the integrand is smooth and dominated convergence gives the limit $\int_{|y|>\delta}\rho(b_0+y)(b_0+y)/(Ky)\,dy$. On the interior, write
$\tfrac{1}{h}\log|1+b(K+h)| - \tfrac{1}{h}\log|1+bK| = \tfrac{1}{h}\log(1+h/K) + \tfrac{1}{h}\log|1-\Delta b/y|$
and Taylor-expand $\rho$ around $b_0$ as $\rho(y+b_0) = \rho(b_0) + \rho'(b_0)y + \mathcal{R}(y)$ with $\mathcal{R}(y) = o(|y|)$.

\emph{Constant term.} The first part of the integrand contributes $2\delta\rho(b_0)h^{-1}\log(1+h/K) \to 2\delta\rho(b_0)/K$. The logarithmic part contributes $\rho(b_0)h^{-1}I_0(\Delta b)$ with $I_0(c) := \int_{-\delta}^\delta\log|1-c/y|\,dy = (\delta-c)\log(\delta-c) + (\delta+c)\log(\delta+c) - 2\delta\log\delta$. Differentiating gives $I_0'(0) = \log 1 = 0$, so by the chain rule and $\Delta b/h\to 1/K^2$, this term vanishes in the limit. The total constant-term contribution is $2\delta\rho(b_0)/K$, which matches the PV identity $\mathrm{PV}\int_{-\delta}^\delta\rho(b_0)(b_0+y)/(Ky)\,dy = 2\delta\rho(b_0)/K$.

\emph{Linear term.} The first part is odd, integrating to zero. The logarithmic part contributes $\rho'(b_0)h^{-1}I_1(\Delta b)$ with $I_1(c) := \int_{-\delta}^\delta y\log|1-c/y|\,dy = \tfrac{1}{2}(\delta^2-c^2)\log((\delta-c)/(\delta+c)) - c\delta$. Differentiating gives $I_1'(0) = -2\delta$, so $h^{-1}I_1(\Delta b)\to -2\delta/K^2$. The total linear-term contribution is $-2\delta\rho'(b_0)/K^2$, which matches $\mathrm{PV}\int_{-\delta}^\delta\rho'(b_0)y(b_0+y)/(Ky)\,dy = 2\delta b_0\rho'(b_0)/K = -2\delta\rho'(b_0)/K^2$ (using $b_0 = -1/K$).

\emph{Remainder.} $\mathcal{R}(y)/y$ is continuous at $0$; by dominated convergence the remainder converges to its non-singular integral.

Summing, the limit $h^{-1}(J(K+h) - J(K))$ exists and equals the PV integral.
\end{proof}

\section{Preliminary phase: driving compact-subset initialization into the basin}\label{sec:SM-prelim}

The rate of \cref{thm:sample-complexity-density-known,thm:total-sample-complexity-pl-basin} assumes $K^{(0)} \in \mathcal{N}_\delta$. We close this gap by exhibiting a preliminary phase that drives an arbitrary $K^{(0)} \in \mathcal{K}_{\rm stab}$ into $\mathcal{N}_\delta$ in a number of samples logarithmic in problem parameters and independent of $\eta$.

\begin{proposition}[preliminary phase, sample cost independent of $\eta$]
\label{prop:preliminary-phase-sample-cost-eta}
Suppose $K^{(0)}$ lies in a compact subset $\mathcal{K}_0 \subset \mathcal{K}_{\rm stab}$ on which $J_1(K) := \mathbb{E}[\tfrac{1}{2}\log((1+BK)^2 + 1)]$ is uniformly strictly convex with PL constant $\mu_1 > 0$ and Lipschitz constant $L_1 < \infty$. Run mini-batch SGD on $J_1$ with mini-batch size $N_0$ and step size $1/L_1$ until the iterate enters $\mathcal{N}_\delta$. Then there exists $N_0 = N_0(\rho, \mathcal{K}_0, \delta)$ and $T_0 = O((L_1/\mu_1)\log(|\mathcal{K}_0|/\delta))$ such that the preliminary phase terminates in at most $T_0$ iterations with total sample cost $T_0 N_0$, independent of $\eta$.
\end{proposition}
\begin{proof}
The fixed regularization $\varepsilon = 1$ makes $J_1$ a $C^\infty$ function on $\mathcal{K}_{\rm stab}$ with bounded gradient and Hessian uniformly on the compact subset $\mathcal{K}_0$. Strict convexity of $J_1$ on $\mathcal{K}_0$ does not hold in general on the full $\mathcal{K}_{\rm stab}$ --- the Hessian degenerates at the channel-edge gains $\{-1/b_{\min}, -1/b_{\max}\}$ --- but on any compact subset $\mathcal{K}_0$ bounded away from these gains, joint continuity of $d^2 J_\varepsilon/dK^2$ in $(K, \varepsilon)$ on $\mathcal{K}_0 \times [0, 1]$ (\cref{thm:hessian-decomposition}) and strict positivity at the interior minimum of $J_1$ on $\mathcal{K}_0$ yield a uniform PL constant $\mu_1 > 0$ depending only on $\rho$ and the distance from $\mathcal{K}_0$ to the channel-edge gains. The single-sample gradient estimator $\psi(B;K,1) = B(1+BK)/((1+BK)^2 + 1)$ is uniformly bounded on $\mathcal{K}_0 \times \mathrm{supp}(\rho)$ (no parity cancellation needed at $\varepsilon = 1$, since the cusp is absent), so its variance is bounded by an absolute constant $\sigma_1^2$ depending only on $\mathcal{K}_0$ and $\rho$. Standard PL-SGD analysis on $J_1$ gives geometric contraction of $J_1(K^{(t)}) - \min_{\mathcal{K}_0} J_1$ at rate $1 - \mu_1/L_1$; the iterate enters $\mathcal{N}_\delta$ once $J_1(K^{(t)}) - \min_{\mathcal{N}_\delta} J_1 \le \mu_1\delta^2/2$. With $N_0 = \lceil\sigma_1^2/(\mu_1^2\delta^4)\rceil$, this happens in $T_0 = O((L_1/\mu_1)\log(|\mathcal{K}_0|/\delta))$ iterations.
\end{proof}

\emph{Boundary behavior.} Writing $d_0 := \mathrm{dist}(\mathcal{K}_0, \{-1/b_{\min}, -1/b_{\max}\})$, the Hessian formula \eqref{eq:hess_Kstar} gives $\mu_1 = \Omega(d_0^2)$ and $L_1 = O(d_0^{-2})$ uniformly on $\mathcal{K}_0$, so the preliminary-phase cost is $T_0 N_0 = O(d_0^{-4}\log(|\mathcal{K}_0|/\delta))$, an $\eta$-independent constant.

\section{Full derivation of the bias formula (\texorpdfstring{\cref{prop:properties-J-epsilon}(c)}{Proposition 3.2(c)})}\label{sec:SM-bias-detail}

We give the full derivation of $J_\varepsilon^* - J^* = \pi\rho(b_\star)\varepsilon/|K^*| + O(\varepsilon^2)$, expanding the sketch in \cref{b1-proof-of-proposition-32c-bias-at-the-optimum}. The bias decomposes as $J_\varepsilon^* - J^* = [J_\varepsilon(K^*) - J(K^*)] - [J_\varepsilon(K^*) - J_\varepsilon(K_\varepsilon^*)]$.

\emph{Step 1: bias at $K^*$.} The integrand $\tfrac{1}{2}\log((1+bK^*)^2 + \varepsilon^2) - \log|1+bK^*| = \tfrac{1}{2}\log(1 + \varepsilon^2/v(b)^2)$ with $v(b) = 1 + bK^*$. Substitute $s = b - b_\star$ and split at $|s| \le R$ (near) and $|s| > R$ (far):
\emph{Far region:} On $|s| > R$, $|v| \ge |K^*|R$, so $\tfrac{1}{2}\log(1 + \varepsilon^2/v^2) = \tfrac{\varepsilon^2}{2v^2} + O(\varepsilon^4)$, contributing $F_\varepsilon = O(\varepsilon^2)$.
\emph{Near region:} The substitution $t = K^*s/\varepsilon$, $ds = \varepsilon/|K^*|\,dt$, $T = |K^*|R/\varepsilon$ gives
\begin{equation}\label{eq:near_region_integral_SM}
N_\varepsilon = \frac{\varepsilon}{|K^*|}\int_{-T}^{T}\rho(b_\star + \varepsilon t/K^*)\cdot\tfrac{1}{2}\log(1 + 1/t^2)\,dt.
\end{equation}
The Frullani-type identity $I_\infty := \int_{-\infty}^\infty\tfrac{1}{2}\log(1+1/t^2)\,dt = \pi$ (provable via $\int_0^\infty\log(1+1/t^2)\,dt = \pi$ by contour or by $\log(1+1/t^2) = \log(1+it) + \log(1-it) - 2\log|t|$), together with the tail estimate $\int_{|t|>T}\tfrac{1}{2}\log(1+1/t^2)\,dt = 1/T + O(T^{-3}) = O(\varepsilon)$ (from $\log(1+1/t^2) = 1/t^2 + O(1/t^4)$ for large $|t|$), gives $\int_{-T}^T = \pi - O(\varepsilon)$. Taylor-expanding $\rho(b_\star + \varepsilon t/K^*) = \rho(b_\star) + \tfrac{\varepsilon t}{K^*}\rho'(b_\star) + \tfrac{1}{2}(\tfrac{\varepsilon t}{K^*})^2\rho''(\xi_t)$ and integrating: the constant term contributes $\pi\rho(b_\star)\varepsilon/|K^*| + O(\varepsilon^2)$; the linear-in-$t$ term is odd over $[-T,T]$ and integrates to zero; the quadratic remainder contributes $O(\varepsilon^2)$ via $\int_{-T}^T t^2\log(1+1/t^2)\,dt = 2T + O(1) = O(1/\varepsilon)$ times $\varepsilon^3/|K^*|^3 = O(\varepsilon^2)$. Summing,
\begin{equation}
J_\varepsilon(K^*) - J(K^*) = \pi\rho(b_\star)\varepsilon/|K^*| + O(\varepsilon^2).
\end{equation}

\emph{Step 2: from $K^*$ to $K_\varepsilon^*$.} By \cref{lem:persistence-regularized-minimizer}(c), $0 \le J_\varepsilon(K^*) - J_\varepsilon(K_\varepsilon^*) \le L_0 C_K^2\varepsilon^2/2 = O(\varepsilon^2)$. Combining with Step 1, $J_\varepsilon^* - J^* = \pi\rho(b_\star)\varepsilon/|K^*| + O(\varepsilon^2)$.

\section{Full derivation of the variance remainder (\texorpdfstring{\cref{thm:variance-scaling-naive-estimator}}{Theorem 4.1})}\label{sec:SM-variance-remainder}

We expand the sketch in \cref{b2-remainder-in-the-variance-scaling-of-theorem-41}. From the proof of \cref{thm:variance-scaling-naive-estimator},
\begin{equation}
\mathbb{E}[\psi^2] = \int_{s_-}^{s_+}\rho(b_{\rm sing}+s)(b_{\rm sing}+s)^2\cdot\frac{K^2s^2}{(K^2s^2+\varepsilon^2)^2}\,ds.
\end{equation}
Substitute $t = Ks/\varepsilon$, $ds = \varepsilon\,dt/|K|$, $T_\pm = K s_\pm/\varepsilon$:
\begin{equation}
\mathbb{E}[\psi^2] = \frac{1}{\varepsilon|K|}\int_{T_-}^{T_+}f(t;\varepsilon)\cdot\frac{t^2}{(1+t^2)^2}\,dt,
\end{equation}
with $f(t;\varepsilon) = \rho(b_{\rm sing}+\varepsilon t/K)(b_{\rm sing}+\varepsilon t/K)^2$ and Taylor expansion $f(t;\varepsilon) = f_0 + \tfrac{\varepsilon t}{K}f_1 + \tfrac{1}{2}(\tfrac{\varepsilon t}{K})^2 f''(\xi_{t,\varepsilon})$, where $f_0 = \rho(b_{\rm sing})b_{\rm sing}^2 = \rho(b_{\rm sing})/K^2$ and $f_1 = \rho'(b_{\rm sing})b_{\rm sing}^2 + 2b_{\rm sing}\rho(b_{\rm sing})$.

\emph{Constant term.} Since $\int_{-\infty}^\infty t^2/(1+t^2)^2\,dt = \pi/2$ and the tail satisfies $\int_{|t|>T}t^2/(1+t^2)^2\,dt = 2/T + O(T^{-3})$ (from $t^2/(1+t^2)^2 = 1/t^2 + O(1/t^4)$), we have $\int_{T_-}^{T_+}t^2/(1+t^2)^2\,dt = \pi/2 - O(\varepsilon)$. Multiplied by $f_0/(\varepsilon|K|) = \rho(b_{\rm sing})/(K^2\varepsilon|K|)$, the leading term is $\pi\rho(b_{\rm sing})/(2|K|^3\varepsilon)$ with $O(1)$ correction.

\emph{Linear-in-$t$ term.} Contributes $(f_1/(K|K|))\int_{T_-}^{T_+}t^3/(1+t^2)^2\,dt$. The antiderivative is $\tfrac{1}{2}[\log(1+t^2) + 1/(1+t^2)]$, giving
\begin{equation}
\int_{T_-}^{T_+}\frac{t^3}{(1+t^2)^2}\,dt = \tfrac{1}{2}\log\frac{1+T_+^2}{1+T_-^2} + O(\varepsilon^2)\xrightarrow{\varepsilon\downarrow 0}\log|s_+/s_-|.
\end{equation}
The contribution is $\mathrm{sgn}(K)\,f_1/K^2\cdot\log|s_+/s_-| + O(\varepsilon)$ (using $f_1/(K|K|) = \mathrm{sgn}(K)\,f_1/K^2$).

\emph{Quadratic remainder.} Bounded by $\tfrac{1}{2}(\varepsilon/|K|)^2\|f''\|_\infty\cdot t^2\cdot t^2/(1+t^2)^2$ pointwise. Since $t^4/(1+t^2)^2 \le 1$, $\int_{T_-}^{T_+}t^4/(1+t^2)^2\,dt = O(|T|) = O(1/\varepsilon)$. Multiplied by $\varepsilon^2$ and divided by $\varepsilon|K|$, the contribution is $O(\varepsilon)$.

Summing the three contributions and subtracting $\mathbb{E}[\psi]^2 = O(1)$, the variance is
\begin{equation}
\mathrm{Var}[\psi] = \frac{\pi\rho(b_{\rm sing}(K))}{2|K|^3\varepsilon} + \mathrm{sgn}(K)\frac{f_1(K)}{K^2}\log|s_+/s_-| + O(1) + o(1).
\end{equation}

\section{Numerical stability of the deterministic PV alternative}
\label{sec:SM-quadrature}

The rate parity established in \cref{prop:plug-and-solve-rate} between
policy gradient and plug-and-solve under density-unknown access raises the
question whether the policy-gradient framework offers any practical
advantage. As noted in
section~\ref{56-comparison-and-structural-perspective}, the answer is
quadrature stability: every Newton step of the deterministic
alternative requires the evaluation of two integrals against a Cauchy
kernel with a moving pole, and these evaluations are not robust to
off-the-shelf adaptive quadrature. We document this here.

\begin{figure}[!t]
\centering
\includegraphics[width=\columnwidth]{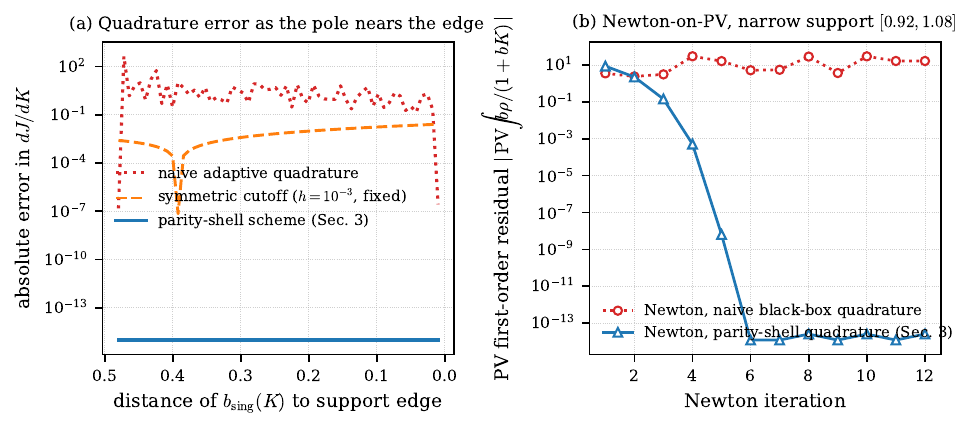}
\caption{Numerical-stability ablation for the deterministic
principal-value alternative
(section~\ref{56-comparison-and-structural-perspective}). \emph{(a)}~Absolute
error in $dJ/dK(K)$ as the moving pole $b_{\rm sing}(K) = -1/K$ sweeps
across the interior of the support $[0.5, 1.5]$ for D2: off-the-shelf
adaptive QUADPACK without pole-aware nodes returns errors of order
$10^1$--$10^2$ uniformly; a naive symmetric-cutoff scheme with a fixed
tolerance $h = 10^{-3}$ returns errors of order $10^{-2}$--$10^{-4}$;
the parity-shell scheme of section~3 (\cref{thm:hessian-decomposition})
returns errors at machine precision ($\sim\!10^{-16}$).
\emph{(b)}~Newton's method on the principal-value first-order
condition for a narrow-support density on $[0.92, 1.08]$: with naive
black-box quadrature for the gradient and a finite-difference Hessian
of the same unreliable gradient, Newton diverges (residual fixed near
$10^{1}$ for twelve iterations); with parity-shell quadrature for both
the gradient and the Hadamard finite-part Hessian, Newton converges at
the expected quadratic rate, reaching residual $2.6 \cdot 10^{-14}$ by
iteration~6.}
\label{fig:SM-quadrature}
\end{figure}

\paragraph{Pole-sweep error (\cref{fig:SM-quadrature}(a))}
We sweep $K$ so that $b_{\rm sing}(K)$ marches from mid-support toward
the right endpoint $b_{\max} = 1.5$ for D2 and record the absolute
error of three schemes in evaluating $dJ/dK$ against the parity-shell
reference: (i)~black-box adaptive QUADPACK (no pole information
passed), (ii)~symmetric cutoff with fixed tolerance $h = 10^{-3}$,
(iii)~parity-shell quadrature with the subtraction identity of section~3.
The maxima over the sweep are $4.3\!\cdot\!10^{2}$, $2.6\!\cdot\!10^{-2}$,
and $1.0\!\cdot\!10^{-16}$ respectively. The naive adaptive scheme
fails by twelve to fifteen orders of magnitude because the adaptive
node refinement of QUADPACK cannot resolve the non-integrable simple
pole; the symmetric-cutoff scheme is sensitive to the cutoff tolerance,
with errors that degrade as the pole nears the support endpoint
because the asymmetry term $\log|s_+/s_-|$ stops being negligible.

\paragraph{Newton-on-PV (\cref{fig:SM-quadrature}(b))}
We close the loop by running Newton's method on the plug-and-solve
principal-value first-order condition with each quadrature scheme. We
use a narrow-support uniform on $[0.92, 1.08]$ to expose the
conditioning issue: the singularity is interior and the domain is
short, leaving the off-the-shelf solver few quadrature nodes to spread
around the moving pole. With naive black-box quadrature for both
$dJ/dK$ and the finite-difference Hessian, the Newton residual
fluctuates near $10^{1}$ for twelve iterations without progress. With
parity-shell quadrature for $dJ/dK$ and the closed-form Hadamard
finite-part Hessian of \eqref{eq:hess_Kstar}, Newton converges at the
expected quadratic rate, reaching residual $2.6 \cdot 10^{-14}$ by
iteration~$6$ and machine-precision floor thereafter.

\paragraph{Operational implication}
The asymptotic $O(\log\log(1/\eta))$ Newton rate of
section~\ref{56-comparison-and-structural-perspective}(i) is realizable in
practice if and only if the underlying quadrature respects the
parity-cancellation structure of section~3; using an off-the-shelf solver
either re-implements the parity-shell decomposition in disguise, or
fails. The regularized smooth-landscape paradigm of section~3 sidesteps the
issue entirely (the integrand of $dJ_\varepsilon/dK$ is smooth for
every $\varepsilon > 0$), which is the operational form of the
structural advantage of the regularized paradigm (section~\ref{56-comparison-and-structural-perspective}, ``Why this regularization'').

\section{Auxiliary validations of objects from sections~2 and~3}\label{sec:SM-aux}

\paragraph{Kink handling for D4}
For the triangular density D4 with apex at $b = 1$, the integrand
$g(b) = 2b\rho(b) + b^2\rho'(b)$ entering the Hadamard finite-part
Hessian via \eqref{eq:hess_Kstar} has a jump discontinuity at the
kink: $g(1^-) = 4$, $g(1^+) = 0$. Passing the kink location to the
quadrature integrator as an explicit break-point yields the reported
$H_4 = 12.96$; a kink-unaware integrator aliases across the
discontinuity and returns $H_4 \approx 25.9$, an over-estimate by a
factor close to two. We verified the parity-shell value against a
direct central-second-difference of $J$ at $K^*$ with step
$h = 2\!\cdot\!10^{-4}$, obtaining agreement to $10^{-3}$ in relative
terms.

\paragraph{Basin geometry}
The basin-size constant $\delta$ of \cref{lem:uniform-PL-constant} is limited by the pole-to-edge margin
\begin{equation*}
\tau(K) := \min(b_{\rm sing}(K)-b_{\min},\,b_{\max}-b_{\rm sing}(K)).
\end{equation*}
For D2 with $K^* = -0.928$, $\tau(K^*) = 0.422$. We swept $\delta\in[0.05, 0.25]$ and computed
$\mu_0(\delta) := \min_{|K-K^*|\le\delta} d^2 J/dK^2(K)$ and
$L_0(\delta) := \max_{|K-K^*|\le\delta} d^2 J/dK^2(K)$
via \eqref{eq:hess_Kstar}; the condition number $L_0/\mu_0$ rises from $1.5$ at $\delta=0.05$ to $7.4$ at $\delta=0.20$, and $\mu_0$ remains bounded away from zero throughout, consistent with \cref{lem:uniform-PL-constant}.

\paragraph{Necessity of \cref{asm:regularity} (D4 robustness)}
D4 violates \cref{asm:regularity} (the density is $C^0$ but not
$C^2$), yet the rate-level conclusions of \cref{6-numerical-validation} hold for it: the variance
slope is $-1.048$ versus theoretical $-1$; the density-known sample
complexity $N(\eta)$ has slope $-0.88$ versus theoretical $-1$. The
constants are consistent with kink-aware quadrature of the Hadamard
finite-part Hessian, suggesting that \cref{asm:regularity} can be
relaxed to piecewise-$C^2$ provided the quadrature respects the kink
locations. A precise statement of this generalization is outside the
scope of the present paper (cf.~limitation L4 in
\cref{72-limitations}).

\paragraph{Tail averaging vs.\ raw iterate}
The \cref{fig:sgd-rate} trajectories use Polyak--Ruppert tail
averaging of the last half of iterates. On D2 we recomputed the
empirical slope (a) using the raw iterate and obtained $-0.83$
(vs.~$-0.94$ with tail averaging), confirming that tail averaging
removes the geometric transient and exposes the asymptotic rate of
\cref{thm:sample-complexity-density-known} more cleanly. The total sample
cost is unaffected; the gap at any fixed $n$ is reduced by roughly a
factor of $1.5$ on average.

\section*{Reproducibility}
All experiments were run in Python 3.11 with NumPy 1.26, SciPy 1.11, and Matplotlib 3.8. Seeds $1,\ldots,n_{\rm seed}$ are deterministic; each figure corresponds to a specific entry point of the accompanying script (\texttt{exp1} through \texttt{exp4}). The script and the processed data are provided in the supplementary archive.